\documentclass[12pt]{article}
\usepackage[affil-it]{authblk}
\usepackage[top=1in,bottom=1in,right=1in,left=1in]{geometry}
\usepackage{amssymb,amsmath,amsthm}
\usepackage{cleveref}
\usepackage{tikz}
\usepackage{pgfplots}
\usepackage{pgfplotstable}
\usepackage{arydshln}
\usepackage{verbatim}
\usepackage{titlesec}
\usepackage{subfigure}

\titleclass{\subsubsubsection}{straight}[\subsection]

\newcounter{subsubsubsection}[subsubsection]
\renewcommand\thesubsubsubsection{\thesubsubsection.\arabic{subsubsubsection}}
\titleformat{\subsubsubsection}
  {\normalfont\normalsize\bfseries}{\thesubsubsubsection}{1em}{}
  \titlespacing*{\subsubsubsection}
  {0pt}{3.25ex plus 1ex minus .2ex}{1.5ex plus .2ex}
\newcommand{\boxit}[1]{#1}
\newcommand{\dashboxit}[1]{#1}

\providecommand{\keywords}[1]{\textbf{\textit{Keywords: }} #1}

\crefname{equation}{}{}
\crefrangeformat{equation}{(#3#1#4) --~(#5#2#6)}
\crefmultiformat{equation}{(#2#1#3)}{, (#2#1#3)}{, (#2#1#3)}{, (#2#1#3)}
\crefname{lem}{Lemma}{Lemmas}
\crefname{section}{Section}{Sections}
\crefname{subsubsubsection}{Section}{Sections}
\crefname{rem}{Remark}{Remarks}
\crefname{figure}{Figure}{Figures}
\crefname{table}{Table}{Tables}
\Crefname{lem}{Lemma}{Lemmas}
\crefname{thm}{Theorem}{Theorems}
\Crefname{thm}{Theorem}{Theorems}

\newcommand{\sgn}{\mathop{\mathrm{sgn}}}
\newcommand{\sinc}{\mathrm{sinc}}

\newcommand{\ders}[2]{\frac{d #1}{d #2}}
\newcommand{\pders}[2]{\frac{\partial #1}{\partial #2}}
\newcommand{\al}[2]{\alpha^{#1}_{#2}}
\newcommand{\bet}[2]{\beta^{#1}_{#2}}
\newcommand\bR{\mathbf R}
\newcommand{\pd}{\partial}
\newcommand{\bx}{{\boldsymbol x}}

\newcommand{\tw}{\tilde{w}}

\newcommand{\btau}{{\boldsymbol \tau}}
\newcommand{\ajo}{\alpha^{j}_{1}}
\newcommand{\ajt}{\alpha^{j}_{2}}
\newcommand{\bjo}{\beta^{j}_{1}}
\newcommand{\bjt}{\beta^{j}_{2}}

\newcommand{\cN}{{\mathcal N}}

\newcommand{\bK}{{\mathbf K}}
\newcommand{\bF}{{\mathbf F}}
\newcommand{\bHp}{\mathbb{H}^{+}}
\newcommand{\bHm}{\mathbb{H}^{-}}
\newcommand{\obHm}{\overline{\mathbb{H}}^{-}}
\newcommand{\obHp}{\overline{\mathbb{H}}^{+}}
\newcommand{\bT}{{\mathbf T}}
\newcommand{\bI}{{\mathbf I}}

\newcommand{\cC}{{\mathcal C}}

\newcommand{\cD}{{\mathcal D}}
\newcommand{\R}{{\mathbb R}}
\newcommand{\bL}{{\mathbb L}}

\newcommand{\N}{{\mathbb N}}
\newcommand{\C}{{\mathbb C}}

\newcommand{\bbmat}{\begin{bmatrix}}
\newcommand{\ebmat}{\end{bmatrix}}
\newcommand{\beq}{\begin{equation}}
\newcommand{\eeq}{\end{equation}}

\newcommand{\eps}{\varepsilon}
\newcommand{\Gammakey}{\Gamma_{\text{key}}}

\newcommand{\bmu}{{\boldsymbol \mu}}

\newcommand{\bgamma}{{\boldsymbol \gamma}}

\newcommand{\cost}{\cos{(\pi\theta)}}
\newcommand{\sint}{\sin{(\pi\theta)}}
\newcommand{\sinzt}{\sin{(\pi z \theta)}}
\newcommand{\sinztc}{\sin{(\pi z (2-\theta))}}

\newcommand{\sinz}{\sin{(\pi z)}}

\newcommand{\sqsint}{\sin^2{(\pi\theta)}}
\newcommand{\sqsinz}{\sin^2{(\pi z)}}
\newcommand{\sqcost}{\cos^2{(\pi\theta)}}
\newcommand{\cubsint}{\sin^3{(\pi\theta)}}

\newcommand{\bA}{{\mathbf{A}}}

\newcommand{\by}{{\boldsymbol y}}

\newcommand{\bn}{{\boldsymbol \nu}}

\newcommand{\bu}{{\boldsymbol u}}
\newcommand{\bucomp}{{\boldsymbol u}_{\text{comp}}}
\newcommand{\buex}{{\boldsymbol u}_{\text{exact}}}

\newcommand{\bt}{{\boldsymbol t}}
\newcommand{\bB}{{\mathbf B}}

\newcommand{\bh}{{\boldsymbol h}}

\newtheorem{thm}{Theorem}

\newtheorem{rem}[thm]{Remark}
\newtheorem{lem}[thm]{Lemma}

\newtheorem{obsv}[thm]{Observation}
\theoremstyle{definition}

\def\qed{\hfill$\blacksquare$\\} \renewenvironment{proof}{\noindent {\bf Proof.}}{\qed}

\title{On the solution of Stokes equation on regions with corners}

\author{M.~Rachh \thanks{Email address: \texttt{manas.rachh@yale.edu}}}
\affil{{\small Applied Mathematics Program, Yale University,
New Haven, CT 06511}}
\author{K.~Serkh \thanks{Email address: \texttt{kserkh@cims.nyu.edu}}}
\affil{{\small Courant Institute of Mathematical Sciences,
New York University, New York, NY 10012}}
\date{}

\begin{document}
\maketitle
\begin{abstract}
In Stokes flow, the stream function associated with the velocity
of the fluid satisfies the biharmonic equation.
The detailed behavior of solutions to the biharmonic equation
on regions with corners has been historically difficult to characterize.
The problem was first examined by Lord Rayleigh in 1920;
in 1973, the existence of infinite oscillations
in the domain Green's function
was proven in the case of the right angle by S.~Osher.
In this paper, we observe that, when the biharmonic equation is formulated
as a boundary integral equation, the solutions are representable
by rapidly convergent series of the form 
$\sum_{j} ( c_{j} t^{\mu_{j}} \sin{(\beta_{j} \log{(t)})}
+ d_{j} t^{\mu_{j}} \cos{(\beta_{j} \log{(t)})} )$, 
where $t$ is the distance from the corner and the parameters $\mu_{j},\beta_{j}$
are real, and are determined via an explicit formula 
depending on the angle at the corner.
In addition to being analytically perspicuous, 
these representations lend themselves to the construction
of highly accurate and efficient numerical discretizations,
significantly reducing the number of degrees of freedom required
for the solution of the corresponding integral equations. 
The results are illustrated by several numerical examples.
\end{abstract}
\keywords{
Integral equations, Stokes flow, Polygonal domains, Biharmonic equation, 
Corners}

%%%%%%%%%%%%%%%%%%%%%%%%%%%%%%%%%%%%%%%%%%%%%%%%%%%%%%%%%%%%%%%%%%%%%%

\section{Introduction}
In classical potential theory, solutions to elliptic partial differential equations
are represented by potentials on the boundaries of the regions. 
Over the last four decades, several well-conditioned boundary 
integral representations have been developed for the biharmonic equation 
with various boundary conditions.
These integral representations have been studied extensively when the boundaries
of the regions are approximated by a smooth curve  
(see,~\cite{greengard1996integral,power1993completed,power1987second,michlin1957integral,pozrikidis1992boundary,Farkas89,Jiang11,askhamrachh}, 
for example).
In all of these representations, the kernels of the integral
equation are at worst weakly singular and, in same cases, even smooth.
The corresponding solutions 
to the integral equations tend to be as smooth as the boundaries 
of the domains and the incoming data.

However, when the boundary of the region has corners, 
the solutions to both the differential equation and the corresponding 
integral equations are known to develop singularities.
The solutions to the differential equation have been studied extensively
on regions with corners
for both Dirichlet and gradient boundary conditions. 
In particular, the behavior of
solutions to the biharmonic equation on a wedge
enclosed by
straight boundaries $\theta=0$ and $\theta=\alpha$, and a circular arc $r=r_{0}$, 
where $(r,\theta)$ are polar coordinates, has received much attention 
over the years.
To the best of our knowledge, this particular problem 
was first studied by Lord Rayleigh
in 1920~\cite{rayleigh}, 
and over the decades, was studied by A. Dixon~\cite{dixon}, Dean and Montagnon~\cite{deanandm}, 
Szeg\"{o}~\cite{szego}, Moffat~\cite{moffat}, Williams~\cite{williams1952stress}, 
and Seif~\cite{seif1973green}, to name a few.
In 1973, S. Osher showed that the Green's function
for the biharmonic equation on a right angle wedge has infinitely
many oscillations in the vicinity of the corner on all but 
finitely many rays~\cite{osher}.
The more complicated structure of the Green's function for the biharmonic
equation is explained, in part, by the fact that the biharmonic equation
does not have a maximum principle associated with it, 
while, for example, Laplace's equation does.

While much has been published about the solutions to the differential
equation, the solutions of the corresponding integral equation
have been studied much less exhaustively.
Recently, a detailed analysis of solutions 
to integral equations corresponding to the Laplace and Helmholtz equations
on regions with corners was carried out 
by the second author and V. Rokhlin~\cite{serkhcorner1,serkhcorner2,serkhcorner3}.
They observed that the solutions to these integral equations can be expressed
as rapidly convergent series of singular powers in the Laplace case 
and Bessel functions of non-integer order in the Helmholtz case.

In this paper, we investigate the solutions to a standard 
integral equation corresponding to the velocity boundary value problem 
for Stokes equation.
For the velocity boundary value problem, 
the stream function associated with the velocity field
satisfies the biharmonic equation with gradient boundary conditions
(it turns out that the same integral equation can be used 
for Dirichlet boundary conditions, see, for example,~\cite{askhamrachh}).
We show that, if the boundary data is smooth on each side of the corner,
then the solutions of this integral equation 
can be expressed a rapidly convergent series of elementary functions
of the form $t^{\mu_{j}} \cos{(\beta_{j} \log{|t|})}$ and
$t^{\mu_{j}} \sin{(\beta_{j} \log{|t|})}$,
where the parameters $\mu_{j},\beta_{j}$ can be computed explicitly
by a simple formula depending only on the angle at the corner.
Furthermore, we prove that, for any $N$, 
there exists a linear combination of the first $N$ of these 
basis functions which satisfies the integral equation 
with error $O(|t|^{N})$, where
$t$ is the distance from the corner.

The detailed information about the analytical behavior of the solution
in the vicinity of corners, discussed in this paper, allows for the 
construction of
purpose-made discretizations of the integral equation.
These discretizations accurately represent the solutions near corners using 
far fewer degrees of freedom than graded meshes, 
which are commonly used in such environments, thereby leading to highly
efficient numerical solvers for the integral equation. 
For an alternative treatment of the 
differential equation see, for example,~\cite{wilkening2002mathematical,sukumar,costabel}.

The rest of the paper is organized as follows. 
In~\cref{sec:prelim},
we discuss the mathematical preliminaries for the governing equation
and it's reformulation as an integral equation. 
In~\cref{sec:analapp},
we derive several analytical results using techniques required for the
principal results which are derived in~\cref{sec:inteq}.
We illustrate the performance of the numerical scheme which
utilizes the explicit knowledge of the structure of the
solution to the integral equation in~\cref{sec:numres}.
In~\cref{sec:concl}, we present generalizations and extensions
of the apparatus of this paper.
The proof of several results in~\cref{sec:inteq} are technical
and are presented in~\cref{sec:appa,sec:appb}.

%%%%%%%%%%%%%%%%%%%%%%%%%%%%%%%%%%%%%%%%%%%%%%%%%%
\section{Preliminaries} \label{sec:prelim}
%\subsection{Some notes on notation}
In this paper, vector-valued
quantities are denoted by bold, lower-case letters
(e.g. $\bh$), while tensor-valued quantities are bold
and upper-case (e.g. $\mathbf{T}$). 
Subscript indices of non-bold characters (e.g. $h_j$ or $T_{jkl}$)
are used to denote the entries within a vector ($\bh$) or tensor ($\bT$).
We use the standard Einstein summation convention; in other words, 
there is an implied sum taken over the repeated indices of 
any term (e.g. the symbol $a_{j} b_{j}$ is used to represent the sum
$\sum_{j} a_{j} b_{j}$).
Let $\cC^{k}$ denote the space of functions which have $k$ continuous
derivatives.

Suppose now that $\Omega$ is a simply connected open subset of $\R^{2}$. 
Let $\Gamma$ denote the boundary of $\Omega$ and
suppose that $\Gamma$ is a simple closed curve of length $L$ with $n_{c}$ corners. 
Let $\bgamma:[0,L]\to \R^{2}$ denote an arc length parameterization of 
$\Gamma$ in the counter-clockwise direction, and suppose that the location of the corners are given 
by $\bgamma(s_{j})$, $j=1,2,\ldots n_{c}$ with 
$0=s_{1}<s_{2} \ldots<s_{n_{c}}<s_{n_{c}+1} = L$.
We assume that the corners all have finite angles, i.e., the region $\Omega$
does not have any cusps.
Furthermore, suppose that $\bgamma$ is analytic on the intervals $(s_{j},s_{j+1})$
for each $j=1,2,\ldots n_{c}$.
Let $\btau(\bx)$ and $\bn(\bx)$ denote the positively-oriented 
unit tangent and the outward unit normal respectively, for $\bx \in \Gamma$.
Let $h_{\tau} = h_{j} \tau_{j}$ and $h_{\nu} = h_{j} \nu_{j}$ denote
the tangential and normal components of the vector $\bh$ 
respectively, see~\cref{fig:boundary}.

\begin{figure}[h!]
\begin{center}
\includegraphics[width=6cm]{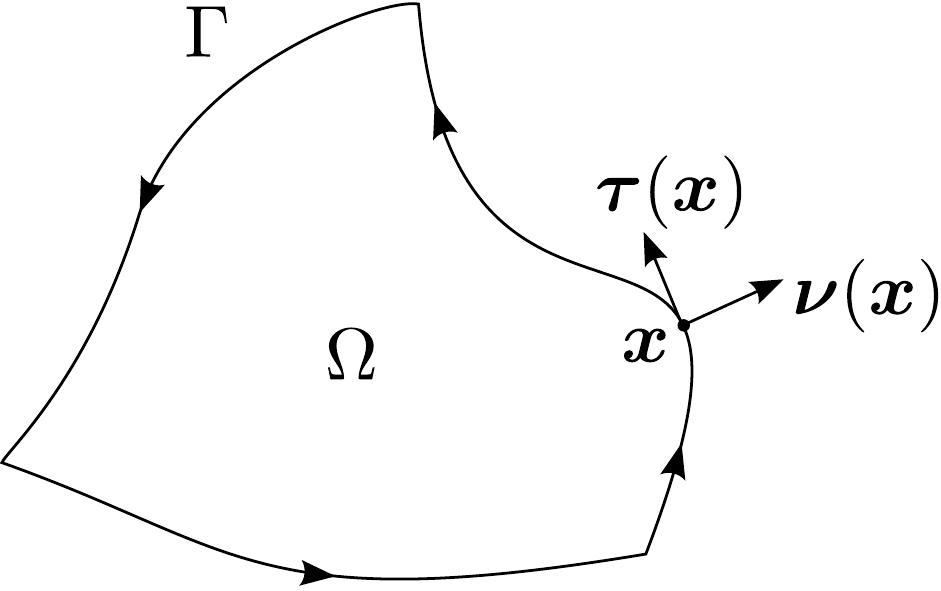}
\caption{A sample domain $\Omega$ with three corners}
\label{fig:boundary}
\end{center}
\end{figure}

\begin{rem}
In this section, we will be concerned with regions of the form
$\Omega$ described above (see~\cref{fig:boundary}).
\end{rem}

\subsection{Velocity boundary value problem}
The equations of incompressible Stokes flow with velocity boundary conditions
on a domain $\Omega$ with boundary $\Gamma$ are
\begin{align}
-\Delta\bu+\nabla p & =0\quad\mbox{in }\Omega,
\label{eq:StokesFlowEq} \\
\nabla\cdot\bu & =0\quad\mbox{in }\Omega,
\label{eq:MassConservation}\\
\bu &= \bh \quad\text{on }\Gamma, 
\label{eq:bcStokesFlow}
\end{align}
where $\bu$ is the velocity of the fluid, $p$ is the fluid pressure and $\bh$ is 
the prescribed velocity on the boundary.
For any $\bh$ which satisfies
\dashboxit{
\beq
\int_{\Gamma} \bh \cdot \bn \, dS = 0 \, ,
\eeq
}
there exists a unique velocity field $\bu$ and a pressure $p$, 
defined uniquely up to a constant, that 
satisfy~\cref{eq:StokesFlowEq,eq:MassConservation,eq:bcStokesFlow}. 
We summarize the result in the 
following lemma (see~\cite{bbuniqueness} for a proof).

\begin{lem}
Suppose $\bh \in \bL^{2}( \Gamma)$ 
and satisfies $\int_{\Gamma} \bh \cdot \bn \, dS = 0$ , then 
there exists a unique velocity $\bu$ 
and a pressure $p$ which is unique up to a constant 
which satisfy the velocity boundary value 
problem~\cref{eq:StokesFlowEq,eq:MassConservation,eq:bcStokesFlow}.
\end{lem}

\begin{rem}
The Stokes equation with velocity boundary conditions 
can be reformulated as a biharmonic equation with gradient
boundary conditions. 
First, we represent the velocity $\bu \colon \Omega \to \R^2$ 
as $\bu = \nabla^{\perp} w$, 
where $w \colon \Omega \to \R$ is the stream function associated with the velocity field and 
$\nabla^{\perp}$ is the operator given by
\beq
\nabla^{\perp} w = 
\bbmat
-\pders{w}{x_{2}} \\
\pders{w}{x_{1}}
\ebmat \, .
\eeq
Next, we observe that $\bu= \nabla^{\perp}w$ automatically 
satisfies the divergence free condition~\cref{eq:MassConservation}.
Finally, taking the dot product of $\nabla^{\perp}$ with~\cref{eq:StokesFlowEq},
we observe that $w$ satisfies the biharmonic equation with gradient boundary
conditions given by
\begin{align}
\Delta^2 w &= 0  \quad \text{in } \Omega \\
\nabla^{\perp} w &= \bh \quad \text{on } \Gamma \, .
\end{align}
\end{rem}

\subsection{Integral equation formulation}
Following the treatment of \cite{kim2005microhydrodynamics,pozrikidis1992boundary},
the fundamental solution to the Stokes equations (the Stokeslet)  
is given by
\begin{equation}
G_{j,k}\left(\bx,\by\right)=\frac{1}{4\pi}
\left[ -\log\left|\bx-\by\right| \delta_{ij} + \frac{\left( x_{j} - y_{j} \right) \left( x_{k} - y_{k}
 \right)} {\left| \bx - \by \right|^2}\right] \, , \quad j,k \in {1,2} \, ,
\end{equation}
for $\bx,\by \in \R^{2}$, and $\bx \neq \by$, where $\delta_{ij}$ 
is the Kronecker delta function.
The stress tensor $T_{j,k,\ell} \left(\bx,\by\right)$ 
associated with the Green's function, 
or the stresslet, is given by
\begin{equation}
T_{j,k,\ell} \left(\bx,\by\right) = 
-\frac{1}{\pi} \frac{\left( x_{j} - y_{j} \right) 
\left( x_{k} - y_{k} \right) 
\left( x_{\ell} - y_{\ell} \right)}{\left| \bx - \by \right|^{4}}
\quad j,k,\ell \in 1,2 \, ,
\end{equation} 
$\bx,\by \in \R^{2}$ and $\bx \neq \by$.
The stresslet $\bT$ is roughly analogous to a dipole in electrostatics.
The double layer Stokes potential is the velocity field due to a surface
density of stresslets $\bmu$ and is defined by
\begin{equation}
\left(\cD_{\Gamma}[\bmu](\bx) \right)_{j}
=\int_{\Gamma}
T_{k,j,\ell}
\left(\by,\bx \right) \mathbf{\mu}_{k} 
\left(\by\right) \nu_{\ell}(\by) \, dS_{\by} \, , \quad j,k,\ell \in 1,2 \, ,
\label{stokesdlpdef}
\end{equation}
for $\bx \in \R^{2}$.
Clearly, $\cD_{\Gamma}[\bmu](\bx)$ satisfies Stokes equation for $\bx \in \Omega$.

The following lemma describes the behavior of the double
layer Stokes potential as $\bx \to \bx_{0}$ where $\bx_{0} \in \Gamma$. 
\begin{lem}
\label{Stokesdlplemma}
Suppose that $\bmu \colon \Gamma \to \R^{2}$ and 
let $\mathcal{D}_{\Gamma}[\bmu]\left(\bx\right)$ denote a double layer Stokes
potential~\eqref{stokesdlpdef}. Then
$\mathcal{D}_{\Gamma}[\bmu]\left(\bx\right)$ satisfies the jump relation:
\begin{align}
\lim_{\substack{\bx\to\bx_{0}\\
\bx\in \Omega}} \mathcal{D}_{\Gamma}[\bmu](\bx) &=
-\frac{1}{2}\bmu(\bx_{0}) 
+ \text{p.v.}\int_{\Gamma} \bK_{0}(\bx_{0},\by) \bmu(\by) dS_{\by} \, , \label{eq:DLP0Stokes} 
\end{align}
where $\bx_{0},\by \in \Gamma$, and $\bK_{0}$ is given by 
\beq
\label{eq:kercart}
\bK_{0}(\bx,\by) = -\frac{1}{\pi} \frac{\left(\by - \bx \right) 
\cdot \bn(\by)}{\left|
\bx - \by \right|^4}
\begin{bmatrix} (y_{1} - x_{1})^2 & (y_{1} - x_{1})(y_{2}-x_{2}) \\
(y_{1}-x_{1})(y_{2}-x_{2}) & (y_{2}-x_{2})^2 \end{bmatrix} \, 
\eeq
for $\bx,\by \in \Gamma$ and $\text{p.v.}\int$ denotes a principal value integral.
\end{lem}

\begin{rem}
A different jump relation holds for the double layer Stokes potential, 
$\cD_{\Gamma}[\bmu](\bx)$ at the corner points of the boundary $\Gamma$. 
However, in the integral equation framework, point values of the density
at the corner points are irrelevant from the perspective of computing the
velocity $\cD_{\Gamma}[\bmu]$ in the region $\Omega$.
In the remainder of the paper, we ignore the point behavior of the density
$\bmu$ at the corners.
\end{rem}

The following lemma states that the kernel $\bK_{0}$ is smooth if the boundary
$\Gamma$ is smooth.
\begin{lem}
The kernel $\bK_{0}$ is $\cC^{k-2}$ 
if $\bgamma$ is $\cC^{k}$ with limiting values
\beq
\lim_{t\to s} \bK_{0}(\bgamma(t),\bgamma(s)) = -\frac{1}{\pi}\kappa(\bgamma(t))
\begin{bmatrix}
\gamma_{1}'(t)^2 & \gamma_{1}'(t) \gamma_{2}'(t) \\
\gamma_{1}'(t) \gamma_{2}'(t) & 
\gamma_{2}'(t)^2 \\
\end{bmatrix} \, ,
\eeq
where $\kappa(\bgamma(t))$ is the curvature at $\bgamma(t)$.
Furthermore, 
$\bK_{0}$ is analytic if $\bgamma$ is analytic. 
\end{lem}

The following theorem reduces the velocity boundary value
problem~\cref{eq:StokesFlowEq,eq:MassConservation,eq:bcStokesFlow} 
to an integral equation on the boundary by representing 
$\bu$ as double layer Stokes potential with unknown density 
$\bmu$, i.e.
\beq
\bu(\bx) = \cD_{\Gamma}[\bmu](\bx) \, \quad \bx \in \Omega.
\eeq
\begin{lem}
\label{lem:uniquenesscart}
Suppose $\bh \in \bL^{2}(\Gamma)$ and that $\int_{\Gamma} 
\bh \cdot \bn \, dS = 0$. 
Then there exists a unique solution $\bmu \in \bL^{2}(\Gamma)$ 
which satisfies 
\beq
-\frac{1}{2} \bmu(\bx)  + \text{p.v.}\int_{\Gamma} \bK_{0}(\bx,\by) 
\bmu(\by)\, dS_{\by}  = \bh(\bx)\, ,\quad \bx \in \Gamma \label{eq:inteqStokes}\, ,
\eeq
and $\int_{\Gamma} \bmu \cdot \bn \, dS = 0$.
Furthermore, $\bu(\bx) = \cD_{\Gamma}[\bmu](\bx)$ satisfies Stokes 
equations~\cref{eq:StokesFlowEq,eq:MassConservation}, along with
the boundary conditions $\bu(\bx) = \bh(\bx)$ for $\bx \in \Gamma$.
\end{lem}
\begin{proof}
See, for example,~\cite{pozrikidis1992boundary} for a proof.
\end{proof}

The following lemma extends~\cref{lem:uniquenesscart} 
to the case where the boundary $\Gamma$ is an open arc.
\begin{lem}
Suppose $\bh \in \bL^{2}(\Gamma)$ 
then there exists a unique solution $\bmu \in \bL^{2}(\Gamma)$ 
which satisfies 
\beq
-\frac{1}{2}\bmu(\bx)
+ \text{p.v.} \int_{\Gamma} \bK_{0}(\bx,\by)
\bmu(\by)
\, dS_{\by}
= 
\bh(\bx)
\quad
\bx \in \Gamma \, ,
\label{eq:inteqStokescart2} 
\eeq
where $\bK_{0}$ is defined by~\cref{eq:kercart}.
\end{lem}

\subsection{Integral equation in tangential and normal coordinates}
It turns out that it is convenient to represent both the 
velocity on the boundary $\bh$ and the solution of the integral
equation $\bmu$ in terms of their tangential and normal coordinates,
denoted by $\bh=(h_{\tau},h_{\nu})$ and $\bmu = (\mu_{\tau},\mu_{\nu})$ 
respectively, as opposed to their Cartesian coordinates. 
In this section, we discuss the representation 
in the tangential and normal coordinates
of the double layer Stokes potential, and the corresponding integral equation for the velocity
boundary value problem. 

Let $\bR(\bx)$ denote the unitary transformation that converts vectors
expressed in Cartesian coordinates to vectors expressed in tangential
and normal components, i.e.
\beq
\bR(\bx) = \begin{bmatrix}
\tau_{1} (\bx) & \tau_{2}(\bx) \\
\nu_{1}(\bx) & \nu_{2} (\bx) 
\end{bmatrix} \, \quad \bx \in \Gamma .
\eeq
Let $\bR^{*}(\bx)$ denote the adjoint of $\bR(\bx)$.
Suppose $\bu(\bx)=(u_{1}(\bx),u_{2}(\bx))$ is the 
double layer Stokes potential 
with density $\bmu(\bx)=(\mu_{\tau}(\bx),\mu_{\nu}(\bx))$, given by
\beq
\label{eq:stokesDLPtn}
\bu(\bx) = \tilde{\cD}_{\Gamma}[\bmu](\bx) \, \quad \bx \in \Omega,
\eeq
where
\beq
\left(\tilde{\cD}_{\Gamma}[\bmu] (\bx)\right)_{j}
= 
\int_{\Gamma} T_{k,j,\ell} (\by,\bx) \left(\bR^{*}(\by) \cdot 
\bbmat
\mu_{\tau} (\by) \\
\mu_{\nu}(\by)
\ebmat
\right)_{k}
\nu_{\ell}(\by) dS_{\by} \, , \quad j,k,\ell=1,2 \, , 
\eeq
where $\bx \in \Omega$.
The following theorem reduces the velocity boundary value
problem~\cref{eq:StokesFlowEq,eq:MassConservation,eq:bcStokesFlow} 
to an integral equation on the boundary in the rotated frame. 
\begin{lem}
\label{lem:uniquenesstn}
Suppose $\bh=(h_{\tau},h_{\nu}) \in \bL^{2}(\Gamma)$ 
and that $\int_{\Gamma} 
h_{\nu} \, dS = 0$. 
Then there exists a unique solution $\bmu=(\mu_{\tau},\mu_{\nu})
\in \bL^{2}(\Gamma)$ 
which satisfies 
\beq
-\frac{1}{2}
\bbmat
\mu_{\tau} (\bx) \\
\mu_{\nu} (\bx)
\ebmat
+ \text{p.v.} \int_{\Gamma} \bK(\bx,\by) \bbmat
\mu_{\tau}(\by) \\
\mu_{\nu}(\by) 
\ebmat
\, dS_{\by}
= 
\bbmat
h_{\tau}(\bx) \\
h_{\nu}(\bx)
\ebmat \, ,
\quad
\bx \in \Gamma \,,
\label{eq:inteqStokestn}
\eeq
along with $\int_{\Gamma} \mu_{\nu} \, dS = 0$,
where 
\beq
\label{eq:kertn} 
\bK(\bx,\by) = \bR(\bx) \bK_{0}(\bx,\by) \bR^{*}(\by) \, , \quad \bx,\by\in\Gamma.
\eeq
Furthermore, $\bu(\bx) = \tilde{\cD}_{\Gamma}[\bmu](\bx)$ satisfies Stokes 
equations~\cref{eq:StokesFlowEq,eq:MassConservation}, along with
the boundary conditions $\bu(\bx) = \bh(\bx)$ for $\bx \in \Gamma$.
\end{lem}
\begin{proof}
The result is a straightforward consequence of~\cref{lem:uniquenesscart}.
\end{proof}

The following lemma extends~\cref{lem:uniquenesstn} 
to the case where the boundary $\Gamma$ is an open arc.
\begin{lem}
\label{lem:uniquenesstnopen}
Suppose that $\Gamma$ is an open arc and 
suppose $\bh=(h_{\tau},h_{\nu}) \in \bL^{2}(\Gamma)$, 
then there exists a unique solution $\bmu=(\mu_{\tau},\mu_{\nu})
\in \bL^{2}(\Gamma)$ 
which satisfies 
\beq
-\frac{1}{2}
\bbmat
\mu_{\tau} (\bx) \\
\mu_{\nu} (\bx)
\ebmat
+ \text{p.v.} \int_{\Gamma} \bK(\bx,\by) \bbmat
\mu_{\tau}(\by) \\
\mu_{\nu}(\by) 
\ebmat
\, dS_{\by}
= 
\bbmat
h_{\tau}(\bx) \\
h_{\nu}(\bx)
\ebmat \, ,
\quad
\bx \in \Gamma \, ,
\label{eq:inteqStokestn2} 
\eeq
where $\bK$ is defined by~\cref{eq:kertn}.
\end{lem}

\subsection{Integral equations on the wedge \label{sec:inteqonwedge}}
In order to investigate the behavior of solutions to integral 
equation~\cref{eq:inteqStokestn} on polygonal domains, we 
first analyze the local behavior of solutions on a wedge 
(see~\cref{fig:wedge1}).
The following observation reduces the analysis of the solution
$\mu$ on polygonal domains to its local behavior on a single wedge.

\begin{obsv}
Let $\Omega$ be a polygonal domain, and let $\bmu$ be the solution
to the integral equation~\cref{eq:inteqStokestn} corresponding to
a prescribed velocity $\bh$ on the boundary.
Using~\cref{lem:uniquenesstn}, we know that 
there exists a unique density $\bmu$ in $\bL^{2}(\Gamma)$. 
Let $\Gamma_{1}$ denote a wedge in the vicinity of one of the corners.
Then, the integral equation~\cref{eq:inteqStokestn} can be rewritten
as 
\beq
\label{eq:wedgeinteqderv}
-\frac{1}{2}
\bbmat
\mu_{\tau}(\bx) \\
\mu_{\nu}(\bx)
\ebmat
+ \text{p.v.}\int_{\Gamma_{1}} \bK(\bx,\by)
\bbmat
\mu_{\tau}(\by) \\
\mu_{\nu}(\by)
\ebmat
dS_{\by} 
= \bh(\bx) - \int_{\Gamma\setminus \Gamma_{1}} \bK(\bx,\by) 
\bbmat
\mu_{\tau}(\by) \\
\mu_{\nu}(\by)
\ebmat
dS_{\by} \,  
\eeq
for $\bx \in \Gamma_{1}$.
For any $\bmu \in \bL^{2}(\Gamma)$, the integral
\beq
\int_{\Gamma\setminus \Gamma_{1}} \bK(\bx,\by) 
\bbmat
\mu_{\tau}(\by) \\
\mu_{\nu}(\by)
\ebmat
dS_{\by} \, ,
\eeq
is a smooth function for $\bx \in \Gamma_{1}$.
%Let $\bmu^{\text{open}}(\bx)$ denote the solution to the integral
%equation~\cref{eq:wedgeinteqderv}.
Using~\cref{lem:uniquenesstnopen}, the unique solution
to the integral equation~\cref{eq:wedgeinteqderv} 
is the restriction of the solution $\bmu(\bx)$ to integral 
equation~\cref{eq:inteqStokestn} from $\Gamma$ to $\Gamma_{1}$.
%i.e.
%\beq
%\mu^{\text{open}} (\bx) = \mu(\bx) \, , \quad \bx \in \Gamma_{1} \, .
%\eeq
Moreover, in nearly all practical settings, the prescribed velocity
$\bh$ is also a piecewise smooth function on either side of the corner.
Thus, to understand the local behavior of the solution $\bmu$ to 
the integral equation on polygonal domains, it suffices to analyze
the restriction of the integral equation to an open wedge
with piecewise smooth velocity prescribed on either side.
\end{obsv}

Suppose $\bgamma(t):[-1,1] \to \R^2$ is 
a wedge with interior angle $\pi \theta$ and side length
$1$ on either side of the corner, 
parametrized by arc-length
(see,~\cref{fig:wedge1}).
\begin{figure}[h!]
\begin{center}
\includegraphics[width=4.5cm]{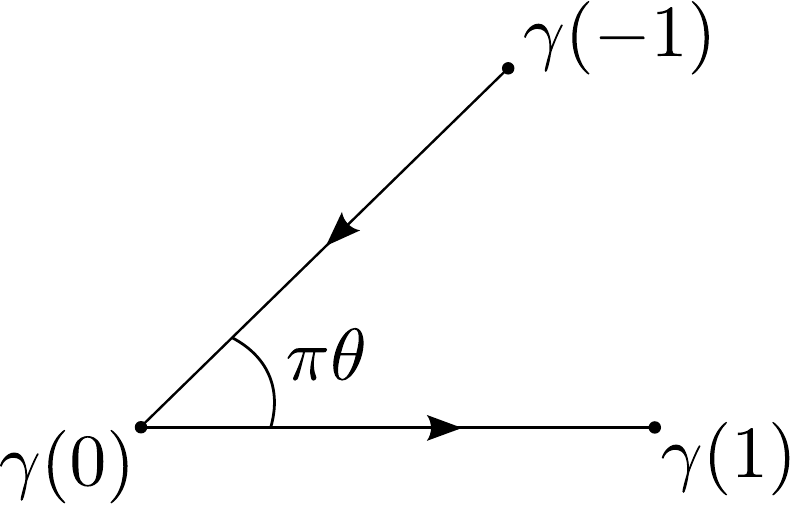}
\caption{A wedge with interior angle $\pi \theta$}
\label{fig:wedge1}
\end{center}
\end{figure}
In a slight misuse of notation, let  
$\mu_{\tau}(t)$ denote $\mu_{\tau}(\bgamma(t))$ for all $-1<t<1$.
Likewise,  let $\mu_{\nu}(t),h_{\tau}(t)$, and $h_{\nu}(t)$ denote
$\mu_{\nu}(\bgamma(t))$, $h_{\tau}(\bgamma(t))$ and
$h_{\nu}(\bgamma(t))$, respectively. 
The integral equation~\cref{eq:inteqStokestn2} 
for the velocity boundary value problem is then given by 
\beq
-\frac{1}{2}
\bbmat
\mu_{\tau}(t) \\
\mu_{\nu} (t)
\ebmat
+ \text{p.v.} \int_{-1}^{1} \bK(\bgamma(t),\bgamma(s)) \bbmat
\mu_{\tau}(s) \\
\mu_{\nu}(s) 
\ebmat
\, ds
= 
\bbmat
h_{\tau}(t) \\
h_{\nu}(t)
\ebmat \, ,
\, \,
-1<t<1 \,,
\label{eq:inteqStokestnwedge}
\eeq
where $\bK$ is defined by~\cref{eq:kertn}.
In this case,  the kernel $\bK$ has a simple form 
which is given by the following lemma.
\begin{lem}
\label{lem:wedker}
Suppose $\bgamma:[-1,1]\to \R^{2}$ is defined by the formula
\beq
\bgamma(t) = 
\begin{cases}
  -t \cdot (\cos(\pi\theta),\sin(\pi\theta)) & \mbox{if $-1< t< 0$} \, , \\
  (t,0) & \mbox{if $0< t < 1$} \, . 
\end{cases}
\label{eq:wedcur}
\eeq
Suppose further that $k_{j}(s,t)$, $j=1,2,3,4$, are defined by
\begin{align}
k_{1}(s,t) &= 
  \frac{t\sin(\pi\theta)}
    {\pi\left(s^2+t^2+2st\cos(\pi\theta)\right)^2} 
    (s +t\cost)(s\cost+t) \, ,\\
k_{2}(s,t) &=  
  \frac{t\sin(\pi\theta)}
    {\pi\left(s^2+t^2+2st\cos(\pi\theta)\right)^2} t\sint(t+s\cost) \, ,\\
k_{3}(s,t) &= - 
  \frac{t\sin(\pi\theta)}
    {\pi\left(s^2+t^2+2st\cos(\pi\theta)\right)^2} (s+t\cost)s\sint  \, ,\\
k_{4}(s,t) &= - 
  \frac{t\sin(\pi\theta)}
    {\pi\left(s^2+t^2+2st\cos(\pi\theta)\right)^2} st\sqsint \, .
\end{align}
for all $-1<s,t<1$.
If $0<t<1$, then
\beq
\bK(\bgamma(t),\bgamma(s))
= 
  \begin{cases}
  \displaystyle 
    \begin{bmatrix}
    k_{1}(s,t) & k_{2}(s,t) \\
    k_{3}(s,t) & k_{4}(s,t)
    \end{bmatrix}
    & \mbox{if $-1< s< 0$} \, ,    \vspace*{3pt} \\
  \begin{bmatrix} 0 & 0 \\
  0 & 0\end{bmatrix} & \mbox{if $0< s< 1$} \, .
  \end{cases} \label{1415}
\eeq
Likewise, if $-1<t<0$, then
\beq
\bK(\bgamma(t),\bgamma(s))
= 
  \begin{cases}
  \begin{bmatrix} 0 & 0 \\
  0 & 0\end{bmatrix} & \mbox{if $-1< s< 0$} \, , \vspace*{3pt} \\
  \displaystyle 
    \begin{bmatrix}
    -k_{1}(s,t) & k_{2}(s,t) \\
    k_{3}(s,t) & -k_{4}(s,t)
    \end{bmatrix}
    & \mbox{if $0< s< 1$}  \, .
  \end{cases} \label{1416}
\eeq
\end{lem}

\begin{rem}
The zero blocks in~\cref{1415,1416} are due to the fact that the 
self-interaction of the kernel on the sides is zero (in fact, the kernel
vanishes on any straight line).
\end{rem}

In the following theorem, we show that, 
when $\Gamma$ is a wedge, the integral equation~\cref{eq:inteqStokestnwedge}
decouples into two independent integral equations on the interval
$[0,1]$. 
\begin{thm}
\label{thm:tenotone}
Suppose $\mu_{\tau},\mu_{\nu}$ are functions in $\bL^{2}[-1,1]$. 
Let $h_{\tau},h_{\nu}$ be defined by~\cref{eq:inteqStokestnwedge}.
We denote the odd and the even parts of a function $f$ 
by $f_{o}$ and $f_{e}$ respectively, where
\beq
\label{eq:oddevendef}
f_{o}(s) = \frac{1}{2} \left(f(s) - f(-s) \right) \, , \quad \text{and} \, , \quad
f_{e}(s) = \frac{1}{2} \left(f(s) + f(-s) \right) \, ,
\eeq
for $-1<s<1$.
Then for $0<t<1$, 
\beq
\bbmat
h_{\tau,e}(t) \\
h_{\nu,o}(t) 
\ebmat
= 
-\frac{1}{2}
\bbmat
\mu_{\tau,e}(t) \\
\mu_{\nu,o}(t)
\ebmat
- \int_{0}^{1} 
\bbmat 
k_{1,1}(s,t) & k_{1,2}(s,t) \\
k_{2,1}(s,t) & k_{2,2}(s,t) 
\ebmat
\bbmat
\mu_{\tau,e}(s) \\
\mu_{\nu,o}(s)
\ebmat
\, ds \, ,
\eeq
and
\beq
\bbmat
h_{\tau,o}(t) \\
h_{\nu,e}(t) 
\ebmat
= 
-\frac{1}{2}
\bbmat
\mu_{\tau,o}(t) \\
\mu_{\nu,e}(t)
\ebmat
+ \int_{0}^{1} 
\bbmat 
k_{1,1}(s,t) & k_{1,2}(s,t) \\
k_{2,1}(s,t) & k_{2,2}(s,t) 
\ebmat
\bbmat
\mu_{\tau,o}(s) \\
\mu_{\nu,e}(s)
\ebmat
\, ds \, ,
\eeq
where $k_{j,\ell}(s,t)$, $j,\ell=1,2$, are given by
\begin{align}
\label{eq:defkers11} k_{1,1}(s,t) &= \frac{1}{2}\left(k_{1}(s,-t) - k_{1}(-s,t) \right) =  \frac{t \sint (s-t\cost)(t-s\cost)}{\pi (s^2 + t^2 -2st \cost)^2} \, ,\\ 
\label{eq:defkers12} k_{1,2}(s,t) &= \frac{1}{2}\left(-k_{2}(s,-t) + k_{2}(-s,t) \right) =  \frac{t^2 \sqsint (t-s\cost)}{\pi (s^2 + t^2 -2st \cost)^2} \, ,\\ 
\label{eq:defkers21} k_{2,1}(s,t) &= \frac{1}{2}\left(k_{3}(s,-t) - k_{3}(-s,t) \right) =  \frac{s t \sqsint (s-t\cost)}{\pi (s^2 + t^2 -2st \cost)^2} \, ,\\ 
\label{eq:defkers22} k_{2,2}(s,t) &= \frac{1}{2}\left(-k_{4}(s,-t) + k_{4}(-s,t) \right) =  \frac{st^2 \cubsint}{\pi (s^2 + t^2 -2st \cost)^2} \, ,
\end{align}
for all $0<s,t<1$.
\end{thm}
\begin{proof}
Substituting the expression for the kernel $\bK$ given by~\cref{1415,1416} 
in~\cref{eq:inteqStokestnwedge}, 
we get
\beq
\label{eq:lrinteq0}
\bbmat
h_{\tau}(t) \\
h_{\nu}(t) 
\ebmat
= -\frac{1}{2} \bbmat
\mu_{\tau}(t) \\
\mu_{\nu}(t) 
\ebmat
+ \int_{-1}^{0} 
\bbmat
k_{1}(s,t) & k_{2}(s,t) \\
k_{3}(s,t) & k_{4}(s,t)
\ebmat
\bbmat
\mu_{\tau}(s) \\
\mu_{\nu}(s)
\ebmat \, ds \, ,
\eeq
for $0<t<1$ and 
\beq
\label{eq:rlinteq0}
\bbmat
h_{\tau}(t) \\
h_{\nu}(t) 
\ebmat
= -\frac{1}{2} \bbmat
\mu_{\tau}(t) \\
\mu_{\nu}(t) 
\ebmat
+ \int_{0}^{1} 
\bbmat
-k_{1}(s,t) & k_{2}(s,t) \\
k_{3}(s,t) & -k_{4}(s,t)
\ebmat
\bbmat
\mu_{\tau}(s) \\
\mu_{\nu}(s)
\ebmat \, ds \, ,
\eeq
for $-1<t<0$. 
Then, making the change of variable $s\to -s$ in~\cref{eq:lrinteq0}
and the change of variable $t \to -t$ in~\cref{eq:rlinteq0}, we get
\beq
\label{eq:lrinteq}
\bbmat
h_{\tau}(t) \\
h_{\nu}(t) 
\ebmat
= -\frac{1}{2} \bbmat
\mu_{\tau}(t) \\
\mu_{\nu}(t) 
\ebmat
+ \int_{0}^{1} 
\bbmat
k_{1}(-s,t) & k_{2}(-s,t) \\
k_{3}(-s,t) & k_{4}(-s,t)
\ebmat
\bbmat
\mu_{\tau}(-s) \\
\mu_{\nu}(-s)
\ebmat \, ds \, ,
\eeq
for $0<t<1$ and 
\beq
\label{eq:rlinteq}
\bbmat
h_{\tau}(-t) \\
h_{\nu}(-t) 
\ebmat
= -\frac{1}{2} \bbmat
\mu_{\tau}(-t) \\
\mu_{\nu}(-t) 
\ebmat
+ \int_{0}^{1} 
\bbmat
-k_{1}(s,-t) & k_{2}(s,-t) \\
k_{3}(s,-t) & -k_{4}(s,-t)
\ebmat
\bbmat
\mu_{\tau}(s) \\
\mu_{\nu}(s)
\ebmat \, ds \, ,
\eeq
for $0<t<1$. 
Finally, combining~\cref{eq:lrinteq,eq:rlinteq}, we get
\begin{align}
\label{eq:teno0}
\bbmat
h_{\tau,e}(t) \\
h_{\nu,o}(t)
\ebmat
&=
-\frac{1}{2}\bbmat
\mu_{\tau,e}(t) \\
\mu_{\nu,o}(t)
\ebmat
%-\frac{1}{2} \int_{0}^{1}
%\bbmat
%-k_{1}(-s,t) + k_{1}(s,-t) & k_{2}(-s,t) - k_{2}(s,-t) \\
%-k_{3}(-s,t) + k_{3}(s,-t) & k_{4}(-s,t) - k_{4}(s,-t)
%\ebmat
-\int_{0}^{1} \bbmat
k_{1,1}(s,t) & k_{1,2}(s,t)  \\
k_{2,1}(s,t) & k_{2,2}(s,t)
\ebmat
\bbmat
\mu_{\tau,e}(s) \\
\mu_{\nu,o}(s) 
\ebmat \, ds \, , \\
\label{eq:tone0}
\bbmat
h_{\tau,o}(t) \\
h_{\nu,e}(t)
\ebmat
&=
-\frac{1}{2}\bbmat
\mu_{\tau,o}(t) \\
\mu_{\nu,e}(t)
\ebmat
%-\frac{1}{2} \int_{0}^{1}
%\bbmat
%-k_{1}(-s,t) + k_{1}(s,-t) & k_{2}(-s,t) - k_{2}(s,-t) \\
%-k_{3}(-s,t) + k_{3}(s,-t) & k_{4}(-s,t) - k_{4}(s,-t)
%\ebmat
+\int_{0}^{1} \bbmat
k_{1,1}(s,t) & k_{1,2}(s,t)  \\
k_{2,1}(s,t) & k_{2,2}(s,t)
\ebmat
\bbmat
\mu_{\tau,o}(s) \\
\mu_{\nu,e}(s) 
\ebmat \, ds \, .
\end{align}
\end{proof}

We represent the given velocity $\bh$ on the boundary in terms
of its tangential $h_{\tau}$ and normal $h_{\nu}$ components. 
We decompose both $h_{\tau}$ and $h_{\nu}$ into their odd and even 
parts, denoted by $h_{\tau,o}$, $h_{\tau,e}$, and $h_{\nu,o}$, $h_{\nu,e}$ 
as follows:
\begin{align}
h_{\tau,o}(t) = \frac{1}{2}\left(h_{\tau}(t) - h_{\tau}(-t) \right)\, , 
\quad
h_{\tau,e}(t) = \frac{1}{2}\left(h_{\tau}(t) + h_{\tau}(-t) \right)\,, 
\label{eq:htoddeven}\\
h_{\nu,o}(t) = \frac{1}{2}\left(h_{\nu}(t) - h_{\nu}(-t) \right)\, ,
\quad
h_{\nu,e}(t) = \frac{1}{2}\left(h_{\nu}(t) + h_{\nu}(-t) \right)\,, 
\label{eq:hnoddeven}
\end{align}
$-1<t<1$.
Suppose now that the densities $\mu_{\tau,e}$, $\mu_{\tau,o}$, $\mu_{\nu,e}$, and 
$\mu_{\nu,o}$
denote the solutions the integral equations~\cref{eq:teno0,eq:tone0}
with $h_{\tau,e}$, $h_{\tau,o}$, $h_{\nu,e}$, and $h_{\nu,o}$ defined
in~\cref{eq:htoddeven,eq:hnoddeven}.
Then $\bmu(t) = (\mu_{\tau}(t),\mu_{\nu}(t))$ defined by
\beq
\bbmat
\mu_{\tau}(t) \\
\mu_{\nu}(t)
\ebmat
= 
\frac{1}{2}\bbmat
\mu_{\tau,o}(t) + \mu_{\tau,e}(t) \\
\mu_{\nu,o}(t) + \mu_{\nu,e}(t) 
\ebmat
\, ,
\quad -1<t<1 \, ,
\eeq
satisfies integral equation~\cref{eq:inteqStokestnwedge}.

Thus, the integral equation~\cref{eq:inteqStokestnwedge} 
clearly splits into two cases:
\begin{itemize}
\item Tangential odd, normal even: the tangential components of
both the velocity field $\bh$ and the density $\bmu$, $h_{\tau}(t)$ and
$\mu_{\tau}(t)$, are odd functions of 
$t$ and the normal components $h_{\nu}(t)$ and $\mu_{\nu}(t)$ are even functions of $t$.
\item Tangential even, normal odd: the tangential components of
both the velocity field $\bh$ and the density $\bmu$, $h_{\tau}(t)$ and
$\mu_{\tau}(t)$, are even functions of 
$t$ and the normal components $h_{\nu}(t)$ and $\mu_{\nu}(t)$ are odd functions of $t$.
\end{itemize}

%%%%%%%%%%%%%%%%%%%%%%%%%%%%%%%%%%%%%%%%%%%%%%%%%%
\section{Analytical Apparatus \label{sec:analapp}}
In this section, we investigate
integrals of the form
\beq
\int_{0}^{1} k_{j,\ell} (s,t) s^{z}\,  ds \, , \quad j,\ell =1,2 \, ,
\eeq
for $0<t<1$, where $z \in \C$ with $\text{Re}(z)>-1$, and $k_{j,\ell}$
are defined in~\cref{eq:defkers11,eq:defkers12,eq:defkers21,eq:defkers22}.
We use the branch of $\log$ with $\mathrm{arg}(s) \in [0,2\pi)$ 
for the definition of $s^{z}$. 
The principal result of this section is~\cref{thm:int01allkers}.

On inspecting the kernels $k_{j,\ell}(s,t)$, $j,\ell=1,2$, we observe that 
it suffices to evaluate integrals of the form
\beq
\label{eq:iztheta}
I(z,\theta,t) = 
\frac{1}{\pi}\int_{0}^{1} \frac{s^{z} }{(s^{2} + t^{2} - 2st\cost)^2}
\, ds \, , \quad \text{for } 0<t<1 \, ,
\eeq
where $\text{Re}(z)>-1$. 
Using standard techniques in complex analysis, we first derive
an expression for the above integral from $0$ to $\infty$ 
in the following lemma.
\begin{rem}
Clearly, the integral~\cref{eq:iztheta}, is well-defined for
$\text{Re}(z)>-1$; in~\cref{lem:int0inf,lem:int1inf}, we make the additional
assumption that $\text{Re}(z)<3$. This restriction will be eliminated later.
\end{rem}
\begin{rem}
We consider integrals of the form~\cref{eq:iztheta} for
$\theta$ in the complex plane, since we will eventually complexify $\theta$
to better understand the behavior of integral equations~\cref{eq:teno0,eq:tone0}
on $\theta \in (0,2)$.
\end{rem}
\begin{lem}
\label{lem:int0inf}
Suppose $z \in \C$, $-1 < \text{Re}(z) < 3$, $z\neq 0,1,2$,  and
$\theta \in \C$. 
Then
\begin{align}
\label{eq:iztheta0inf}
I_{1}(z,\theta,t) = 
\frac{1}{\pi}\int_{0}^{\infty} 
\frac{s^z}{\left(s^2 + t^2 - 2 s t \cost \right)^2} \, ds  =
a(z,\theta) t^{z-3} \, , 
\end{align}
for $0<t<1$ 
where
\beq
a(z,\theta) = 
\frac{
z \sin{(2\pi\theta)} \cos{(\pi (1-\theta)z)} +  
2\sin{(\pi(1-\theta)z)}\left( 1-z\sqsint \right) 
}  
{4 \sinz \cubsint } \,.
\eeq
\end{lem}
\begin{proof}
Let $\Gammakey = \Gamma_{\eps}\cup \Gamma_{1} \cup \Gamma_{R} \cup \Gamma_{2}$ 
denote a keyhole contour 
where 
\begin{align*}
\Gamma_{\eps} =\left\{  -\eps e^{ix} \, ,\quad -\frac{\pi}{2} < x < 
\frac{\pi}{2} \right\} \, , &\quad 
\Gamma_{1}=\left\{x+i\eps\, , \quad 0\leq x \leq \sqrt{R^2-\eps^2}\right\} \, ,\\
\Gamma_{R}=\left\{R e^{ix}\, , \quad x_{0} < x < 2\pi - x_{0} \right \}\, , &\quad  
\Gamma_{2}=\left\{x-i\eps\, , \quad 0\leq x \leq \sqrt{R^2-\eps^2}\right\} \, ,
\end{align*}
where $x_{0} = \arctan{\eps/R}$, see~\cref{fig:keyhole}.
\begin{figure}
\begin{center}
\includegraphics[width=6cm] {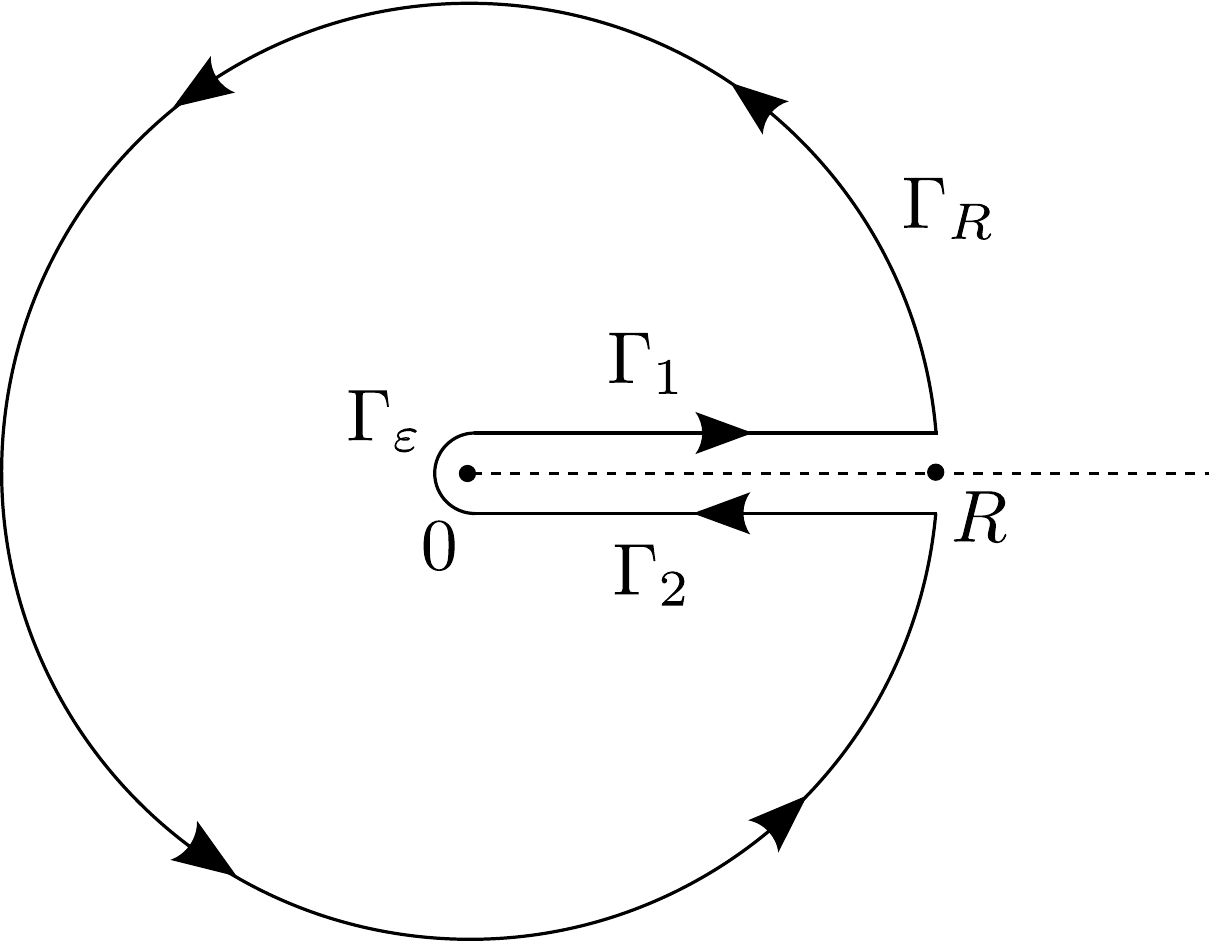} 
\end{center}
\caption{A keyhole contour}
\label{fig:keyhole}
\end{figure}
Using Cauchy's integral theorem,
\begin{align}
\int_{\Gammakey} 
& \frac{s^z}{\left(s^2 + t^2 - 2 s t \cost \right)^2} \, ds =
\left(\int_{\Gamma_{\eps}} + \int_{\Gamma_{1}} +  \int_{\Gamma_{R}} 
+ \int_{\Gamma_{2}} \right) 
\frac{s^z}{\left(s - te^{i\pi \theta} \right)^2 
\left( s- t e^{i \pi (2-\theta)} 
\right)^2} \, ds \, , \\
&= 2\pi i \left(
\frac{d}{ds} \left(
\frac{s^z}{\left(s-te^{i \pi(2-\theta)}\right)^2} \right) 
\Bigg\vert_{s = te^{i\pi \theta}}+
\frac{d}{ds}\left(
\frac{s^z}{\left(s - te^{i \pi \theta} \right)^2} \right) 
\Bigg\vert_{s = te^{i\pi (2-\theta)}} \right) \, , \\
&= -\frac{2\pi t^{z-3}}{\cubsint} 
\left( (z-2) 
\left(
e^{i \pi \theta z} 
- e^{2\pi i z} e^{-i \pi \theta z} 
\right)
- 
z
\left(
e^{i \pi \theta (z-2)} 
- e^{2\pi i z} e^{-i \pi \theta (z-2)} 
\right) 
\right) \, , \label{eq:residue}
\end{align}
for $0<t<1$.
Suppose $0<t<1$, then there exists a constant $M_{1}<\infty$
such that
\beq
\left| \frac{s^z}{\left(s^2 + t^2 - 2st \cost \right)^2} \right| \leq
M_{1} \eps^{\text{Re}(z)} \, , 
\eeq
for all $s \in \Gamma_{\eps}$.
Taking the limit as $\eps \to 0$, we get
\beq
\lim_{\eps\to 0} \left| 
\int_{\Gamma_{\eps}} 
\frac{s^z}{\left(s^2 + t^2 - 2 s t \cost \right)^2} \, ds 
\right|  \leq \lim_{\eps\to 0} \pi M_{1} \eps^{\text{Re}(z)+1} = 0 \, 
\label{eq:lim1} \, ,
\eeq
since $\text{Re}(z)>-1$.
Similarly, for $0<t<1$, there exists a constant $M_{2} < \infty$ 
such that
\beq
\left| \frac{s^z}{\left(s^2 + t^2 - 2st \cost \right)^2} \right| \leq
\frac{M_{2}}{R^{4-\text{Re}(z)}} \, , 
\eeq
for all $s\in \Gamma_{R}$.
Taking the limit as $R\to \infty$, we get
\beq
\lim_{R\to \infty} \left| 
\int_{\Gamma_{R}} 
\frac{s^z}{\left(s^2 + t^2 - 2 s t \cost \right)^2} \, ds 
\right| \leq \lim_{R \to \infty } \frac{2\pi M}{R^{3-\text{Re}(z)}} = 0 \, 
\label{eq:lim2} \, ,
\eeq
since $\text{Re}(z)< 3$.
On $\Gamma_{1}$ and $\Gamma_{2}$, taking the limits as $\eps \to 0$ and
$R \to \infty$, we get
\begin{align}
\lim_{\eps \to 0} \lim_{R\to\infty}
\int_{\Gamma_{1}} 
\frac{s^z}{\left(s^2 + t^2 - 2 s t \cost \right)^2} \, ds &= 
I_{1}(z,\theta,t)  
\label{eq:lim3}\, ,\\
\lim_{\eps \to 0} \lim_{R\to\infty}
\int_{\Gamma_{2}} 
\frac{s^z}{\left(s^2 + t^2 - 2 s t \cost \right)^2} \, ds &= -e^{2\pi i z}
I_{1}(z,\theta,t) \, \label{eq:lim4}.
\end{align}
The result follows by taking the limit $\eps \to 0$ and $R\to \infty$
in~\cref{eq:residue}, and using~\cref{eq:lim1,eq:lim2,eq:lim3,eq:lim4}.
\end{proof}

Suppose now that
$I_{2}(z,\theta,t)$ is defined by
\beq
I_{2}(z,\theta) = \int_{1}^{\infty} \frac{s^{z}}{(s^2+t^2 - 2st \cost)^2} \, 
ds \, ,
\eeq
for $0<t<1$.
Clearly, 
\beq
I(z,\theta,t) = I_{1}(z,\theta,t) - I_{2}(z,\theta,t) \, ,
\eeq
where $I(z,\theta,t)$ is given by~\cref{eq:iztheta} and $I_{1}(z,\theta,t)$ 
by~\cref{eq:iztheta0inf}.
In the following lemma, we compute an expression for $I_{2}(z,\theta,t)$ by 
deriving a Taylor expansion for 
\beq
f(s,t,\theta) = 
\frac{1}{(s^{2} +t^{2} -2st\cost)^2} \, , \quad \text{for } |s|>1, |t|< 1
\nonumber \, .
\eeq

\begin{lem}
\label{lem:int1inf}
Suppose that $\theta \in \C$. Then for all $|t|< 1$ and $|s|>1$
\beq
\label{eq:tayexpker}
f(s,t,\theta) = \frac{1}{\left( s^2 + t^2 - 2st \cost \right)^2}
= \sum_{n=0}^{\infty}  a_{n}(s,\theta) t^{n}  \, ,
\eeq
where
\beq 
\label{eq:anst}
a_{n} (s,\theta) = 
\frac{1}{4s^{n+4}\cubsint}
\left( (n+3)\sin{((n+1)\pi \theta)} - (n+1)\sin{((n+3)\pi \theta)} \right) \, . 
\eeq
Furthermore, for $-1<\text{Re}(z)<3$, and $z\neq 0,1,2$,  
\beq
I_{2}(z,\theta,t) = 
\frac{1}{\pi}\int_{1}^{\infty} 
s^z f(s,t,\theta) \, ds
= \sum_{n=0}^{\infty} F(n,z,\theta) t^{n} \, ,
\eeq
where
\beq
\label{eq:Fnzt}
F(n,z,\theta) = \frac{(n+1)\sin{((n+3)\pi \theta)}- 
(n+3)\sin((n+1)\pi \theta)}{ 4\pi(-z+n+3) \cubsint } \, .
\eeq
\end{lem}
\begin{proof}
For fixed $|s|>1$, $f(s,t)$ is analytic in $|t|< 1$. 
Using Cauchy's integral
formula, the coefficients of the Taylor series are given by
\begin{align}
a_{n}(s,\theta) &= \frac{1}{2\pi i}\int_{|\xi|=1} \frac{f(s,\xi,\theta)}{\xi^{n+1}} d\xi = 
\frac{1}{2\pi}\int_{0}^{2\pi} e^{-i n x} f(s,e^{i x},\theta) \, dx \, , \\
&= \frac{1}{2\pi} \int_{0}^{2\pi} \frac{e^{-i n x}}
{(s-e^{i (\pi \theta+x)})^2 (s-e^{i(-\pi \theta + x)})^2 } \, dx \, , \\
&= \frac{1}{2\pi s^4} \int_{0}^{2\pi} 
\frac{e^{-inx}}{\left(1- \frac{e^{i(\pi \theta +x)}}{s} \right)^2 
\left(1- \frac{e^{i(-\pi \theta + x)}}{s} \right)^2} \, dx \, .
\label{eq:fourker0}
\end{align}
Using the Taylor expansion of $1/(1-x)^2$ given by
\beq
\frac{1}{(1-x)^2} = \sum_{n=0}^{\infty} (n+1)x^{n} \, ,
\eeq
for $0<x<1$, equation~\cref{eq:fourker0} simplifies to,
\beq
a_{n}(s,\theta) = \frac{1}{2\pi} \frac{1}{s^{4}}\int_{0}^{2\pi} 
e^{-i n x}
\sum_{\ell=0}^{\infty}
\sum_{m=0}^{\infty} (\ell+1) (m+1) 
\frac{e^{i (\ell+m)x + (\ell -m) \pi \theta}}{s^{\ell+m}} \, dx
\label{eq:fourexpker}
\eeq
Due to the orthogonality of the Fourier basis $\{ e^{i n x} \}$
in $\bL^{2}[0,2\pi]$, the
only terms that contribute to the integral in~\cref{eq:fourexpker} are
when $\ell+m = n$.
Thus,
\beq
a_{n}(s,\theta) = \frac{1}{s^{n+4}}
\sum_{\ell+m = n} (\ell+1)(m+1) e^{i (\ell -m) \pi \theta} =
\frac{e^{-i n \pi \theta}}{s^{n+4}}\sum_{\ell=0}^{n} (\ell+1)(n-\ell+1) 
e^{2 i \ell \pi \theta} \, .
\eeq
The result for the Taylor series for $f(s,t)$ 
then follows by using
\beq
\sum_{\ell=0}^{n} \ell^{p} e^{i \ell x} = \frac{d^{p}}{dx^{p}} 
\left( \frac{1-e^{i (n+1) x}}{1-e^{ix}}\right) \, \quad p=0,1,2 \, .
\eeq

Given the Taylor expansion for $f(s,t)$, the integral 
$I_{2}(z,\theta,t)$ can be computed by switching the order of summation and 
integration and using 
\beq
\int_{1}^{\infty} s^{z} a_{n}(s,\theta) \, ds = F(n,z,\theta) \, ,
\eeq
which concludes the proof.
\end{proof}

\begin{rem}
If $m$ is an integer, the functions $a_{n}(s,\theta)$ and
$F(n,z,\theta)$ defined in~\cref{eq:anst,eq:Fnzt} respectively, remain
bounded as $\theta \to m$: 
\begin{align*}
\lim_{\theta \to m}a_{n}(s,\theta) = -\frac{(n+1)(n+2)(n+3)}{6s^{n+4}} \, , &\quad
\lim_{\theta \to m}F(n,z,\theta) = -\frac{(n+1)(n+2)(n+3)}{6\pi(-z+n+3)} \, ,
\end{align*}
when $m$ is even, and
\begin{align*}
\lim_{\theta \to m}a_{n}(s,\theta) = (-1)^{n}\frac{(n+1)(n+2)(n+3)}{6s^{n+4}} \, , &\quad
\lim_{\theta \to m}F(n,z,\theta) = (-1)^n\frac{(n+1)(n+2)(n+3)}{6\pi(-z+n+3)}  \, ,
\end{align*}
when $m$ is odd.
\end{rem}

Using~\cref{lem:int0inf,lem:int1inf}, we compute $I(z,\theta,t)$
in the following theorem.
\begin{thm}
\label{thm:int01}
Suppose $-1<\text{Re}(z)<3$, and $z\neq 0,1,2$, then
\beq
I(z,\theta,t) = \int_{0}^{1} \frac{s^{z}}{(s^2 + t^2 -2st\cost)^2} 
= a(z,\theta) t^{z-3} - \sum_{n=1}^{\infty} F(n,z,\theta) t^{n} \, ,
\label{eq:int01}
\eeq
for $0<t<1$.
\end{thm}
We observe that both the expressions on the left and right of~\cref{eq:int01} 
are analytic functions of $z$ for $\text{Re}(z)>-1$ 
and $z$ not an integer.
In the following theorem, we extend the definition
of $I(z,\theta,t)$ to $\text{Re}(z)>-1$ and $z$ not an integer.
\begin{thm}
\label{thm:int01analyticcont}
Suppose that $z\in \C$, $\text{Re}(z)>-1$, $z\neq 0,1,2\ldots$, and $\theta \in \C$.
Then
\beq
I(z,\theta,t) = \int_{0}^{1} \frac{s^{z}}{(s^2 + t^2 -2st\cost)^2} 
= a(z,\theta)t^{z-3} - \sum_{n=1}^{\infty} F(n,z,\theta) t^{n} \, ,
\eeq
for $0<t<1$.
\end{thm}
\begin{proof}
The result follows from~\cref{eq:int01} and using analytic continuation.
\end{proof}

We now present the principal result of this section in the following theorem.
\begin{thm}
\label{thm:int01allkers}
Suppose $z\in \C$, $\text{Re}(z)>-1$, 
$z\neq 0,1,2,\ldots$, and $\theta \in \C$.
Recall that,
\begin{align}
k_{1,1}(s,t) &= \frac{t \sint (s-t\cost)(t-s\cost)}{\pi (s^2 + t^2 -2st \cost)^2} \, ,\\ 
k_{1,2}(s,t) &= \frac{t^2 \sqsint (t-s\cost)}{\pi (s^2 + t^2 -2st \cost)^2} \, ,\\ 
k_{2,1}(s,t) &= \frac{s t \sqsint (s-t\cost)}{\pi (s^2 + t^2 -2st \cost)^2} \, ,\\ 
k_{2,2}(s,t) &= \frac{st^2 \cubsint}{\pi (s^2 + t^2 -2st \cost)^2} \, , 
\end{align}
for $0<s,t<1$.
Then
\begin{align}
\int_{0}^{1} k_{1,1}(s,t) s^{z} ds &= a_{1,1}(z,\theta) \cdot t^{z} + \sum_{n=1}^{\infty} 
F_{1,1} (n,z,\theta) \cdot t^{n} \, , \\
\int_{0}^{1} k_{1,2}(s,t) s^{z} ds &= a_{1,2}(z,\theta) \cdot t^{z} + \sum_{n=2}^{\infty} 
F_{1,2} (n,z,\theta) \cdot t^{n} \, , \\
\int_{0}^{1} k_{2,1}(s,t) s^{z} ds &= a_{2,1}(z,\theta) \cdot t^{z} + \sum_{n=1}^{\infty} 
F_{2,1} (n,z,\theta) \cdot t^{n} \, , \\
\int_{0}^{1} k_{2,2}(s,t) s^{z} ds &= a_{2,2}(z,\theta) \cdot t^{z} + \sum_{n=2}^{\infty} 
F_{2,2} (n,z,\theta) \cdot t^{n} \, , 
\end{align}
for $0<t<1$,
where
\begin{align}
\label{eq:defsing11} 
a_{1,1}(z,\theta) &= 
\frac{1}{2\sinz}\left[(z+1)\sint \cos{(\pi z(1-\theta))} 
-\sin{(\pi(z(1-\theta)+\theta))}\right] \, , \\
\label{eq:defsing12} 
a_{1,2}(z,\theta) &= 
-\frac{1}{2\sinz}
(z-1)
\sint
\sin{(\pi z(1-\theta))} \, , \\
\label{eq:defsing21} 
a_{2,1}(z,\theta) &= 
\frac{1}{2\sinz}(z+1)
\sint
\sin{(\pi z(1-\theta) )} \, ,\\
\label{eq:defsing22} 
a_{2,2}(z,\theta) &= 
\frac{1}{2\sinz}\left[(z+1)\sint \cos{(\pi z(1-\theta))} 
+\sin{(\pi(z(1-\theta) - \theta))}\right] \, , 
\end{align}
and
\begin{align}
\label{eq:defsmooth11} 
F_{1,1}(n,z,\theta) &= 
\frac{n\sint\cos{(n\pi\theta)}+
\cost \sin{(n\pi\theta)}}{2 \pi (n-z)}  \, , \\
\label{eq:defsmooth12} 
F_{1,2}(n,z,\theta) &= 
\frac{(n-1)\sint\sin{(n\pi\theta)}}{2\pi (n-z)}  \, , \\
\label{eq:defsmooth21} 
F_{2,1}(n,z,\theta) &=
-\frac{(n+1)\sint\sin{(n\pi\theta)}}{2\pi(n-z)} \, , \\
\label{eq:defsmooth22} 
F_{2,2}(n,z,\theta) &= 
\frac{n\sint\cos{(n\pi\theta)}-
\cost \sin{(n\pi\theta)}}{2\pi(n-z)} \, . 
\end{align}
\end{thm}
\begin{proof}
The result follows from observing that
\begin{align}
\int_{0}^{1} k_{1,1}(s,t)s^{z} \, ds &= 
-\frac{\sin{(2\pi\theta)}}{2} \left(t\cdot I(z+2,\theta,t) + t^{3} \cdot I(z,\theta,t) 
\right)  \\ 
& \hspace{5ex}+ t^{2} (1+ \sqcost) \cdot I (z+1,\theta,t) \, , \nonumber \\
\int_{0}^{1} k_{1,2}(s,t)s^{z} \, ds &= t^{2} \sqsint \left(
t\cdot I(z,\theta,t)- \cost I(z+1,\theta,t) \right) \, , \\
\int_{0}^{1} k_{2,1}(s,t)s^{z} \, ds &=
t\sqsint \left( I(z+2,\theta,t) - t \cost \cdot I(z+1,\theta,t) \right) \, ,\\
\int_{0}^{1} k_{2,2}(s,t) s^{z} \, ds &=
t^2 \cubsint \cdot I(z+1,\theta,t) \, ,
\end{align}
and using the formula for $I(z,\theta,t)$ derived in~\cref{thm:int01analyticcont}.
\end{proof}
%%%%%%%%%%%%%%%%%%%%%%%%%%%%%%%%%%%%%%%%%%%%%%%%%%%%
\section{Analysis of the integral equation \label{sec:inteq}}
Suppose that $\bgamma(t):[-1,1] \to \R^2$ is 
a wedge with interior angle $\pi \theta$ and side length
$1$ on either side of the corner, 
parametrized by arc length
(see~\cref{fig:wedge1}).
Suppose further that the odd and the even parts of a function $f$ 
are denoted by $f_{o}$ and $f_{e}$ respectively (see~\cref{eq:oddevendef}).
In~\cref{sec:inteqonwedge}, we observed that the integral 
equation 
\beq
-\frac{1}{2}
\bbmat
\mu_{\tau}(t) \\
\mu_{\nu} (t)
\ebmat
+ \text{p.v.} \int_{-1}^{1} \bK(\bgamma(t),\bgamma(s)) \bbmat
\mu_{\tau}(s) \\
\mu_{\nu}(s) 
\ebmat
\, ds
= 
\bbmat
h_{\tau}(t) \\
h_{\nu}(t)
\ebmat \, ,
\, \,
-1<t<1 \,,
\label{eq:inteqStokestnwedge2}
\eeq
simplifies into two uncoupled integral equations on the interval [0,1],
given by
\beq
\label{eq:tone}
-\frac{1}{2}
\bbmat
\mu_{\tau,o}(t) \\
\mu_{\nu,e}(t)
\ebmat
+ \int_{0}^{1} 
\bbmat 
k_{1,1}(s,t) & k_{1,2}(s,t) \\
k_{2,1}(s,t) & k_{2,2}(s,t) 
\ebmat
\bbmat
\mu_{\tau,o}(s) \\
\mu_{\nu,e}(s)
\ebmat
\, ds  \, = 
\bbmat
h_{\tau,o}(t) \\
h_{\nu,e}(t) 
\ebmat \, ,
\eeq
for $0<t<1$,
and 
\beq
\label{eq:teno}
-\frac{1}{2}
\bbmat
\mu_{\tau,e}(t) \\
\mu_{\nu,o}(t)
\ebmat
- \int_{0}^{1} 
\bbmat 
k_{1,1}(s,t) & k_{1,2}(s,t) \\
k_{2,1}(s,t) & k_{2,2}(s,t) 
\ebmat
\bbmat
\mu_{\tau,e}(s) \\
\mu_{\nu,o}(s)
\ebmat
\, ds \, =
\bbmat
h_{\tau,e}(t) \\
h_{\nu,o}(t) 
\ebmat \, , 
\eeq
for $0<t<1$, where $\bK$ is defined in~\cref{eq:kertn} and 
the kernels $k_{j,\ell}(s,t)$, $j,\ell=1,2$, are defined 
in~\cref{eq:defkers11,eq:defkers12,eq:defkers21,eq:defkers22}.
As in~\cref{sec:inteqonwedge}, 
we refer to~\cref{eq:tone} as the tangential odd, normal even case 
and to~\cref{eq:teno} as the tangential even, normal odd case.

In~\cref{sec:tone}, we investigate equation~\cref{eq:tone}, i.e., the tangential 
odd, normal even case. 
We determine two countable collections of values 
$p_{n,j},q_{n,j},z_{n,j} \in \C$,
$n=1,2,\ldots$, $j=1,2$, depending
on the angle $\pi \theta$, such that 
if $\mu_{\tau,o}(t)$, and $\mu_{\nu,e}(t)$ are defined by
\beq
\label{eq:dftone}
\mu_{\tau,o}(t) = p_{n,j} \cdot |t|^{z_{n,j}} \sgn(t) \, , \quad \text{and} 
\quad \mu_{\nu,e}(t) = q_{n,j} \cdot |t|^{z_{n,j}} \, , 
\eeq
then the corresponding components of the velocity 
$h_{\tau,o}(t)$ and $h_{\nu,e}(t)$, 
defined by~\cref{eq:tone} are smooth.
We also prove the converse. 
Suppose that $N$ is a positive integer. 
Suppose further that $\alpha_{n},\beta_{n} \in \C$, 
$n=0,1,2\ldots N$, and
$h_{\tau,o}(t)$
and $h_{\nu,e}(t)$ are given by 
\beq
\label{eq:defhtone}
h_{\tau,o}(t) = \left(\sum_{n=0}^{N} \alpha_{n} |t|^{n} \right)\sgn(t) 
\, ,  \quad \text{and}
\quad h_{\nu,e}(t) = \sum_{n=0}^{N} \beta_{n} |t|^{n} \, , 
\eeq
for $-1<t<1$.
Then for all but countably many $0<\theta<2$, 
there exist unique numbers $c_{n},d_{n} \in \C$, $n=0,1,\ldots N$,
such that $\mu_{\tau,o}(t)$ and $\mu_{\nu,e}(t)$ defined by
\begin{align}
\label{eq:defmu1tone}
\mu_{\tau,o}(t) &= \left(c_{0}+ \sum_{n=1}^{N} c_{n} p_{n,1} |t|^{z_{n,1}}
+ d_{n} p_{n,2} |t|^{z_{n,2}} \right) \sgn(t) \, ,\\
\label{eq:defmu2tone}
\mu_{\nu,e}(t) &= d_{0} + \sum_{n=1}^{N} c_{n} q_{n,1} |t|^{z_{n,1}} + 
d_{n} q_{n,2} |t|^{z_{n,2}} \, ,
\end{align}
for $-1<t<1$
satisfy equation~\cref{eq:tone} with error $O(|t|^{N+1})$. 
We prove this result in~\cref{thm:mainBinvertibility}.

Similarly, in~\cref{sec:teno}, we investigate equation~\cref{eq:teno}, i.e., the tangential 
even, normal odd case.
We determine two countable collections 
of values $p_{n,j},q_{n,j},z_{n,j} \in \C$,
$n=1,2,\ldots$, $j=1,2$, depending
on the angle $\pi \theta$, such that, if 
$\mu_{\tau,e}(t)$, and $\mu_{\nu,o}(t)$ are defined by
\beq
\label{eq:dfnote}
\mu_{\tau,e}(t) = p_{n,j} \cdot |t|^{z_{n,j}}  \, , \quad \text{and} 
\quad \mu_{\nu,o}(t) = q_{n,j} \cdot |t|^{z_{n,j}} \sgn{(t)} \, , 
\eeq
then the corresponding components of the velocity 
$h_{\tau,e}(t)$ and $h_{\nu,o}(t)$, 
defined by~\cref{eq:teno} are smooth.
We also prove the converse. 
Suppose that $N$ is a positive integer. 
Suppose further that $\alpha_{n},\beta_{n} \in \C$, 
$n=0,1,2\ldots N$, and
$h_{\tau,e}(t)$
and $h_{\nu,o}(t)$ are given by 
\beq
\label{eq:defhnote}
h_{\tau,e}(t) = \sum_{n=0}^{N} \alpha_{n} |t|^{n} 
\, ,  \quad \text{and}
\quad h_{\nu,o}(t) = \left( \sum_{n=0}^{N} \beta_{n} |t|^{n} \right)\sgn(t) \, , 
\eeq
for $-1<t<1$.
Then for all but countably many $0<\theta<2$, 
there exist unique numbers $c_{n},d_{n} \in \C$, $n=0,1,\ldots N$,
such that $\mu_{\tau,e}(t)$ and $\mu_{\nu,o}(t)$ defined by
\begin{align}
\label{eq:defmu1note}
\mu_{\tau,e}(t) &= c_{0}+ \sum_{n=1}^{N} c_{n} p_{n,1} |t|^{z_{n,1}}
+ d_{n} p_{n,2} |t|^{z_{n,2}} \, ,\\
\label{eq:defmu2note}
\mu_{\nu,o}(t) &= \left(d_{0} + \sum_{n=1}^{N} c_{n} q_{n,1} |t|^{z_{n,1}} + 
d_{n} q_{n,2} |t|^{z_{n,2}} \right) \sgn{(t)} \, ,
\end{align}
for $-1<t<1$
satisfy equation~\cref{eq:teno} with error $O(|t|^{N+1})$. 
We prove this result in~\cref{thm:mainBinvertibilitynote}.

\begin{rem}
Although we use the same symbols for the 
countable collection of values $p_{n,j}$, $q_{n,j}$, and $z_{n,j}$,
$n=1,2\ldots$, $j=1,2$, for the tangential even, normal odd case are
and the tangential odd, normal even case, their values are in fact 
different---they are defined by different formulae.
\end{rem}

\begin{rem}
\label{rem:infosc}
We note that the the real and imaginary parts of the function 
$|t|^{z}$ are given by 
$|t|^{\alpha} \cos{(\beta \log{|t|})}$ 
and $|t|^{\alpha} \sin{(\beta \log |t|)}$ respectively, where 
$z = \alpha + $i$ \beta$.
Analogous results where the density $\bmu$ 
is expressed in terms of the
functions $|t|^{\alpha} \cos {(\beta \log|t|)}$
and $|t|^{\alpha} \sin{(\beta \log|t|)}$, as opposed to $|t|^{z}$, 
can be derived for both the tangential odd, normal even case,
and the tangential even, normal odd case.
\end{rem}

\subsection{Tangential odd, normal even case \label{sec:tone}}
In this section, we investigate the tangential odd, normal even case 
(see equation~\cref{eq:tone}).
In~\cref{sec:singpowtone}, we determine the values $p_{n,j}$, $q_{n,j}$ and
$z_{n,j}$, $n=1,2,\ldots$, $j=1,2$, 
in~\cref{eq:dftone} for which the resulting components
of the velocity are smooth functions.
%For any integer $N$, we show that there exists a collection of $2N$
%values of $p$, $q$ and $z$ and derive an explicit expression for
%the velocity when the density is a linear combination of the densities
%defined in~\cref{eq:dftone}.
In~\cref{sec:binv}, we show that, for every $\bh$ of the form~\cref{eq:defhtone}, 
there exists a density $\bmu$ 
of the form~\cref{eq:defmu1tone,eq:defmu2tone}, 
which satisfies the integral equation~\cref{eq:tone} to order $N$.

\subsubsection{The values of $p_{n,j}$, $q_{n,j}$, and $z_{n,j}$ 
in~\cref{eq:dftone} 
\label{sec:singpowtone}}
Suppose that $\mu_{\tau,o}(t)$ and $\mu_{\nu,e}(t)$ are given by 
\beq
\mu_{\tau,o}(t)=p \cdot |t|^{z} \sgn(t) \, ,\quad 
\text{and} \quad \mu_{\nu,e}(t) = q \cdot |t|^{z} \, , \label{eq:mudeftone}
\eeq
for $-1<t<1$, where $p,q,z \in \C$.
In this section, we determine the values of $p,q$ and $z$ such that 
$h_{\tau,o}(t)$ and $h_{\nu,e}(t)$ defined by
\beq
\label{eq:tonecaseeq1}
\bbmat
h_{\tau,o}(t) \\
h_{\nu,e}(t) 
\ebmat
= 
-\frac{1}{2}
\bbmat
\mu_{\tau,o}(t) \\
\mu_{\nu,e}(t)
\ebmat
+ \int_{0}^{1} 
\bbmat 
k_{1,1}(s,t) & k_{1,2}(s,t) \\
k_{2,1}(s,t) & k_{2,2}(s,t) 
\ebmat
\bbmat
\mu_{\tau,o}(s) \\
\mu_{\nu,e}(s)
\ebmat
\, ds \, ,
\eeq
are smooth functions of $t$ for $0<t<1$, where the kernels 
$k_{j,\ell}(s,t)$, $j,\ell=1,2$, are defined 
in~\cref{eq:defkers11,eq:defkers12,eq:defkers21,eq:defkers22}.
The principal result of this section is~\cref{lem:multpowexp}.

The following lemma describes sufficient conditions for
$p,q$, and $z$ such that, if $\bmu$ is
defined by~\cref{eq:mudeftone}, then the velocity
$\bh$ given by~\cref{eq:tonecaseeq1} is smooth. 
\begin{thm}
\label{thm:singpowexp}
Suppose that $\theta \in (0,2)$, $z$ is not an integer, and $z$ satisfies 
$\det{\bA(z,\theta)} = 0$,
where 
\beq
\label{eq:defamat}
\bA(z,\theta) = -\frac{1}{2}\bI + \bbmat
a_{1,1}(z,\theta) & a_{1,2}(z,\theta) \\
a_{2,1}(z,\theta) & a_{2,2}(z,\theta)   
\ebmat
\, ,
\eeq
$\bI$ is the $2\times2$ identity matrix, and $a_{j,\ell}(z,\theta)$,
$j,\ell=1,2$,
are given by~\cref{eq:defsing11,eq:defsing12,eq:defsing21,eq:defsing22}.
Furthermore, suppose that $(p,q) \in \cN\{\bA(z,\theta)\}$, 
where $\cN\{\bA \}$ denotes the null space of the matrix $\bA$.
Suppose finally that
\beq
\mu_{\tau,o}(t) = p \cdot t^{z} \, , \quad \text{and}
\quad \mu_{\nu,e}(t) = q \cdot t^{z} \, ,
\eeq
for $0<t<1$. 
Then $h_{\tau,o}(t)$ and $h_{\nu,e}(t)$ defined by~\cref{eq:tonecaseeq1} satisfy
\beq
\bbmat 
h_{\tau,o}(t)  \\
h_{\nu,e} (t)
\ebmat
=  
\sum_{n=1}^{\infty}
\bF(n,z,\theta)
\bbmat
p \\
q
\ebmat 
\cdot t^{n} \, ,
\label{eq:singlesingpowexp}
\eeq
for $0<t<1$, where
\beq
\label{eq:kerssmoothmat}
\bF(n,z,\theta) = 
\bbmat
F_{1,1}(n,z,\theta) & F_{1,2}(n,z,\theta) \\
F_{2,1}(n,z,\theta) & F_{2,2}(n,z,\theta)
\ebmat \, ,
\eeq
and $F_{j,\ell}(n,z,\theta)$, $j,\ell=1,2$, are given
by~\cref{eq:defsmooth11,eq:defsmooth12,eq:defsmooth21,eq:defsmooth22}.
\end{thm}
\begin{proof}
Substituting $\mu_{\tau,o}(t) = p \cdot t^{z}$ 
and $\mu_{\nu,e}(t) = q \cdot t^{z}$ in~\cref{eq:tonecaseeq1} and 
using~\cref{thm:int01allkers}, 
we get
\begin{align}
\bbmat
h_{\tau,o}(t) \\
h_{\nu,e}(t)
\ebmat
&=
\bbmat
-1/2 p \cdot t^{z} + \int_{0}^{1} 
\left(p \cdot k_{1,1}(s,t) + q \cdot k_{1,2}(s,t)  \right) s^{z} \, ds \\ 
-1/2 q \cdot t^{z} + \int_{0}^{1} 
\left(p \cdot k_{2,1}(s,t) + q \cdot k_{2,2}(s,t)  \right) s^{z} \, ds \\ 
\ebmat
\, , \\
&=
\bA(z,\theta)
\bbmat
p \\
q
\ebmat
t^{z} + 
\sum_{n=1}^{\infty}
\bF(n,z,\theta)
\bbmat
p \\
q
\ebmat 
\cdot t^{n} \, , 
\label{eq:velboundarysingpow} 
\end{align}
Since $(p,q) \in \cN\{\bA(z,\theta)\}$, we note that 
$\bA(z,\theta) \cdot (p,q) = (0,0)$ and thus
\beq
\bbmat
h_{\tau,o}(t) \\
h_{\nu,e}(t)
\ebmat
= 
\sum_{n=1}^{\infty}
\bbmat
F_{1,1}(n,z,\theta) & F_{1,2}(n,z,\theta) \\
F_{2,1}(n,z,\theta) & F_{2,2}(n,z,\theta)
\ebmat
\bbmat
p \\
q
\ebmat 
\cdot t^{n} \, .
\eeq
\end{proof}

A straightforward calculation shows that
\beq
\label{eq:detval}
\det{\bA(z,\theta)} = 
\frac{\left( z \sint - \sinzt \right) \left( z \sint - \sinztc \right)}{4 
\sqsinz}
\eeq
Thus,
if $z$ is not an integer and
either
\beq
\label{eq:implfuns00}
z\sint - \sinztc = 0 \, , 
\eeq
or 
\beq
z \sint - \sinzt = 0 \, , 
\label{eq:implfuns01}
\eeq
then $\det{\bA(z,\theta)} = 0$. 

\begin{rem}
\label{rem:pqdetermination}
If $z$ satisfies the implicit relations~\cref{eq:implfuns00,eq:implfuns01},
we have that $\det{\bA(z,\theta)=0}$. 
It is then straightforward to determine $p$ and $q$ via an explicit formula
in terms of the entries of $\bA(z,\theta)$, since $(p,q) \in \cN(\bA (z,\theta))$.
\end{rem}

In the following theorem,
we prove the existence of the implicit
functions $z(\theta)$  
defined by~\cref{eq:implfuns00,eq:implfuns01}
on the interval $\theta \in (0,2)$.

\begin{thm}
\label{thm:mainsingpowtone}
Suppose that $N\geq 2$ is an integer.
Then there exists $3N-2$ real numbers 
$\theta_{1},\theta_{2},\ldots \theta_{3N-2} \in (0,2)$
such that the following holds.
Suppose that $D$ is the strip in the upper half plane 
with $0<\text{Re}(\theta)<2$ that includes the
interval $(0,2)\setminus \{ \theta_{j} \}_{j=1}^{3N-2}$, i.e.
\beq
\label{eq:defDmaintone}
D = \{ \theta \in \C \, : \, \text{Re}(\theta) \in (0,2) \, , 
\quad 0 \leq \text{Im}(\theta) < \infty \} \setminus \{ \theta_{j} \}_{j=1}^{3N-2} \, . 
\eeq
Then, there exists a simply connected open set $D \subset V \subset \C$
and analytic functions $z_{n,1}(\theta): V\to \C$, $n=1,2\ldots N$, 
which satisfy
\beq
z\sint - \sinztc = 0 \, , \quad z(1) =  n \, , \label{eq:implfun1main}
\eeq
for $\theta \in V$, and analytic functions
$z_{n,2}(\theta):V \to \C$, $n=2,3,\ldots N$, which satisfy
\beq
z\sint - \sinzt = 0 \, , \quad z(1) =  n \, , \label{eq:implfun2main}
\eeq
for $\theta \in V$ (see~\cref{fig:vdommain} for an illustrative domain $V$).
Moreover, the functions 
$z_{n,1}(\theta)$, $n=1,2\ldots N$, do not take integer values for all
$\theta \in V \setminus \{ 1\}$,  
and satisfy
$\det{\bA(z_{n,1}(\theta),\theta)}=0$, $n=1,2\ldots N$, 
for all $\theta \in V$
(see~\cref{eq:defamat,eq:detval}).
Similarly, the functions
$z_{n,2}(\theta)$, $n=2,3,\ldots N$, do not take integer values
for all $\theta \in V \setminus \{ 1 \}$, and satisfy
$\det{\bA(z_{n,2}(\theta),\theta)}=0$, $n=2,3\ldots N$, 
for all $\theta \in V$
(see~\cref{eq:defamat,eq:detval}).
\end{thm}

\begin{proof}
The proof is technical and is contained in~\cref{sec:appa}.
\end{proof}

\begin{figure}
\begin{center}
\includegraphics[width=8cm]{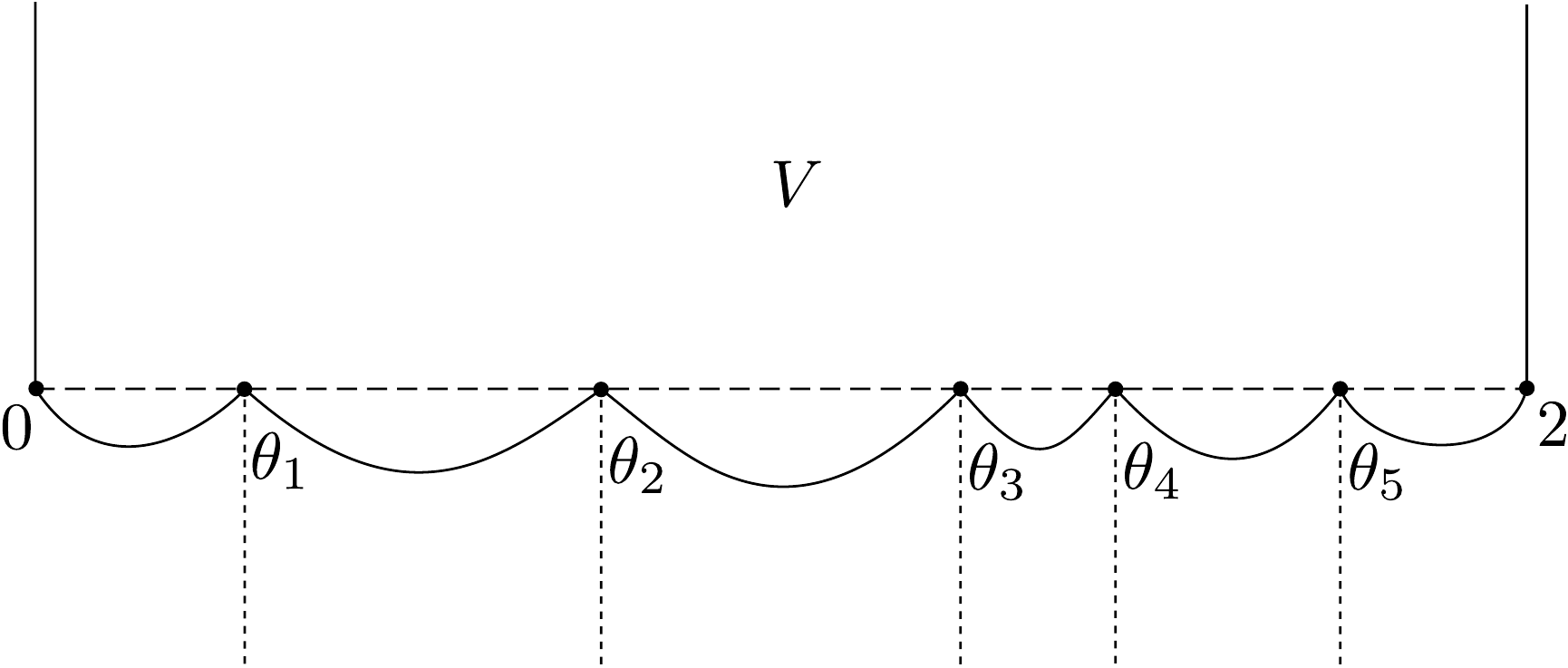}
\caption{An illustrative domain $V$ for the case $N=2$, where 
$\theta_{1}, \theta_{2}, \ldots \theta_{5}$ are the combined 
branch points of the functions 
$z_{1,1}(\theta)$, $z_{2,1}(\theta)$, and $z_{2,2}(\theta)$.
The large dashes are used to denote the interval $(0,2)$ for reference. 
The thin dashes are the locations of the branch cuts at $\theta_{1},
\theta_{2}\ldots \theta_{5}$.}
\label{fig:vdommain}
\end{center}
\end{figure}

\begin{rem}
\label{rem:branchpoints}
In fact, the real numbers $\theta_{1},\theta_{2}\ldots \theta_{3N-2}$ 
are the combined branch points of the functions $z_{n,1}(\theta)$, 
$n=1,2,\ldots N$, and $z_{n,2}(\theta)$, $n=2,3,\ldots N$. 
(see~\cref{sec:appa}).
\end{rem}

\begin{rem}
As shown in~\cref{sec:appa}, the functions
$z_{n,1}(\theta)$, $n=2\ldots N$, and $z_{n,2}(\theta)$, $n=3,\ldots N$, 
have exactly three branch singularities each, and the functions
$z_{1,1}(\theta)$ and $z_{2,2}(\theta)$ have exactly one branch singularity
each.
We plot $z_{2,1}(\theta)$, $z_{2,2}(\theta)$, 
$z_{3,1}(\theta)$, and $z_{3,2}(\theta)$ in~\cref{fig:singpow}.
\end{rem}

\begin{figure}[h!]
\begin{center}
\includegraphics[width=7cm]{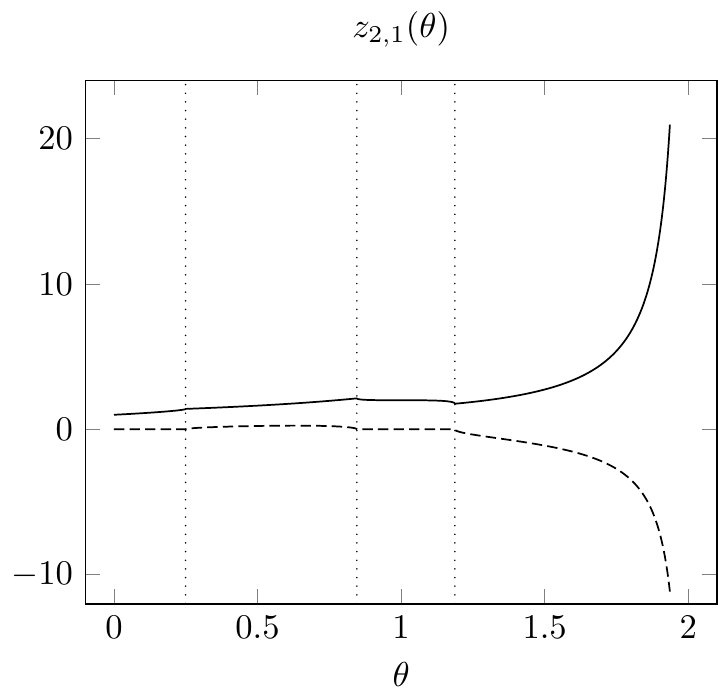} \hspace{0.5cm}
\includegraphics[width=7cm]{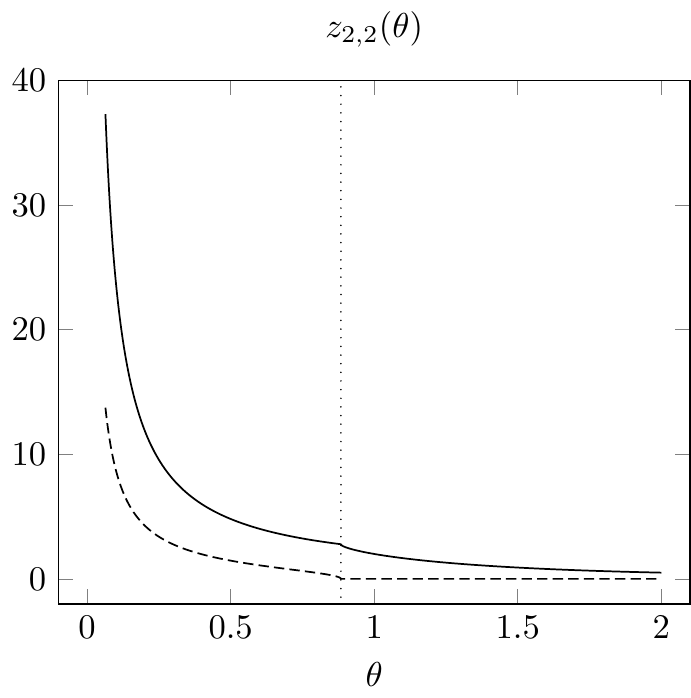}
\includegraphics[width=7cm]{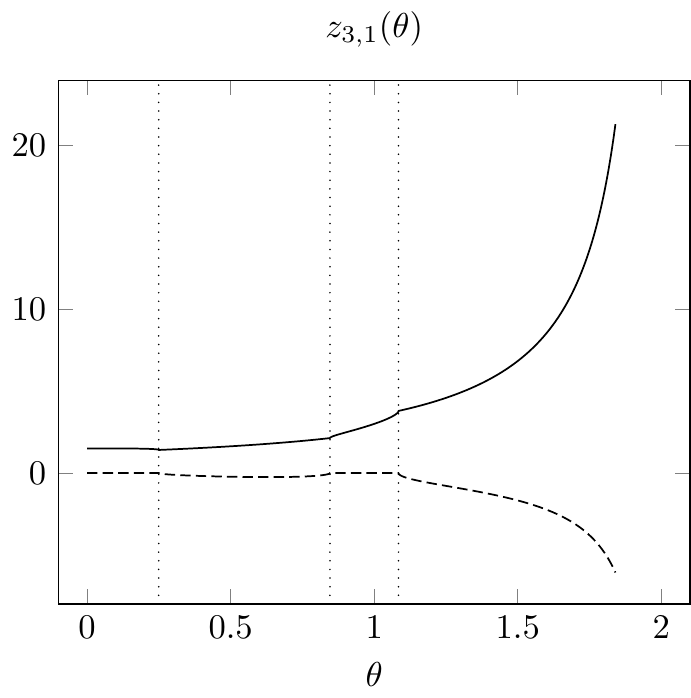} \hspace{0.5cm}
\includegraphics[width=7cm]{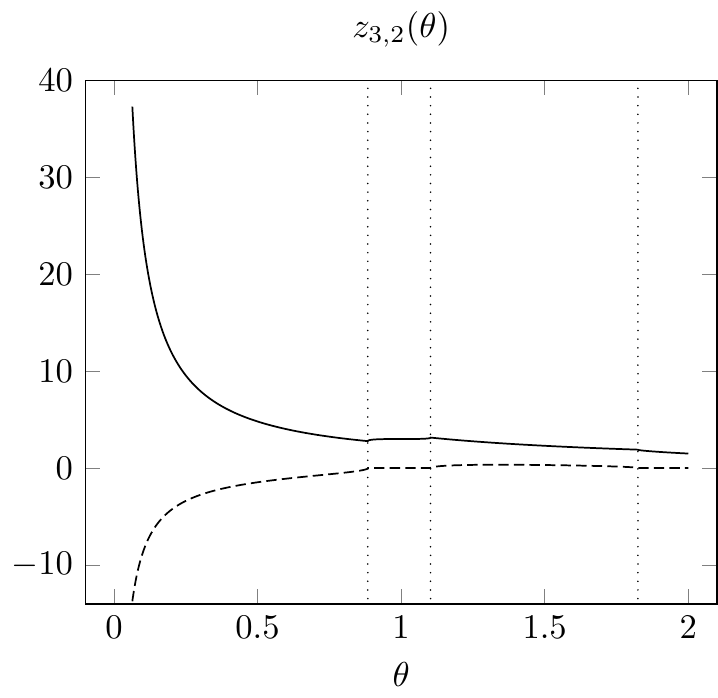}
\caption{ Plots for the real and imaginary parts of the functions
$z_{2,1}(\theta)$ (top left), $z_{2,2}(\theta)$ (top right),
$z_{3,1}(\theta)$ (bottom left), and $z_{3,2}(\theta)$ (bottom right) 
for $0<\theta<2$.
The solid lines represent the real part of $z$, the dashed line represents
the imaginary part of $z$, and the vertical dotted lines indicate the locations
of the branch points.}
\label{fig:singpow}
\end{center}
\end{figure}

Suppose now that, as in~\cref{thm:mainsingpowtone}, $z_{n,1}(\theta)$, $n=1,2,\ldots N$, 
are analytic functions which satisfy $\det{\bA(z_{n,1}(\theta),\theta)} = 0$,
and $z_{n,2}(\theta)$, 
$n=2,3,\ldots N$,
are analytic functions which satisfy $\det{\bA(z_{n,2}(\theta),\theta)} = 0$,
$n=2,3,\ldots N$.
We observed  in~\cref{rem:pqdetermination} that $p_{n,1}(\theta)$ and $q_{n,1}(\theta)$ 
are determined explicitly by $z_{n,1}(\theta)$ 
for $n=1,2\ldots N$,
and, similarly
$p_{n,2}(\theta)$ and $q_{n,2}(\theta)$ 
are determined explicitly by $z_{n,2}(\theta)$ 
for $n=2,3\ldots N$.
We recall that, if 
\beq
\label{eq:defmutone2}
\mu_{\tau,o}(t) = p_{n,j} \cdot |t|^{z_{n,j}}  \,, \quad \text{and} 
\quad \mu_{\nu,e}(t) = q_{n,j} \cdot |t|^{z_{n,j}} \,  ,
\eeq
then the corresponding components of the velocity 
$h_{\tau,o}(t)$ and $h_{\nu,e}(t)$ defined by~\cref{eq:tonecaseeq1} 
satisfy
\beq
\bbmat 
h_{\tau,o}(t)  \\
h_{\nu,e} (t)
\ebmat
=  
\sum_{m=1}^{\infty}
\bF(m,z_{n,j},\theta)
\bbmat
p_{n,j} \\
q_{n,j}
\ebmat 
\cdot t^{m} \, ,
\label{eq:singlesingpowexp2}
\eeq
for $0<t<1$, $n=1,2,\ldots N$ when $j=1$, and $n=2,3\ldots N$ when $j=2$,  
where $\bF$ is defined by~\cref{eq:kerssmoothmat}
(see~\cref{thm:singpowexp}).

We note that the implicit functions $z_{n,2}(\theta)$, 
satisfying~\cref{eq:implfun2main},
are defined for $n\geq 2$, as opposed to the implicit functions
$z_{n,1}(\theta)$, satisfying~\cref{eq:implfun1main}, which are defined
for $n\geq 1$. 
We observe that the function $z_{1,2}(\theta)$, 
defined by $z_{1,2}(\theta) \equiv 1$, satisfies~\cref{eq:implfun2main},
since when $z=1$, 
\beq
z \sint - \sinzt = \sint - \sint = 0 \, ,
\eeq
for all $\theta$.
In the following lemma, we compute the velocity field
when $(\mu_{\tau,o}(t), \mu_{\nu,e}(t)) = (0,1) \cdot t$.

\begin{lem}
\label{lem:densz1}
Suppose that $\theta \in \C$, $\mu_{\tau,o}(t) = 0$ and $\mu_{\nu,e}(t) = t$,
for $0<t<1$.
Then $h_{\tau,o}(t)$ and $h_{\nu,e}(t)$ defined by~\cref{eq:tonecaseeq1} satisfy
\beq
\bbmat 
h_{\tau,o}(t)  \\
h_{\nu,e} (t)
\ebmat
=
\bF_{1}(\theta) 
t
+
\sum_{n=2}^{\infty}
\bF(n,1,\theta)
\bbmat
0 \\
1
\ebmat 
\cdot t^{n} \, ,
\label{eq:singlesingpowexp1}
\eeq
for $0<t<1$, where $\bF$ is defined in~\cref{eq:kerssmoothmat} and
\beq
\label{eq:a11def}
\bF_{1}(\theta) =
-\frac{1}{2\pi}
\bbmat
-\sqsint \\
\pi \theta - \sint \cost
\ebmat \, .
\eeq
\end{lem}

\begin{proof}
Substituting $\mu_{\tau,o}(t) =0$ and $\mu_{\nu,e}(t) = t^{z}$ in~\cref{eq:tonecaseeq1}, 
where $z$ is not an integer, the corresponding components on the boundary 
$h_{\tau,o}(t)$ and $h_{\nu,e}(t)$ using~\cref{thm:int01allkers} are given by
\beq
\bbmat
h_{\tau,o} (t) \\
h_{\nu,e}(t) 
\ebmat
= 
\bbmat
a_{1,2}(z,\theta) \\
-1/2 + a_{2,2}(z,\theta)
\ebmat
t^{z} 
+ \sum_{n=2}^{\infty}
\bbmat
F_{1,2}(n,z,\theta) \\
F_{2,2}(n,z,\theta)
\ebmat
t^{n} \, ,
\eeq
for $0<t<1$ , where $a_{1,2}(z,\theta),a_{2,2}(z,\theta)$ are defined
in~\cref{eq:defsing12,eq:defsing22} respectively, and 
$F_{1,2}(n,z,\theta), F_{2,2}(n,z,\theta)$ are 
defined in~\cref{eq:defsmooth12,eq:defsmooth22} respectively.
The result then follows from taking the limit as $z\to 1$.
\end{proof}

%\begin{rem}
%Since $z(\theta) \equiv 1$ satisfies~\cref{eq:implfun2}, with $z(1)=1$,
%we set $z_{1,2}(\theta) = 1$ to stay consistent with our notation.
%We note that even though $\det{\bA(z_{1,2}(\theta),\theta)}\neq 0$, 
%if $p_{1,2} = 0$, and $q_{1,2} = 1$ and $\mu_{\tau,o}(t),\mu_{\nu,e}(t)$
%are defined by~\cref{eq:defmutone2}, 
%then the corresponding components of the velocity
%defined by~\cref{eq:tonecaseeq1} are smooth functions (see~\cref{lem:densz1}).
%\end{rem}

It is clear from~\cref{eq:singlesingpowexp2,eq:singlesingpowexp1}  
that there is no constant term in the Taylor series of the components
of the velocity $h_{\tau,o}(t)$ and $h_{\nu,e}(t)$.
The following lemma computes the velocity field when the components of the 
density $\mu_{\tau,o}(t)$
and $\mu_{\nu,e}(t)$ are constants.

\begin{lem}
\label{lem:constdens}
Suppose that $\theta \in \C$, $\mu_{\tau,o}(t) = p_{0}$ and 
$\mu_{\nu,e}(t) = q_{0}$, where $p_{0},q_{0}$ are constants.
Then $h_{\tau,o}(t)$ and $h_{\nu,e}(t)$ defined by~\cref{eq:tonecaseeq1} satisfy
\beq
\bbmat 
h_{\tau,o}(t)  \\
h_{\nu,e} (t)
\ebmat
=
\bF_{0}(\theta) 
\bbmat
p_{0} \\
q_{0}
\ebmat +
\sum_{n=1}^{\infty}
\bF(n,0,\theta)
\bbmat
p_{0} \\
q_{0}
\ebmat 
\cdot t^{n} \, ,
\label{eq:singlesingpowexp0}
\eeq
for $0<t<1$, where $\bF$ is defined in~\cref{eq:kerssmoothmat} and
\beq
\label{eq:a00def}
\bF_{0}(\theta) =
-\frac{1}{2\pi}
\bbmat
\pi-\sint + \pi(1-\theta)\cost & 
-\pi(1-\theta)\sint \\
-\pi(1-\theta)\sint &
\pi-\sint - \pi(1-\theta)\cost  
\ebmat \, .
\eeq

\end{lem}
\begin{proof}
The result follows from taking the limit $z\to 0$ in~\cref{eq:velboundarysingpow}.
\end{proof}

In the following theorem, we
describe the matrix $\bB(\theta)$ that maps
the coefficients of the basis functions
$(p_{n,j}|t|^{z_{n,j}}, q_{n,j}|t|^{z_{n,j}})$ 
to the Taylor
expansion coefficients of the corresponding velocity field.

\begin{thm}
\label{lem:multpowexp}
Suppose $N\geq 2$ is an integer.
Suppose further that, as in~\cref{thm:mainsingpowtone},
$\theta_{1}, \theta_{2}, \ldots \theta_{3N-2}$
are real numbers on the interval $(0,2)$, and that
$z_{n,1}(\theta)$, $n=1,2,\ldots N$, are analytic functions
satisfying $\det{\bA(z_{n,1}(\theta),\theta)}=0$
for $\theta \in V \subset \C$, where 
$V$ is a simply connected open set containing 
the strip $D$ with $\text{Re}(\theta) \in (0,2)$ and the
interval $(0,2)\setminus\{ \theta_{j} \}_{j=1}^{3N-2}$.
Similarly, suppose that $z_{n,2}(\theta)$, $n=2,3\ldots N$, are
analytic functions satisfying $\det{\bA(z_{n,2}(\theta),\theta)}=0$ 
for $\theta \in V$. 
Let $(p_{n,1}, q_{n,1}) \in \cN \{ \bA(z_{n,1}(\theta),\theta) \} $,
$n=1,2,\ldots N$, and 
$(p_{n,2},q_{n,2}) \in \cN \{ \bA(z_{n,2}(\theta),\theta) \}$,
$n=2,3,\ldots N$.
Suppose that $z_{1,2}(\theta) \equiv 1$, $p_{1,2} = 0$,
and $q_{1,2} = 1$.
Finally, suppose that 
\begin{align}
\mu_{\tau,o}(t) &= \left(c_{0} + 
\sum_{n=1}^{N} c_{n} p_{n,1} |t|^{z_{n,1}} + d_{n}
p_{n,2} |t|^{z_{n,2}} \right) 
\sgn(t)  \,, \\
\mu_{\nu,e}(t) &= \left(d_{0} +  
\sum_{n=1}^{N} c_{n} q_{n,1} |t|^{z_{n,1}} + d_{n}
q_{n,2} |t|^{z_{n,2}} \right) \, ,
\end{align}
for $-1<t<1$, where $c_{j},d_{j} \in \C$, $j=0,1\ldots N$.
Then
\begin{align}
h_{\tau,o}(t) &= \left( \sum_{n=0}^{N} \alpha_{n} |t|^{n} \right) 
\sgn(t) + O(|t|^{N+1}) \\
h_{\nu,e}(t) &= \left( \sum_{n=0}^{N} \beta_{n} |t|^{n} \right) 
+ O(|t|^{N+1}) \, , 
\end{align}
for $-1<t<1$, where 
\begin{align}
\bbmat
\alpha_{0} \\
\beta_{0} \\ 
%\alpha_{2} \\
%\beta_{2} \\
\vdots \\
\alpha_{N} \\
\beta_{N}
\ebmat
= 
\bB(\theta)
\bbmat
c_{0} \\
d_{0} \\ 
%c_{2} \\
%d_{2} \\
\vdots \\
c_{N} \\
d_{N}
\ebmat \, ,
\label{eq:linsyscoeff}
\end{align}
$\bB(\theta)$ is a $(2N+2) \times (2N+2)$ matrix,  and $\theta \in V$.
The $2\times 2$ block of $\bB(\theta)$ which 
maps $c_{n}, d_{n}$ to $\alpha_{\ell}, \beta_{\ell}$
is given by
\beq
\label{eq:bmatdef}
\bB_{\ell,n} (\theta)  = 
\left[
\begin{array}{c;{2pt/2pt}c}
\bF(\ell,z_{n,1}(\theta),\theta)
\bbmat
p_{n,1}(\theta) \\
q_{n,1}(\theta)
\ebmat &
\bF(\ell,z_{n,2}(\theta),\theta)
\bbmat
p_{n,2}(\theta) \\
q_{n,2}(\theta)
\ebmat
\end{array}
\right] \, ,
\eeq
for $\ell,n=1,2,\ldots N$, where $\bF$ is defined in~\cref{eq:kerssmoothmat}, 
except for the case $\ell=n=1$.
In the case $\ell=n=1$, the matrix $\bB_{1,1}(\theta)$ is given by
\beq
\bB_{1,1}(\theta) = 
\left[
\begin{array}{c;{2pt/2pt}c}
\bF(1,z_{1,1}(\theta),\theta)
\bbmat
p_{1,1}(\theta) \\
q_{1,1}(\theta)
\ebmat &
\bF_{1}(\theta)
\end{array}
\right] \, ,
\eeq
where $\bF$ is defined in~\cref{eq:kerssmoothmat}, and 
$\bF_{1}$ is defined in~\cref{eq:a11def}.
Finally, if either $\ell=0$ or $n=0$, then the matrices $\bB_{\ell,n}(\theta)$
are given by
\begin{align}
\bB_{\ell,0}(\theta) &= 
\bbmat
0 & 0 \\
0 & 0
\ebmat \, , \\
\bB_{0,0}(\theta) &= \bF_{0}(\theta) \, , \\
\bB_{0,n}(\theta) &=
\bF(n,0,\theta) \, , 
\end{align}
for $\ell,n=1,2,\ldots N$, 
where $\bF$ is defined in~\cref{eq:kerssmoothmat}, 
and $\bF_{0}$ is defined in~\cref{eq:a00def}. 
\end{thm}

\begin{proof}
It follows from~\cref{thm:mainsingpowtone} 
that for $\theta \in V$, 
$\det{\bA(z_{n,1}(\theta),\theta)} = 0$,  
$n \in 1,2,\ldots N$, and $\det{\bA(z_{n,2}(\theta),\theta)}=0$,
$n=2,3\ldots N$.
Moreover, $z_{n,1}(\theta)$, $n=1,2\ldots N$, and $z_{n,2}(\theta)$,
$n=2,3\ldots N$, are not integers for $\theta \in V \setminus \{ 1 \}$.
Since $(p_{n,1},q_{n,1}) \in \cN\{ \bA(z_{n,1},\theta) \}$, 
$n=1,2,\ldots N$, and, 
$(p_{n,2},q_{n,2}) \in \cN \{ \bA(z_{n,2},\theta) \}$, 
$n=2,3,\ldots N$, 
we observe that $p_{n,1},q_{n,1}$, and $z_{n,1}$, $n=1,2\ldots N$
and $p_{n,2},q_{n,2}$, and $z_{n,2}$, $n=2,3\ldots N$ 
satisfy the conditions of~\cref{thm:singpowexp} for $\theta \in V \setminus\{ 1\}$
and the corresponding entries of the matrix $\bB(\theta)$ can be derived
from~\cref{eq:singlesingpowexp} in~\cref{thm:singpowexp}.

Furthermore, we observe that the density corresponding to
$p_{1,2},q_{1,2}$ and $z_{1,2}$ satisfy the conditions for~\cref{lem:densz1},
and thus the corresponding entries of the matrix $\bB(\theta)$ can
be derived from~\cref{eq:singlesingpowexp1} in~\cref{lem:densz1}.
Finally, the entries of $\bB(\theta)$ corresponding to 
$\mu_{\tau,o}(t),\mu_{\nu,e}(t) = (c_{0},d_{0})$ can be obtained 
from~\cref{lem:constdens}, with which the result follows for 
$\theta \in V \setminus \{ 1 \}$.

The result for $\theta=1$ 
follows by taking the limit $\theta \to 1$ in~\cref{eq:linsyscoeff} 
and from the observation that the limit
$\bB(\theta)$ as $\theta \to 1$ exists 
(see~\cref{thm:blimapptone} in~\cref{sec:appb}).  
\end{proof}

\subsubsection{Invertibility of $\bB$ in~\cref{eq:linsyscoeff} \label{sec:binv}}
The matrix $\bB(\theta)$ is a mapping from coefficients of the basis functions
$(p_{n,j}|t|^{z_{n,j}}, q_{n,j}|t|^{z_{n,j}})$ 
to the Taylor
expansion coefficients of the corresponding velocity 
field (see~\cref{eq:linsyscoeff}).
In this section, we observe that $\bB(\theta)$ 
is invertible for all $\theta \in (0,2)$ except
for countably many values of $\theta$.
We then use this result to derive a converse of~\cref{lem:multpowexp}.
The principal result of this section is~\cref{thm:mainBinvertibility}.

In the following lemma, we show that $\det{\bB(\theta)}$ is an analytic
for $\theta \in V$ and does not vanish at $\theta=1$.
\begin{lem}
\label{lem:banalytic}
Suppose that $V\subset \C$ 
is the open set defined in~\cref{thm:mainsingpowtone}.
Then $\det{\bB(\theta)}$ is an analytic function for $\theta \in V$ 
with $\det{\bB(1)} \neq 0$.
\end{lem}
\begin{proof}
The functions $z_{n,j}(\theta)$, $n=1,2,\ldots N$, $j=1,2$, are analytic
functions for $\theta \in V$ which do not take on integer values
for $\theta \in V \setminus\{1 \}$.
It follows from~\cref{eq:bmatdef}
that the entries of $\bB(\theta)$ take the form
\beq
\frac{P(\theta)p_{n,j}(\theta) + Q(\theta)q_{n,j}(\theta)}{z_{n,j}(\theta) - \ell} 
\eeq
where $P(\theta),Q(\theta)$ are trignometric polynomials in $\theta$, 
$\ell,n \in \{ 1,2,\ldots N \}$ 
and $j \in \{1,2 \}$.
Thus, the entries of $\bB(\theta)$ are analytic functions of $\theta$ for
$\theta \in V \setminus \{1\}$. 
Furthermore, using~\cref{thm:blimapptone} 
in~\cref{sec:appb}, 
we observe that
$\bB(\theta)$ is analytic at $\theta=1$ as well, since 
\beq
\lim_{\theta \to 1} \bB_{\ell,j} (\theta) = 
\begin{cases}
\bbmat
-1/2 & 0 \\
0 & -1/2 
\ebmat
& \quad \ell=j=0 \vspace*{3pt}\\
\bbmat
0 & -1/2 \\
-1/2 & 0 
\ebmat & \quad  \ell=j=2m \neq 0 \vspace*{3pt}\\
\bbmat
-1/2 & 0 \\
0 & -1/2 
\ebmat & \quad  \ell=j=2m+1 \vspace*{3pt}\\
\bbmat
0 & 0 \\
0 & 0
\ebmat
&\quad \text{otherwise}
\end{cases} \, .
\label{eq:matatone}
\eeq
Thus, $\det{\bB(\theta)}$ is an
analytic function for $\theta \in V$.
Moreover, using~\cref{eq:matatone}, $\det(\bB(1)) \neq 0$.
\end{proof}

In the following lemma, we show that   
$\bB(\theta)$ is invertible for all $\theta\in (0,2)$ 
except for countably many values of $\theta$.

\begin{lem}
\label{lem:alldetBzeros}
Suppose that $N \geq 2$ is an integer.
There exists a countable set $\{\phi_{m} \}_{m=1}^{\infty} \subset 
(0,2)$,
such that $\bB(\theta)$ is an invertible matrix for $\theta \in
(0,2) \setminus \{ \phi_{m} \}_{m=1}^{\infty}$.
Moreover, the limit points of $\{\phi_{m}\}_{m=1}^{\infty}$ 
are a subset of $\{\theta_{j} \}_{j=1}^{3N-2} \cup \{ 0,2 \}$,
where $\theta_{j}$, $j=1,2\ldots 3N-2$, are the branch points of the functions
$z_{n,1}(\theta)$, $n=1,2\ldots N$, and $z_{n,2}(\theta)$, $n=2,3\ldots,N$,
defined in~\cref{thm:mainsingpowtone}.
\end{lem}

\begin{proof}
Suppose that $V$ is as defined in~\cref{thm:mainsingpowtone}.
Recall that the interval $(0,2) \setminus \{ \theta_{j} \}_{j=1}^{3N-2} \subset V$.
Using~\cref{lem:banalytic}, $\det{\bB(\theta)}$ is an analytic function
for $\theta \in V$ and satisfies $\det{\bB(1)} \neq 0$. 
Using standard results in complex analysis, since $\det{\bB(\theta)}$ is 
not identically zero, we conclude that 
the matrix $\bB(\theta)$ 
is invertible for all $\theta \in (0,2)$ except for a
countable set of values of $\theta=\phi_{m}$, $m=1,2,\ldots$.
Moreover, the set of limit points of $\det{\bB(\theta)} = 0$, i.e.
the values of $\theta$ for which $\bB(\theta)$ is not invertible,
is a subset of $\pd V \cap (0,2)$, 
where $\pd V$ is the boundary of the set $V$. 
Clearly, $\pd V \cap (0,2) = \{ \theta_{j} \}_{j=1}^{3N-2} \cup \{0,2\}$.

\end{proof}

\begin{rem}
In~\cref{lem:alldetBzeros}, we show that 
$\det\{\bB(\theta)\}$
has countably many zeros on the interval $(0,2)$.
In fact, it is possible to show that there are finitely many zeros of
$\det\{\bB(\theta)\}$ on the interval $(0,2)$. 
On inspecting the form of the entries of $\bB(\theta)$, we note that
$\det\{ \bB(\theta) \}$ is a linear combination of 
$T(\theta)/P(\theta)$ 
where $T(\theta)$ is a trigonometric polynomial of degree less than or equal to
$N$
and $P(\theta)$ is a finite product of functions $ (z_{j,\ell}(\theta)-k_{j,\ell})$
for some integer $k_{j,\ell}$.
A detailed analysis shows that the functions $z_{j,\ell}(\theta)$
are essentially non-oscillatory for $\theta\in (0,2)$.
Since both the functions $T(\theta)$ and $P(\theta)$ 
are band-limited functions, $\det\{ \bB(\theta) \}$ cannot have infinitely
many zeros for $\theta \in (0,2)$.
\end{rem}

The following theorem is the principal result of this section.
\begin{thm}
\label{thm:mainBinvertibility}
Suppose that $N \geq 2$ is an integer.
Then for each $\theta \in  (0,2)$
except for countably many values, 
there exist $p_{n,j},q_{n,j},z_{n,j} \in \C$, $n=1,2\ldots N$ and $j=1,2$, 
such that the following holds.
Suppose $\alpha_{n},\beta_{n} \in \C$, $n=0,1,\ldots N$, and $h_{\tau,o}(t)$
and $h_{\nu,e}(t)$ are given by
\beq
h_{\tau,o}(t) = \left(\sum_{n=0}^{N} \alpha_{n} |t|^{n} \right) \sgn(t) \, ,
\quad \text{and} \quad
h_{\nu,e}(t) = \left(\sum_{n=0}^{N} \beta_{n} |t|^{n} \right) \, ,
\label{eq:rhsmainthmtone}
\eeq
for $-1<t<1$.
Then there exist unique numbers $c_{n},d_{n} \in \C$, $n=0,1,\ldots N$, such that, 
if $\mu_{\tau,o}(t)$ and $\mu_{\nu,e}(t)$ defined by
\begin{align}
\mu_{\tau,o}(t) &= \left(c_{0}+\sum_{n=1}^{N} c_{n} p_{n,1} |t|^{z_{n,1}}
+ d_{n} p_{n,2} |t|^{z_{n,2}} \right) \sgn(t) \, , \label{eq:finalreptone} \\
\mu_{\nu,e}(t) &= d_{0} + \sum_{n=1}^{N} c_{n} q_{n,1} |t|^{z_{n,1}} + 
d_{n} q_{n,2} |t|^{z_{n,2}} \, , \nonumber
\end{align}
for $-1<t<1$, then $\mu_{\tau,o}(t)$ and $\mu_{\nu,e}(t)$ 
satisfy
\beq
\label{eq:tonecaseeq1final}
\bbmat
h_{\tau,o}(t) \\
h_{\nu,e}(t) 
\ebmat
= 
-\frac{1}{2}
\bbmat
\mu_{\tau,o}(t) \\
\mu_{\nu,e}(t)
\ebmat
+ \int_{0}^{1} 
\bbmat 
k_{1,1}(s,t) & k_{1,2}(s,t) \\
k_{2,1}(s,t) & k_{2,2}(s,t) 
\ebmat
\bbmat
\mu_{\tau,o}(s) \\
\mu_{\nu,e}(s)
\ebmat
\, ds \, ,
\eeq
for $-1<t<1$ with error $O(|t|^{N+1})$, where $k_{j,\ell}(s,t)$, 
$j,\ell=1,2$,
are defined in~\cref{eq:defkers11,eq:defkers12,eq:defkers21,eq:defkers22}.
\end{thm}
\begin{proof}
Suppose $z_{n,1}(\theta)$, $n=1,2,\ldots N$,
and $z_{n,2}(\theta)$,
$n=2,3\ldots N$, are the implicit functions that satisfy
$\det{\bA(z,\theta)} =0$ as defined in~\cref{thm:mainsingpowtone}.
%Furthermore, let $V$ be the common domain of analyticity
%of $z_{n,1}(\theta)$, $n=1,2,\ldots N$, and
%$z_{n,2}(\theta)$, $n=2,3,\ldots N$, as defined in~\cref{thm:mainsingpowtone}.
Let $(p_{n,1},q_{n,1}) \in \cN\{\bA(z_{n,1}(\theta),\theta)\}$,
$n=1,2\ldots N$, and 
$(p_{n,2},q_{n,2}) \in \cN \{ \bA(z_{n,2}(\theta),\theta) \}$.
Let $z_{1,2}=1$, $p_{1,2}=0$, and $q_{1,2}=1$.
Given $p_{n,j},q_{n,j}$, and $z_{n,j}$, $n=1,2\ldots N$ and
$j=1,2$, let $\bB(\theta)$ be the $(2N+2) \times (2N+2)$ matrix
defined in~\cref{lem:multpowexp}.
Suppose further that $\{ \phi_{m} \}_{m=1}^{\infty} \subset (0,2)$
are the values of $\theta$ for which $\bB(\theta)$ is not invertible 
(see~\cref{lem:alldetBzeros}).
Finally, since $\bB(\theta)$ is invertible for all 
$\theta \in (0,2) \setminus \{ \phi_{m} \}_{m=1}^{\infty}$, let
\beq
\bbmat
c_{0} \\
d_{0} \\ 
\vdots \\
c_{N} \\
d_{N}
\ebmat
= 
\bB^{-1}(\theta)
\bbmat
\alpha_{0} \\
\beta_{0} \\ 
\vdots \\
\alpha_{N} \\
\beta_{N}
\ebmat \, .
\label{eq:linsyscoeffinvtone}
\eeq
The result then follows from using~\cref{lem:multpowexp}.
\end{proof}

\begin{rem}
\label{rem:infosc2}
In~\cref{rem:infosc}, we noted that the components of the density
$\bmu$ can be expressed in terms of functions of the form
$|t|^{\alpha} \cos{(\beta \log{|t|})}$ and  
$|t|^{\alpha} \sin{(\beta \log{|t|})}$, as
opposed to $|t|^{z}$, where $z = \alpha + i \beta$.
The precise statement is as follows. 
We observe that, when $\theta$ is real, if $\det{\bA(z,\theta) }=0$, then
$\det{\bA(\bar{z},\theta)} =0$.
Moreover, if $(p,q) \in \cN\{ \bA(z,\theta) \}$, 
then $(\bar{p},\bar{q}) \in \cN\{\bA(\bar{z},\theta)\}$.
Thus, the numbers $p,q$, and $z$ occur in complex conjugates.
Results analogous to~\cref{lem:multpowexp} and~\cref{thm:mainBinvertibility}
can be derived for the case when the components of $\bmu$
are given by
\begin{align}
\mu_{\tau,o}(t) = \bigg( c_{0} &+ \sum_{n=1}^{N} 
c_{n} \left( r_{n} |t|^{\alpha_{n}} \cos{(\beta_{n} \log{|t|})} 
- s_{n} |t|^{\alpha_{n}} \sin{(\beta_{n} \log{|t|})}
\right) \label{eq:mureal1} \\ 
+ 
& \sum_{n=1}^{N} d_{n} \left( s_{n} |t|^{\alpha_{n}} \cos{(\beta_{n} \log{|t|})} 
+ r_{n} |t|^{\alpha_{n}} \sin{(\beta_{n} \log{|t|})}
\right) \bigg) \sgn (t) \nonumber \\
\mu_{\nu,e}(t) = \bigg( d_{0} &+ \sum_{n=1}^{N} 
c_{n} \left( v_{n} |t|^{\alpha_{n}} \cos{(\beta_{n} \log{|t|})} 
- w_{n} |t|^{\alpha_{n}} \sin{(\beta_{n} \log{|t|})}
\right) \label{eq:mureal2} \\ 
+ 
& \sum_{n=1}^{N} d_{n} \left( w_{n} |t|^{\alpha_{n}} \cos{(\beta_{n} \log{|t|})} 
+ v_{n} |t|^{\alpha_{n}} \sin{(\beta_{n} \log{|t|})}
\right) \bigg) \, , \nonumber
\end{align}
for $-1<t<1$, where $z_{n} = \alpha_{n} + i \beta_{n}$, $p_{n} = r_{n} + is_{n}$,
and $q_{n} = v_{n} + iw_{n}$.
The advantage for using the 
representation~\cref{eq:mureal1,eq:mureal2}
for the density $\bmu$ is the following.  
If the components of the velocity
$h_{\tau,o}(t)$ and $h_{\nu,e}(t)$ are real, then 
the solution $c_{n},d_{n}$, $n=0,1\ldots N$ when $\bmu$ defined
by~\cref{eq:mureal1,eq:mureal2} which satisfies~\cref{eq:tonecaseeq1} to
order $N$ accuracy, is also real.
\end{rem}

%\begin{rem}
%Instead of taking $N$ singular functions $t^{z_{n,1}}(\theta)$ 
%and $t^{z_{n,2}}(\theta)$, if instead one took
%$N_{1}$ functions from the set $t^{z_{n,1}}(\theta)$ and
%and $N_{2}$ functions from the set $t_{z_{n,2}}(\theta)$ 
%where $N_{1} + N_{2} = 2N$, 
%then the corresponding linear mapping $\bB(\theta)$ is 
%invertible for all $\theta \in (0,2)$. 
%The proof is similar to the proof discussed in~\ref{kircornersii}, 
%which relies on proving that the system of functions defined by 
%the columns of $\bB_{\theta}$ form a Chebyshev system.
%The details of the proof are cumbersome and will be published at
%a later date.
%\end{rem}

\subsection{Tangential even, normal odd case \label{sec:teno}}
In this section, we investigate tangential even, normal odd case 
(see equation~\cref{eq:teno}).
In~\cref{sec:singpownote}, we investigate the values of $p$, $q$ and
$z$ in~\cref{eq:dftone} for which the resulting components
of the velocity are smooth functions.
In~\cref{sec:binvnote}, we show that, for every $\bh$ of the 
form~\cref{eq:defhnote}, 
there exists a density $\bmu$ 
of the form~\cref{eq:defmu1note,eq:defmu2note}, 
which satisfies the integral equation~\cref{eq:teno} to order $N$.
The proofs of the results in this section
are essentially identical to the corresponding proofs in~\cref{sec:tone}.
For brevity, we present the statements of the theorems without proof.

\subsubsection{The values of $p_{n,j}$, $q_{n,j}$, and $z_{n,j}$ in~\cref{eq:dfnote} 
\label{sec:singpownote}}
Suppose that $\mu_{\tau,e}(t)$ and $\mu_{\nu,o}(t)$ are given by 
\beq
\mu_{\tau,e}(t)=p \cdot  |t|^{z}  \, ,\quad 
\text{and} \quad \mu_{\nu,o}(t) = q \cdot |t|^{z} \sgn(t) \, , 
\label{eq:mudefnote}
\eeq
for $-1<t<1$, where $p,q,z \in \C$.
In this section, we determine the values of $p,q$ and $z$ such that 
$h_{\tau,e}(t)$ and $h_{\nu,o}(t)$ defined by
\beq
\label{eq:notecaseeq1}
\bbmat
h_{\tau,e}(t) \\
h_{\nu,o}(t) 
\ebmat
= 
-\frac{1}{2}
\bbmat
\mu_{\tau,e}(t) \\
\mu_{\nu,o}(t)
\ebmat
- \int_{0}^{1} 
\bbmat 
k_{1,1}(s,t) & k_{1,2}(s,t) \\
k_{2,1}(s,t) & k_{2,2}(s,t) 
\ebmat
\bbmat
\mu_{\tau,e}(s) \\
\mu_{\nu,o}(s)
\ebmat
\, ds \, ,
\eeq
are smooth functions of $t$ for $0<t<1$.
The principal result of this section is~\cref{lem:multpowexpnote}.

The following lemma describes sufficient conditions for
$p,q$, and $z$ such that, if $\bmu$ is
defined by~\cref{eq:mudefnote}, then the velocity
$\bh$ given by~\cref{eq:notecaseeq1} is smooth. 
\begin{thm}
\label{thm:singpowexpnote}
Suppose that $\theta \in (0,2)$, $z$ is not an integer, and $z$ satisfies 
$\det{\bA(z,\theta)} = 0$
where 
\beq
\label{eq:defamatnote}
\bA(z,\theta) = -\frac{1}{2}\bI - \bbmat
a_{1,1}(z,\theta) & a_{1,2}(z,\theta) \\
a_{2,1}(z,\theta) & a_{2,2}(z,\theta)   
\ebmat
\, ,
\eeq
$\bI$ is the $2\times2$ identity matrix and $a_{j,\ell}(z,\theta)$,
$j,\ell=1,2$,
are given by~\cref{eq:defsing11,eq:defsing12,eq:defsing21,eq:defsing22}.
Furthermore, suppose that $(p,q) \in \cN\{\bA(z,\theta)\}$, 
where $\cN\{\bA \}$ denotes the null space of the matrix $\bA$.
Suppose finally that
\beq
\mu_{\tau,e}(t) = p \cdot t^{z} \, , \quad \text{and}
\quad \mu_{\nu,o}(t) = q \cdot t^{z} \, ,
\eeq
for $0<t<1$. 
Then $h_{\tau,e}(t)$ and $h_{\nu,o}(t)$ defined by~\cref{eq:notecaseeq1} satisfy
\beq
\bbmat 
h_{\tau,o}(t)  \\
h_{\nu,e} (t)
\ebmat
=  
-\sum_{n=1}^{\infty}
\bF(n,z,\theta)
\bbmat
p \\
q
\ebmat 
\cdot t^{n} \, ,
\label{eq:singlesingpowexpnote}
\eeq
for $0<t<1$, where
\beq
\label{eq:kerssmoothmatnote}
\bF(n,z,\theta) = 
\bbmat
F_{1,1}(n,z,\theta) & F_{1,2}(n,z,\theta) \\
F_{2,1}(n,z,\theta) & F_{2,2}(n,z,\theta)
\ebmat \, ,
\eeq
and $F_{j,\ell}(n,z,\theta)$, $j,\ell=1,2$, are given
by~\cref{eq:defsmooth11,eq:defsmooth12,eq:defsmooth21,eq:defsmooth22}.
\end{thm}

A straightforward calculation shows that
\beq
\label{eq:detvalnote}
\det{\bA(z,\theta)} = 
\frac{\left( z \sint + \sinzt \right) \left( z \sint + \sinztc \right)}{4 
\sqsinz}
\eeq
Thus, if $z$ is not an integer and
either
\beq
\label{eq:implfuns00note}
z\sint + \sinztc = 0 \, , 
\eeq
or 
\beq
z \sint + \sinzt = 0 \, , 
\label{eq:implfuns01note}
\eeq
then $\det{\bA(z,,\theta)}=0$.

In the following theorem,
we prove the existence of the implicit
functions $z(\theta)$  
defined by~\cref{eq:implfuns00note,eq:implfuns01note}
on the interval $(0,2)$.

\begin{thm}
\label{thm:mainsingpownote}
Suppose that $N\geq 2$ is an integer.
Then there exists $3N-2$ real numbers 
$\theta_{1},\theta_{2},\ldots \theta_{3N-2} \in (0,2)$
such that the following holds.
Suppose that $D$ is the strip in the upper half plane 
with $0<\text{Re}(\theta)<2$ that includes the
interval $(0,2)\setminus \{ \theta_{j} \}_{j=1}^{3N-2}$, i.e.
\beq
\label{eq:defDmaintone}
D = \{ \theta \in \C \, : \, \text{Re}(\theta) \in (0,2) \, , 
\quad 0 \leq \text{Im}(\theta) < \infty \} \setminus \{ \theta_{j} \}_{j=1}^{3N-2} \, . 
\eeq
Then, there exists a simply connected open set $D \subset V \subset \C$
and analytic
functions 
$z_{n,1}(\theta):V \to \C$, $n=2,3\ldots N$, which satisfy
\beq
z\sint + \sinztc = 0 \, , \quad z(1) =  n \, , \label{eq:implfun3main}
\eeq
for $\theta \in V$, and analytic functions 
$z_{n,2}(\theta):V \to \C$, $n=1,2,\ldots N$, which satisfy
\beq
z\sint + \sinzt = 0 \, , \quad z(1) =  n \, , \label{eq:implfun4main}
\eeq
for $\theta \in V$ (see~\cref{fig:vdommain} for an illustrative domain $V$).
Moreover, the functions 
$z_{n,1}(\theta)$, $n=2,3\ldots N$, do not take integer values for
all $\theta \in V \setminus \{ 1\}$  
and satisfy
$\det{\bA(z_{n,1}(\theta),\theta)}=0$, $n=2,3\ldots N$, 
for all $\theta \in V$
(see~\cref{eq:defamatnote,eq:detvalnote}).
Similarly, the functions
$z_{n,2}(\theta)$, $n=1,2,\ldots N$, do not take integer values
for all $\theta \in V \setminus \{ 1 \}$ and satisfy
$\det{\bA(z_{n,2}(\theta),\theta)}=0$, $n=1,2\ldots N$, 
for all $\theta \in V$
(see~\cref{eq:defamatnote,eq:detvalnote}).
\end{thm}

We note that the implicit functions $z_{n,1}(\theta)$, 
satisfying~\cref{eq:implfun3main},
are defined for $n\geq 2$, as opposed to, the implicit functions
$z_{n,2}(\theta)$, satisfying~\cref{eq:implfun4main}, which are defined
for $n\geq 1$.
We observe that the function $z_{1,1}(\theta)$ 
defined by $z_{1,1}(\theta) \equiv 1$, satisfies~\cref{eq:implfun3main}, 
since when $z=1$, 
\beq
z \sint + \sinztc = \sint - \sint = 0 \, ,
\eeq
for all $\theta$.
In the following lemma, we compute the velocity field
when $(\mu_{\tau,e}(t), \mu_{\nu,o}(t)) = (0,1) t$.

\begin{lem}
\label{lem:densz1note}
Suppose that $\theta \in \C$, $\mu_{\tau,e}(t) = 0$ and $\mu_{\nu,o}(t) = t$,
for $0<t<1$.
Then $h_{\tau,e}(t)$ and $h_{\nu,o}(t)$ defined by~\cref{eq:notecaseeq1} satisfy
\beq
\bbmat 
h_{\tau,e}(t)  \\
h_{\nu,o} (t)
\ebmat
=
\bF_{1}(\theta) 
t
-
\sum_{n=2}^{\infty}
\bF(n,1,\theta)
\bbmat
0 \\
1
\ebmat 
\cdot t^{n} \, ,
\label{eq:singlesingpowexp1note}
\eeq
for $0<t<1$, where $\bF$ is defined in~\cref{eq:kerssmoothmat} and
\boxit{
\beq
\label{eq:a11defnote}
\bF_{1}(\theta) =
-\frac{1}{2\pi}
\bbmat
\sqsint \\
\pi (2-\theta) + \sint \cost
\ebmat \, .
\eeq
}
\end{lem}

It is clear from~\cref{eq:singlesingpowexpnote,eq:singlesingpowexp1note}  
that there is no constant term in the Taylor series of the components
of the velocity $h_{\tau,e}$ and $h_{\nu,o}$.
The following lemma computes the velocity field when the components of the 
density $\mu_{\tau,e}$
and $\mu_{\nu,o}$ are constants.

\begin{lem}
\label{lem:constdensnote}
Suppose that $\theta \in \C$, 
$\mu_{\tau,e}(t) = p_{0}$ and $\mu_{\nu,o}(t) = q_{0}$, 
where $p_{0},q_{0}$ are constants.
Then $h_{\tau,e}(t)$ and $h_{\nu,o}(t)$ defined by~\cref{eq:notecaseeq1} satisfy
\beq
\bbmat 
h_{\tau,e}(t)  \\
h_{\nu,o} (t)
\ebmat
=
\bF_{0}(\theta) 
\bbmat
p_{0} \\
q_{0}
\ebmat -
\sum_{n=1}^{\infty}
\bF(n,0,\theta)
\bbmat
p_{0} \\
q_{0}
\ebmat 
\cdot t^{n} \, ,
\label{eq:singlesingpowexp0note}
\eeq
for $0<t<1$, where $\bF$ is defined in~\cref{eq:kerssmoothmatnote} and
\boxit{
\beq
\label{eq:a00defnote}
\bF_{0}(\theta) =
-\frac{1}{2\pi}
\bbmat
\pi+\sint - \pi(1-\theta)\cost & 
\pi(1-\theta)\sint \\
\pi(1-\theta)\sint &
\pi+\sint + \pi(1-\theta)\cost  
\ebmat \, .
\eeq
}
\end{lem}

In the following theorem, we
describe the matrix $\bB(\theta)$ that maps
the coefficients of the basis functions
$(p_{n,j}|t|^{z_{n,j}}, q_{n,j}|t|^{z_{n,j}})$ 
to the Taylor
expansion coefficients of the corresponding velocity field.

\begin{thm}
\label{lem:multpowexpnote}
Suppose $N\geq 2$ is an integer.
Suppose further that, as in~\cref{thm:mainsingpownote},
$\theta_{1},\theta_{2}, \ldots \theta_{3N-2}$ 
are real numbers on the interval $(0,2)$ 
$z_{n,1}(\theta)$, $n=2,3,\ldots N$, are analytic functions
satisfying $\det{\bA(z_{n,1}(\theta),\theta)}=0$
for $\theta \in V \subset \C$, where
V is a simply connected open set containing 
the strip $D$ with $\text{Re}(\theta) \in (0,2)$ and the
interval $(0,2)\setminus\{ \theta_{j} \}_{j=1}^{3N-2}$.
Similarly, suppose that $z_{n,2}(\theta)$, $n=1,2,\ldots N$, are
analytic functions satisfying $\det{\bA(z_{n,2}(\theta),\theta)}=0$ 
for $\theta \in V$. 
Let $(p_{n,1}, q_{n,1}) \in \cN \{ \bA(z_{n,1}(\theta),\theta) \} $,
$n=2,3,\ldots N$, and 
$(p_{n,2},q_{n,2}) \in \cN \{ \bA(z_{n,2}(\theta),\theta) \}$,
$n=1,2,\ldots N$.
Suppose that $z_{1,1}(\theta) \equiv 1$, $p_{1,1} = 0$,
and $q_{1,1} = 1$.
Finally, suppose that 
\begin{align}
\mu_{\tau,e}(t) &= \left(c_{0} + 
\sum_{n=1}^{N} c_{n} p_{n,1} |t|^{z_{n,1}} + d_{n}
p_{n,2} |t|^{z_{n,2}} \right) \, ,\\
\mu_{\nu,o}(t) &= \left(d_{0} +  
\sum_{n=1}^{N} c_{n} q_{n,1} |t|^{z_{n,1}} + d_{n}
q_{n,2} |t|^{z_{n,2}} \right) 
\sgn(t)  \,, \\
\end{align}
for $-1<t<1$, where $c_{j},d_{j} \in \C$, $j=0,1\ldots N$.
Then
\begin{align}
h_{\tau,o}(t) &= \left( \sum_{n=0}^{N} \alpha_{n} |t|^{n} \right) 
 + O(|t|^{N+1}) \\
h_{\nu,e}(t) &= \left( \sum_{n=0}^{N} \beta_{n} |t|^{n} \right) \sgn(t) 
+ O(|t|^{N+1}) \, , 
\end{align}
for $-1<t<1$, where 
\begin{align}
\bbmat
\alpha_{0} \\
\beta_{0} \\ 
%\alpha_{2} \\
%\beta_{2} \\
\vdots \\
\alpha_{N} \\
\beta_{N}
\ebmat
= 
\bB(\theta)
\bbmat
c_{0} \\
d_{0} \\ 
%c_{2} \\
%d_{2} \\
\vdots \\
c_{N} \\
d_{N}
\ebmat \, ,
\label{eq:linsyscoeffnote}
\end{align}
$\bB(\theta)$ is a $(2N+2) \times (2N+2)$ matrix,  and $\theta \in V$.
The $2\times 2$ block of $\bB(\theta)$ which 
maps $c_{n}, d_{n}$ to $\alpha_{\ell}, \beta_{\ell}$
is given by
\beq
\label{eq:bmatdefnote}
\bB_{\ell,n} (\theta)  = 
-\left[
\begin{array}{c;{2pt/2pt}c}
\bF(\ell,z_{n,1}(\theta),\theta)
\bbmat
p_{n,1}(\theta) \\
q_{n,1}(\theta)
\ebmat &
\bF(\ell,z_{n,2}(\theta),\theta)
\bbmat
p_{n,2}(\theta) \\
q_{n,2}(\theta)
\ebmat
\end{array}
\right] \, ,
\eeq
for $\ell,n=1,2,\ldots N$, where $\bF$ is defined in~\cref{eq:kerssmoothmat}, 
except for the case $\ell=n=1$.
In the case $\ell=n=1$, the matrix $\bB_{1,1}(\theta)$ is given by
\beq
\bB_{1,1}(\theta) = 
\left[
\begin{array}{c;{2pt/2pt}c}
\bF_{1}(\theta) &
-\bF(1,z_{1,2}(\theta),\theta)
\bbmat
p_{1,2}(\theta) \\
q_{1,2}(\theta)
\ebmat 
\end{array}
\right] \, ,
\eeq
where $\bF$ is defined in~\cref{eq:kerssmoothmat}, and 
$\bF_{1}$ is defined in~\cref{eq:a11defnote}.
Finally, if either $\ell=0$ or $n=0$, then the matrices $\bB_{\ell,n}(\theta)$
are given by
\begin{align}
\bB_{\ell,0}(\theta) &= 
\bbmat
0 & 0 \\
0 & 0
\ebmat \, , \\
\bB_{0,0}(\theta) &= \bF_{0}(\theta) \, , \\
\bB_{0,n}(\theta) &=
-\bF(n,0,\theta) \, , 
\end{align}
for $\ell,n=1,2,\ldots N$, 
where $\bF$ is defined in~\cref{eq:kerssmoothmat}, 
and $\bF_{0}$ is defined in~\cref{eq:a00defnote}. 
\end{thm}

\subsubsection{Invertibility of $\bB$ in~\cref{eq:linsyscoeffnote} \label{sec:binvnote}}
The matrix $\bB(\theta)$ is a mapping from coefficients of the basis functions
$(p_{n,j}|t|^{z_{n,j}}, q_{n,j}|t|^{z_{n,j}})$, 
to the corresponding Taylor
expansion coefficients of the velocity field on the boundary $\bh$.
In this section, we observe that $\bB(\theta)$ 
is invertible for all $\theta \in (0,2)$ except
for countably many values of $\theta$.
We then use this result to derive a converse of~\cref{lem:multpowexpnote}.
The following theorem is the principal result of this section.
\begin{thm}
\label{thm:mainBinvertibilitynote}
Suppose that $N\geq 2$ is an integer.
Then for each $\theta \in  (0,2)$
except for countably many values, 
there exist $p_{n,j},q_{n,j},z_{n,j} \in \C$, $n=1,2\ldots N$ and $j=1,2$, 
such that the following holds.
Suppose $\alpha_{n},\beta_{n} \in \C$, $n=0,1,\ldots N$, and $h_{\tau,e}(t)$
and $h_{\nu,o}(t)$ are given by
\beq
h_{\tau,e}(t) = \left(\sum_{n=0}^{N} \alpha_{n} |t|^{n} \right) \, , 
\quad \text{and} \quad
h_{\nu,e}(t) = \left(\sum_{n=0}^{N} \beta_{n} |t|^{n} \right) \sgn(t) \, ,
\label{eq:rhsmainthmnote}
\eeq
for $-1<t<1$.
Then there exist unique numbers $c_{n},d_{n} \in \C$, $n=0,1,\ldots N$, such that, 
if $\mu_{\tau,e}(t)$ and $\mu_{\nu,o}(t)$ defined by
\begin{align}
\mu_{\tau,o}(t) &= \left(c_{0}+\sum_{n=1}^{N} c_{n} p_{n,1} |t|^{z_{n,1}}
+ d_{n} p_{n,2} |t|^{z_{n,2}} \right)  \, , \label{eq:finalrepnote} \\
\mu_{\nu,e}(t) &= \left(d_{0} + \sum_{n=1}^{N} c_{n} q_{n,1} |t|^{z_{n,1}} + 
d_{n} q_{n,2} |t|^{z_{n,2}} \right) \sgn(t) \, , \nonumber
\end{align}
for $-1<t<1$, then $\mu_{\tau,e}(t)$ and $\mu_{\nu,o}(t)$ 
satisfy
\beq
\label{eq:notecaseeq1final}
\bbmat
h_{\tau,e}(t) \\
h_{\nu,o}(t) 
\ebmat
= 
-\frac{1}{2}
\bbmat
\mu_{\tau,e}(t) \\
\mu_{\nu,o}(t)
\ebmat
- \int_{0}^{1} 
\bbmat 
k_{1,1}(s,t) & k_{1,2}(s,t) \\
k_{2,1}(s,t) & k_{2,2}(s,t) 
\ebmat
\bbmat
\mu_{\tau,e}(s) \\
\mu_{\nu,o}(s)
\ebmat
\, ds \, ,
\eeq
for $-1<t<1$ with error $O(|t|^{N+1})$, where $k_{j,\ell}$, $j,\ell=1,2$, 
are defined in~\cref{eq:defkers11,eq:defkers12,eq:defkers21,eq:defkers22}.
\end{thm}

%%%%%%%%%%%%%%%%%%%%%%%%%%%%%%%%%%%%%%%%%%%%%%%%%%
\section{Numerical Results \label{sec:numres}}
%To solve the integral equations~\cref{eq:inteqStokestn2} on polygonal
%domains, we use the representations~\cref{eq:finalreptone} 
%and~\cref{eq:finalrepnote} to construct
%purpose made discretizations using the following procedure.

To solve the integral equation~\cref{eq:inteqStokestn2} 
on polygonal domains, there are two general approaches to incorporating 
the representations~\cref{eq:finalreptone,eq:finalrepnote} 
into a numerical algorithm: Galerkin
methods, in which the solution is represented directly in terms of
coefficients of these functions; and Nystr\"om methods, in which the
solution is represented by its values at certain specially chosen
discretization nodes. 
For efficiency (in order to avoid computing the double
integrals required by Galerkin methods) we will consider only the Nystr\"om
formulation---we note however that the following approach can also be
reformulated as a Galerkin method.

The accuracy and order of a Nystr\"om scheme is equal to the accuracy 
and order of the underlying discretization and quadrature schemes. Any
Nystr\"om scheme consists of the following two components.

\begin{itemize}
\item First, it must provide a discretization of the solution $\bmu$ on the
boundary $\Gamma$, so that $\bmu$ is represented, to the desired
precision, by its values at a collection of nodes 
$\{ \bx_i \}_{i=1}^{N_{d}} \in \Gamma$.
The integral equation~\cref{eq:inteqStokestn2} is thus reduced to the system of
equations
\beq
-\frac{1}{2}
\bbmat
\mu_{\tau}(\bx_{i}) \\
\mu_{\nu} (\bx_{i})
\ebmat
+ 
\text{p.v.} \int_{\Gamma} \bK(\bx_{i},\by) \bbmat
\mu_{\tau}(\by) \\
\mu_{\nu}(\by) 
\ebmat
\, dS_{\by}
= 
\bbmat
h_{\tau}(\bx_{i}) \\
h_{\nu}(\bx_{i})
\ebmat \, ,
\, \,
\label{eq:inteqStokestn3}
\eeq
$i=1,2,\ldots, N_{d}$.
\item Next, it must provide a quadrature rule for each integral of the form
\beq
\text{p.v.} \int_{\Gamma} \bK(\bx_{i},\by) \bbmat
\mu_{\tau}(\by) \\
\mu_{\nu}(\by) 
\ebmat
\, dS_{\by}
\eeq
for each $i=1,2,\ldots,N_{d}$.
\end{itemize}

We subdivide the boundary $\Gamma$ into a collection of panels
$\Gamma_{j}$, $j=1,2,\ldots M$, where panels which meet at a corner
are of equal length. 
For panels away from the corners, the solution $\bmu$ is smooth and, 
thus, we use Gauss-Legendre nodes for discretizing $\bmu$ on those panels.
For panels at a corner, 
we construct special purpose discretization nodes
which allow stable interpolation for the functions $\{ t^{z_{n,j}} \}$,
$n=1,2,\ldots N$, $j=1,2$, $0<t<1$,
where $z_{n,j}$ defined in~\cref{eq:finalreptone,eq:finalrepnote} 
are the values associated with the angle $\pi \theta$ (the angle subtended
at the corner). 
These discretization nodes are readily obtained using the procedure
discussed in~\cite{interpmartinsson}.
Briefly stated, the method constructs an orthogonal basis for the span
of functions $\{ t^{z_{n,j}} \}$ using pivoted Gram-Schmidt algorithm, 
and uses the interpolative decomposition to obtain the discretization 
nodes.

Given the discretization nodes $\{\bx_{i} \}$, $i=1,2,\ldots N_{d}$, 
the quadrature rules for the products
$\bK(\bx_{i},\by) \bmu(\by)$ can be obtained in a similar manner.
For panels away from the corner, both the kernel $\bK(\bx_{i},\by)$, 
for each $\bx_{i}$, and the density $\bmu(\by)$ are smooth functions of $\by$;
thus we use the Gauss-Legendre quadratures for those panels.
For panels at a corner, 
we use the procedure
outlined in~\cite{ggrokhlinma,ggrokhlinyarvin}, to construct
``generalized Gaussian quadratures'' for the functions
$\bK(\bx_{i},\bgamma(t)) t^{z_{n,j}}$, $n=1,2,\ldots N$, $j=1,2$, 
$0<t<1$, where 
$\bgamma(t): [-c,c] \to \R^{2}$ is the parameterization of the corner
with angle $\pi \theta$, and 
$z_{n,j}$ are the values defined in~\cref{eq:finalreptone,eq:finalrepnote}
associated with the angle $\pi \theta$. 
A detailed description of the procedure will be published at a later date.

\begin{rem}
The order of convergence of the method (like any Nystr\"{o}m scheme) 
is dependent on the
order of Gauss-Ledgendre panels used on the smooth panels and
the number of basis functions $N$ used for the density $\bmu$ in
the vicinity of each corner. 
\end{rem}

We illustrate the performance of the algorithm with several numerical
examples. 
The interior velocity boundary problem was solved on each of the domains
in~\cref{fig:cone,fig:equitri,fig:righttri,fig:star,fig:tank},
where the boundary data is generated by five Stokeslets located outside
the respective domains. We then compute the error $E$ given by
\beq
E = \sqrt{\frac{\sum_{m=1}^{5} |\bucomp(\bt_{m}) - \buex(\bt_{m})|^2}
{\sum_{m=1}^{5} |\buex(\bt_{m})|^2}} \, ,
\eeq
where $\bt_{m}$ are targets in the interior of the domain, 
$\bucomp(\bt)$ is velocity computed numerically using the algorithm,
and $\buex(\bt)$ is the exact velocity at the target $\bt$.
We plot the spectrum for the discretized linear systems corresponding to 
the associated integral equations 
in~\cref{fig:cone,fig:equitri,fig:righttri,fig:star,fig:tank}.
In~\cref{tab:res}, we report
the number of discretization nodes $n$, the condition number 
of the discrete linear system $\kappa$, and the error $E$,
for each domain.

\begin{table}[h!]
\begin{center}
\begin{tabular}{cccc}
\hline
 & $n$ & $E$ & $\kappa$ \\ \hline
$\Gamma_{1}$ & $220$ & $3.0 \times 10^{-15}$ & $4.5 \times 10^{2}$ \\ \hline
$\Gamma_{2}$ & $285$ & $2.1 \times 10^{-14}$ & $2.9 \times 10^{5}$ \\ \hline
$\Gamma_{3}$ & $489$ & $4.8 \times 10^{-14}$ & $8.7 \times 10^{5}$ \\ \hline
$\Gamma_{4}$ & $968$ & $1.1 \times 10^{-12}$ & $7.9 \times 10^{5}$ \\ \hline
$\Gamma_{5}$ & $1343$ & $6.7 \times 10^{-13}$ & $6.1 \times 10^{6}$ \\ \hline 
\end{tabular}
\caption{Condition number and error for polygonal domains $\Gamma_{1}$ -- $\Gamma_{5}$}
\label{tab:res}
\end{center}
\end{table}

\begin{figure}[h!]
\begin{center}
\includegraphics[width=4.5cm] {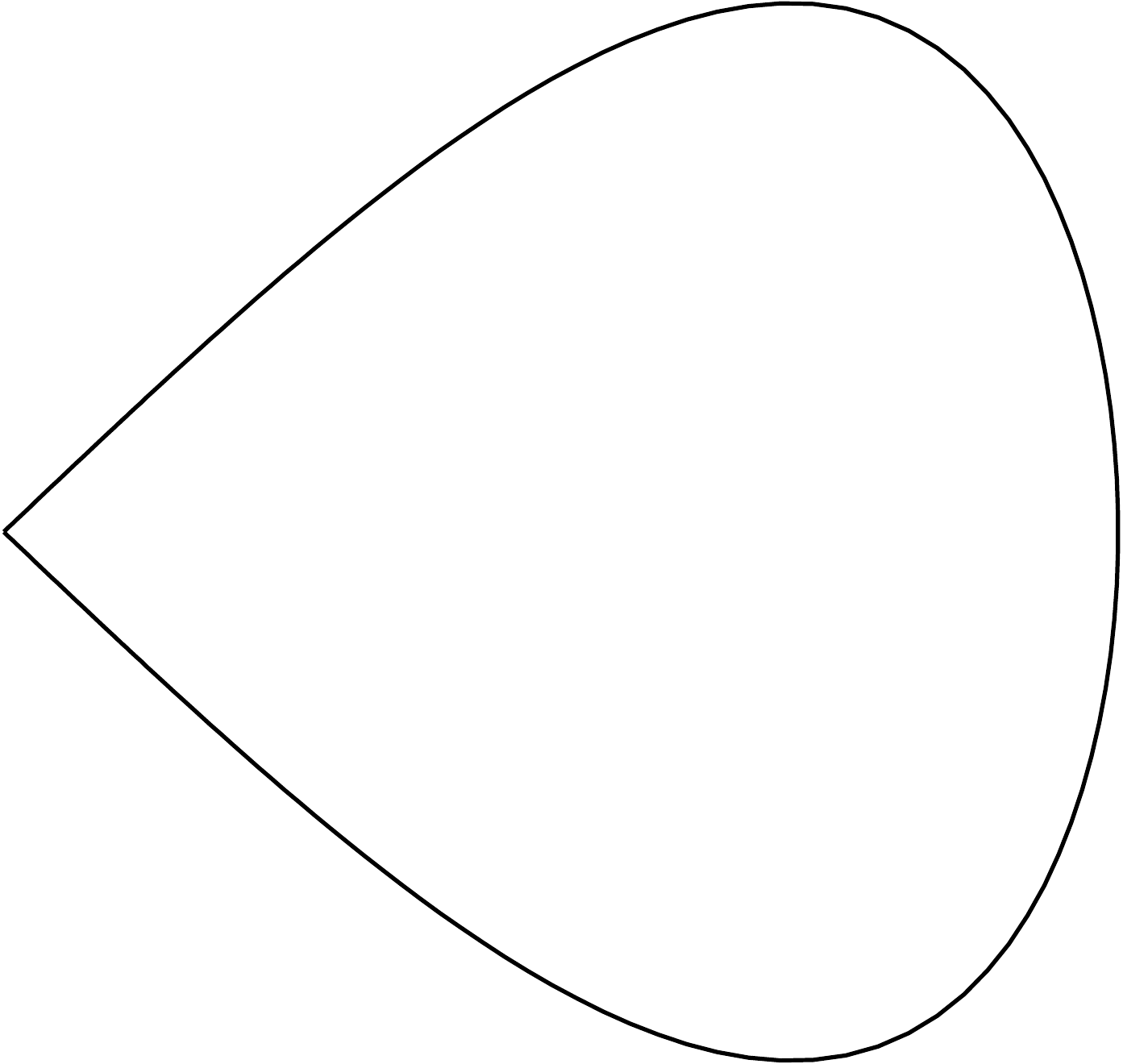} \hspace{1cm}
\includegraphics[width=4.5cm] {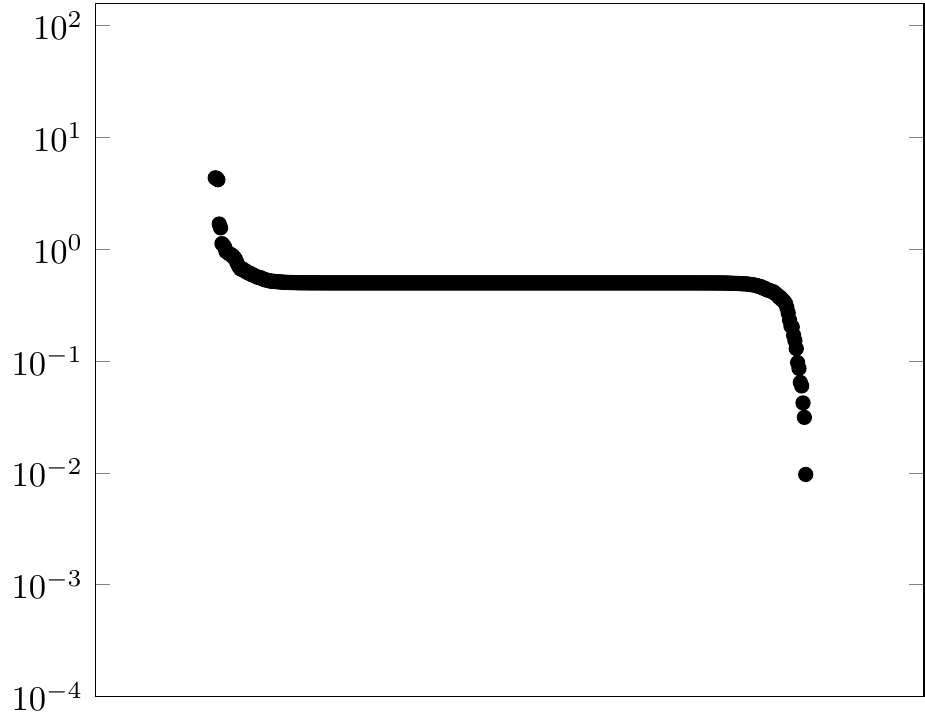}  
\end{center} 
\caption{A cone $\Gamma_{1}$ (left) and the spectrum
of the discretized linear system for $\Gamma_{1}$ (right)}
\label{fig:cone}
\end{figure}

\begin{figure}[h!]
\begin{center}
\includegraphics[height=4.5cm] {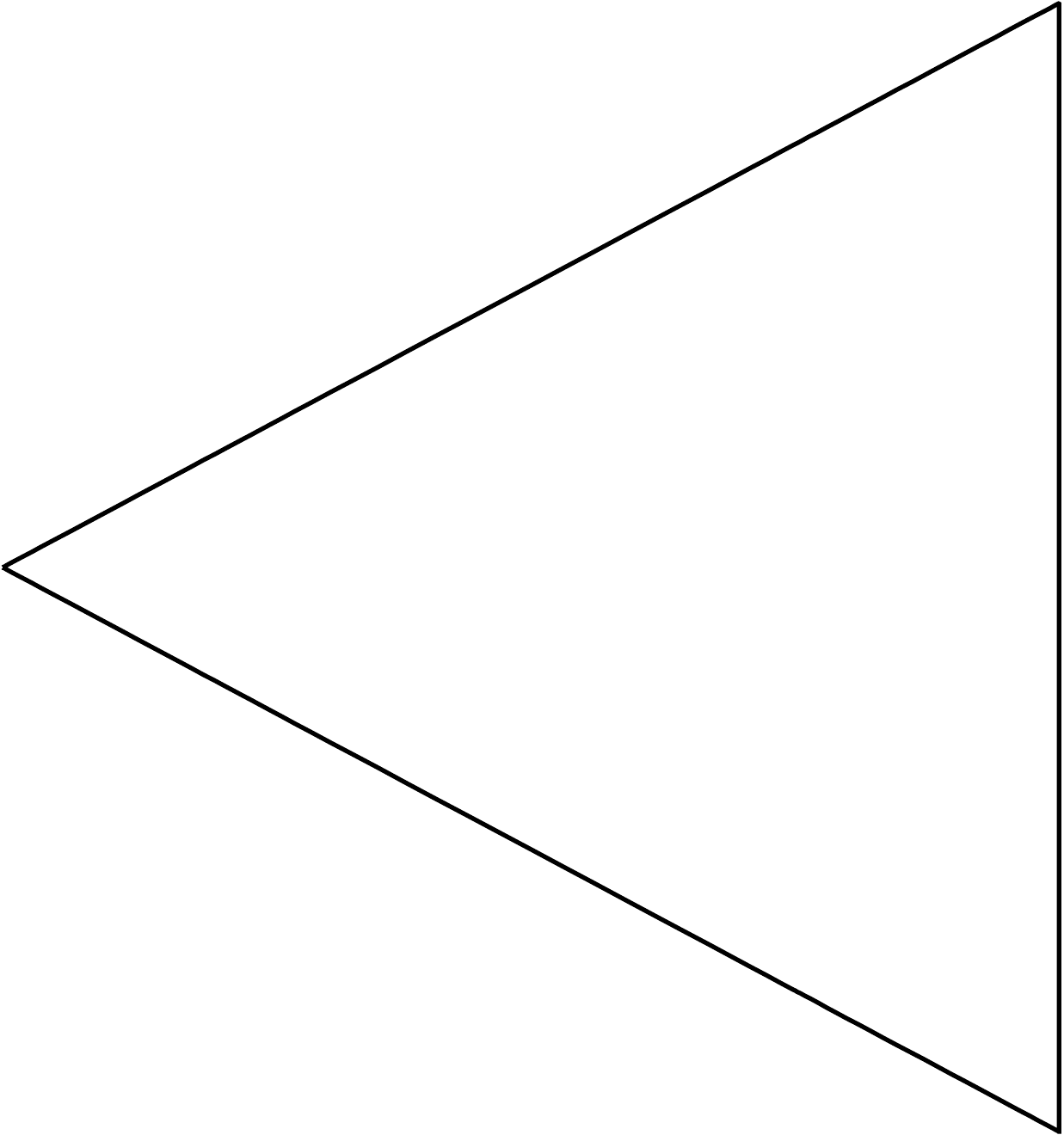} \hspace{1cm}
\includegraphics[width=4.5cm] {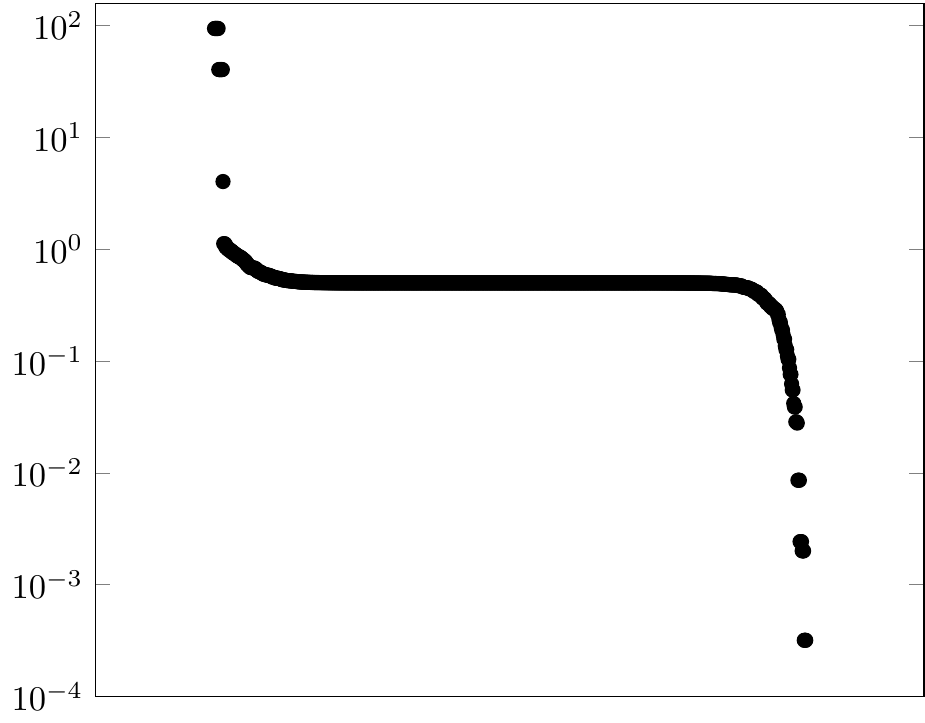}  
\end{center} 
\caption{An equilateral triangle  $\Gamma_{2}$ (left) and the spectrum
of the discretized linear system for $\Gamma_{2}$ (right)}
\label{fig:equitri}
\end{figure}

\begin{figure}[h!]
\begin{center}
\includegraphics[width=4.5cm] {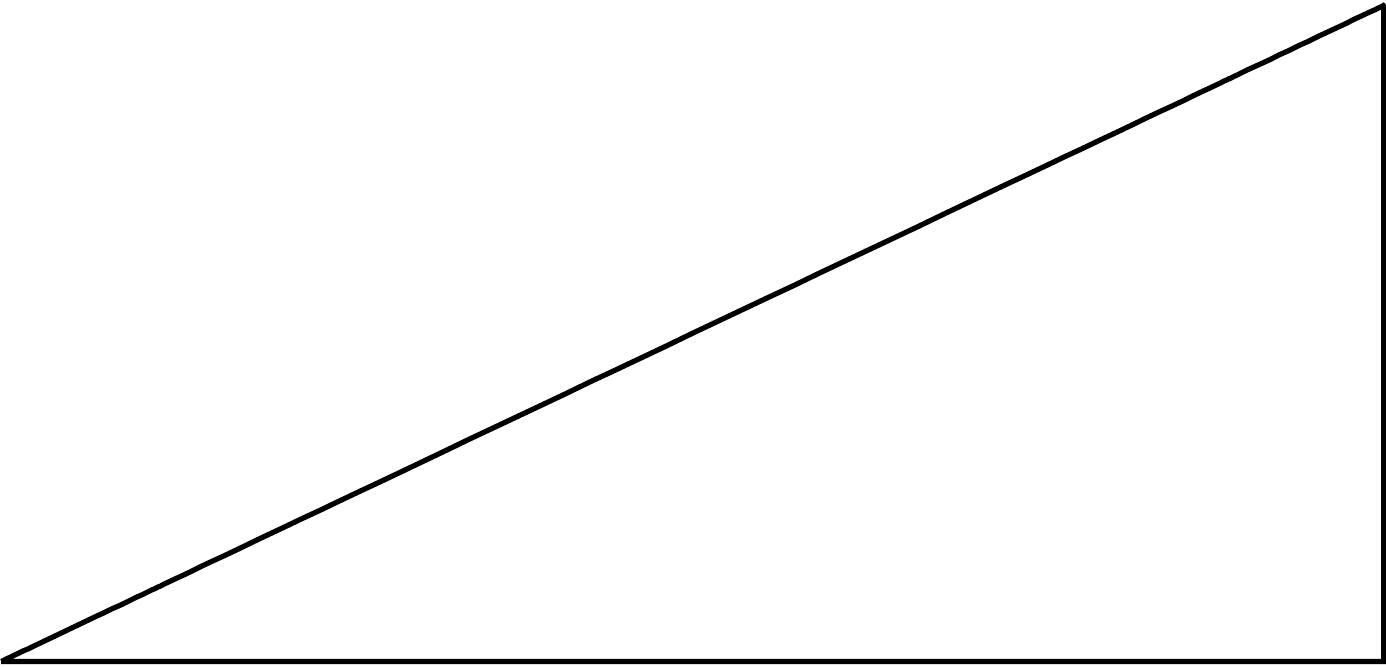} \hspace{1cm}
\includegraphics[width=4.5cm] {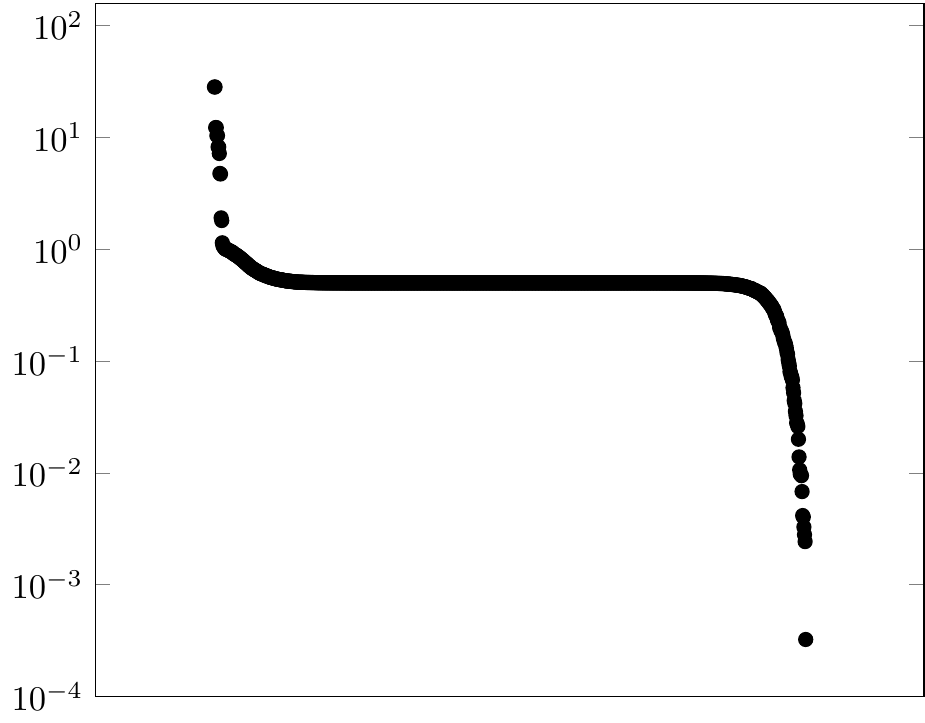}  
\end{center} 
\caption{A right angle triangle $\Gamma_{3}$ (left) and the spectrum
of the discretized linear system for $\Gamma_{3}$ (right)}
\label{fig:righttri}
\end{figure}

\begin{figure}[h!]
\begin{center}
\includegraphics[width=4.5cm] {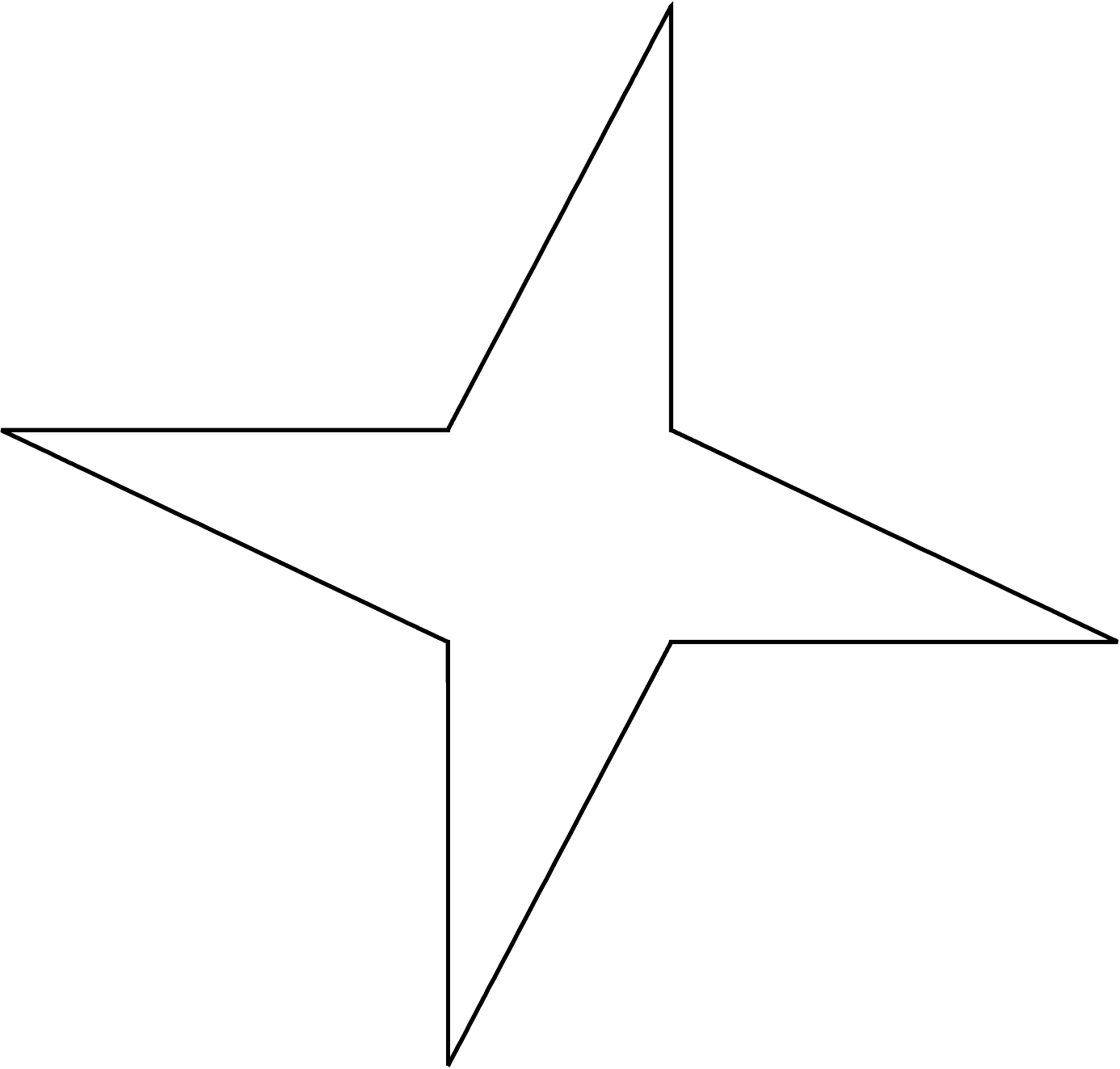} \hspace{1cm}
\includegraphics[width=4.5cm] {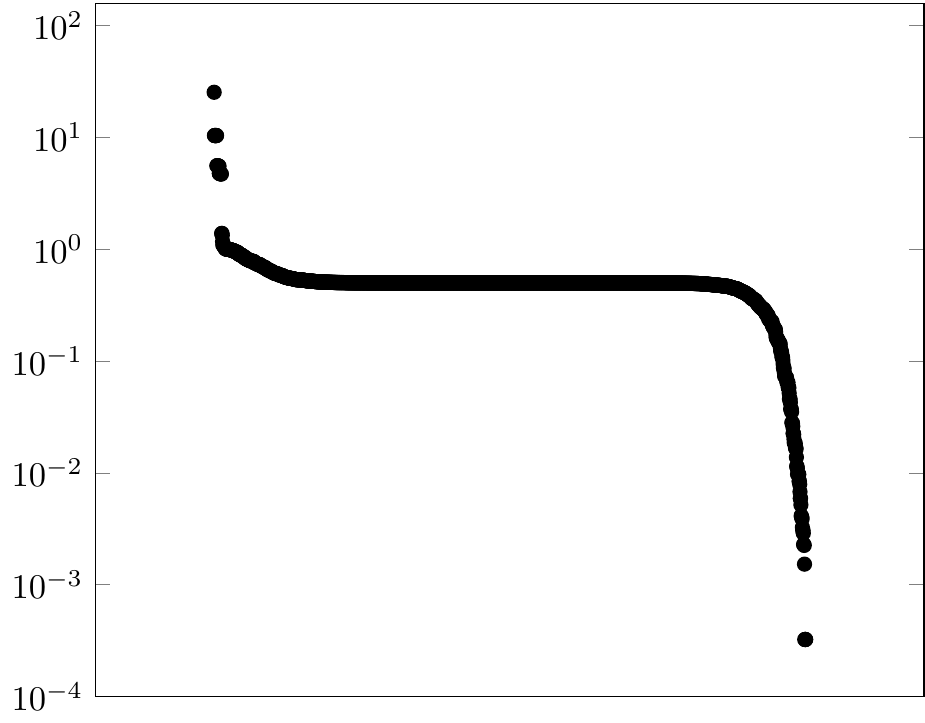}  
\end{center} 
\caption{A star-shaped curve $\Gamma_{4}$ (left) and the spectrum
of the discretized linear system for $\Gamma_{4}$ (right)}
\label{fig:star}
\end{figure}

\begin{figure}[h!]
\begin{center}
\includegraphics[width=4.5cm] {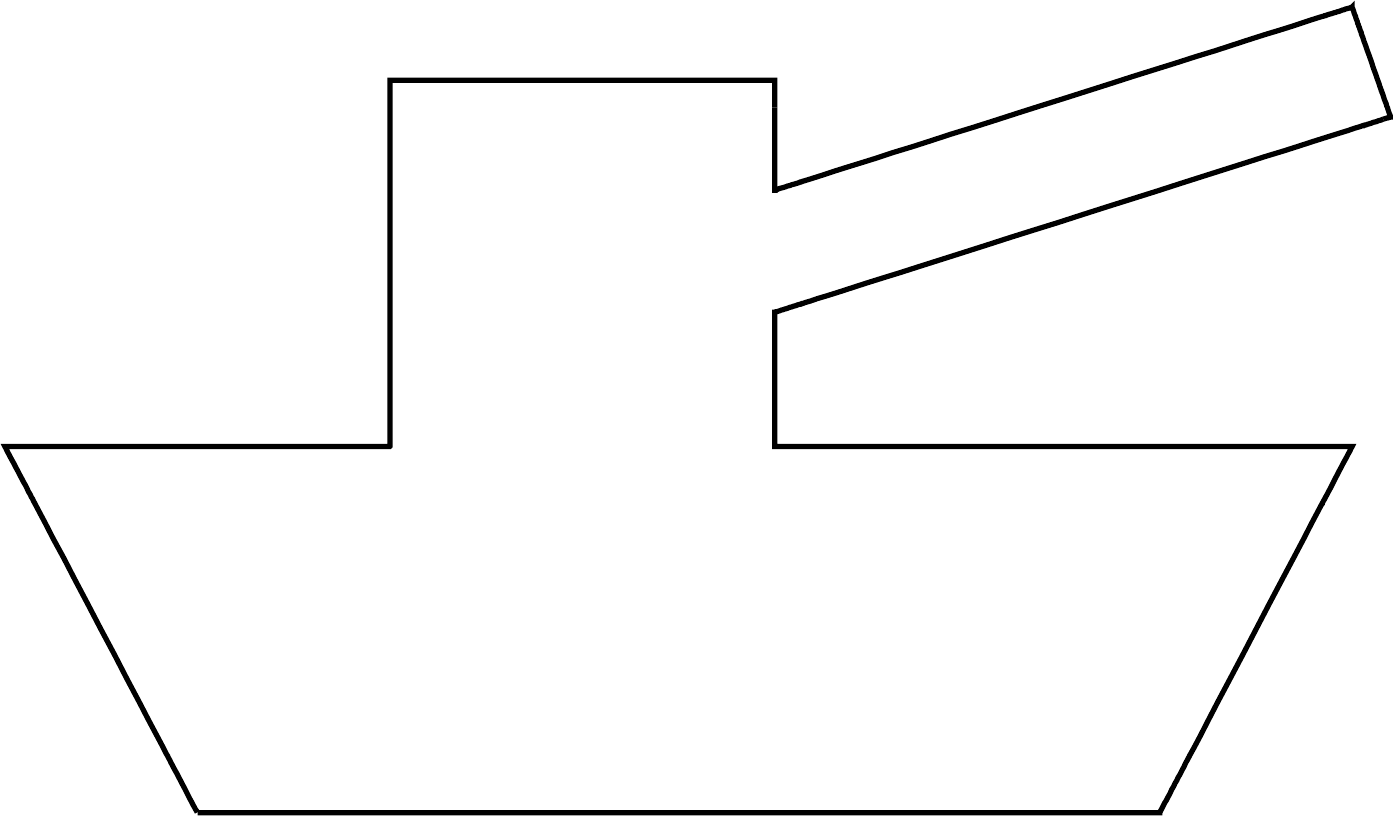} \hspace{1cm}
\includegraphics[width=4.5cm] {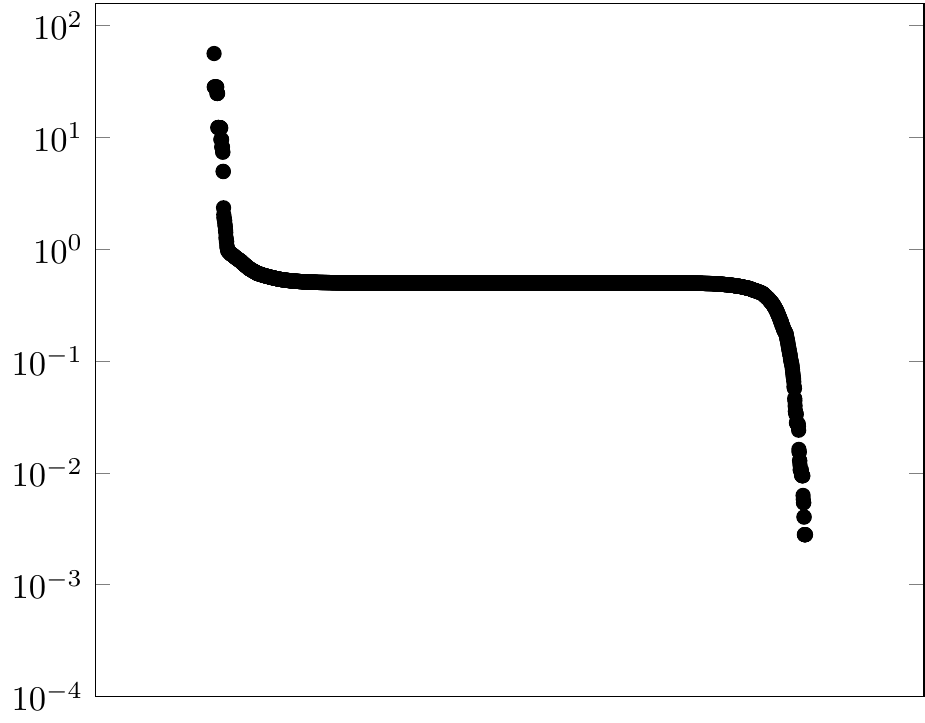}  
\end{center} 
\caption{A tank-shaped curve $\Gamma_{5}$ (left) and the spectrum
of the discretized linear system for $\Gamma_{5}$ (right)}
\label{fig:tank}
\end{figure}

\begin{rem}
We note that the condition numbers reported for the boundaries $\Gamma_{j}$, 
$j=2,3,4,5$, are larger than the condition numbers of the underlying
integral equations. 
This issue can be remedied by using a slightly more involved 
scaling of the discretization scheme (see, for example,~\cite{bremer2012}). 
\end{rem}
%%%%%%%%%%%%%%%%%%%%%%%%%%%%%%%%%%%%%%%%%%%%%%%%%%
\section{Conclusions and extensions \label{sec:concl}}
In this paper, we construct an explicit basis for the solution 
of a standard integral equation
corresponding to the biharmonic equation with gradient boundary conditions,
on polygonal domains. 
The explicit and detailed knowledge of the behavior of solutions
to the integral equation in the vicinity of the corner was
used to create purpose-made discretizations, 
resulting in efficient numerical schemes accurate to essentially machine precision.
In this section, we discuss several directions in which these results
are generalizable.

\subsection{Infinite oscillations of the biharmonic Green's function}
In 1973, S.~Osher showed that the domain Green's function for the biharmonic
equation has infinite oscillations for a right-angled wedge, and conjectured
that this result holds for any domain with corners.
Earlier, we observed that the representations of solutions to the
associated integral equations also exhibit infinite oscillations near the corner 
(see~\cref{rem:infosc2}).
The oscillatory behavior of these basis functions suggests that 
Osher's conjecture may be amenable to an analysis similar to the one presented in
this paper.

\subsection{Curved boundaries}
In this paper, we derive a representation for the solutions of the integral
equations associated with the biharmonic equation, on polygonal domains. 
In the more general case of curved boundaries with corners, the apparatus
of this paper also leads to detailed representations of the solutions
near corners. 
Specifically, the solutions to the associated integral equations 
are representable by rapidly convergent series 
of products of complex powers of $t$  and logarithms of $t$,
where $t$ is the distance from the corner. 
This analysis closely mirrors the generalization of the 
authors' analysis of Laplace's equation on polygonal domains (see~\cite{serkhcorner1})
to the authors' analysis on domains having curved boundaries with corners (see~\cite{serkhcornercurved}).

\subsection{Generalization to three dimensions}
The apparatus of this paper admits a straightforward generalization to
surfaces with edge singularities, where the parts of the surface on either
side of the edge meet at a constant angle along on the edge.
The generalization to the case of edges with more 
complicated geometries is more involved
and will be presented at a later date.

\subsection{Other boundary conditions}
In this paper, we analyze the integral equations associated with
the velocity boundary value problem for Stokes equation.
The approach of this paper extends to a number of other boundary conditions, 
including the traction boundary value problem, and the mobility problem.
In particular, the traction boundary value problem 
and the mobility problem can be formulated as boundary integral equations
that are the adjoints of the integral equations for the velocity
boundary value problem.

\subsection{Modified biharmonic equation}
The modified biharmonic equation for a potential $\psi$ is given by
$\Delta^{2} \psi -\alpha \Delta \psi = 0$.
The equation naturally arises when mixed implicit-explicit schemes are used for the
time discretization of the incompressible Navier-Stokes equation (see, for 
example~\cite{lesliens}).
A preliminary analysis indicates that the solutions of the 
integral equations associated with the modified biharmonic equation
are representable as rapidly convergent series of Bessel
functions of certain non-integer complex orders.
The generalization of the analysis of the biharmonic
equation, presented in this paper, 
to the analysis of the modified biharmonic equation
closely mirrors the generalization of the authors' analysis
of Laplace's equation (see~\cite{serkhcorner1}) to the authors' analysis of 
the Helmholtz equation (see~\cite{serkhcorner3}).

\section{Acknowledgements}
The authors would like to thank Philip Greengard, Leslie Greengard,
Jeremy Hoskins, Vladimir Rokhlin, and Jon Wilkening 
for many useful discussions.

M.~Rachh's work was supported in part by 
the Office of Naval Research award number N00014-14-1-0797/16-1-2123 and
by AFOSR award number FA9550-16-1-0175.
K.~Serkh's work was supported in part by the 
NSF Mathematical Sciences Postdoctoral
Research Fellowship (award number 1606262) and by AFOSR award 
number FA9550-16-1-0175.

%%%%%%%%%%%%%%%%%%%%%%%%%%%%%%%%%%%%%%%%%%%%%%%%%%
\section{Appendix A \label{sec:appa}}
In this appendix, we prove theorems~\ref{thm:mainsingpowtone},
\ref{thm:mainsingpownote}
which are restated here as 
theorems~\ref{thm:mainsingpowtoneapp},
and~\ref{thm:mainsingpownoteapp}, respectively.
\cref{sec:appatone} deals with the tangential odd, normal even case (see~\cref{eq:tone}),
and~\cref{sec:appanote} deals with the tangential even, normal odd case (see~\cref{eq:teno}).

\subsection{Tangential odd, normal even case \label{sec:appatone}}
Suppose that $\bA(z,\theta)$ is the $2\times 2$ matrix 
defined in~\cref{eq:defamat}. 
We recall that 
\beq
\label{eq:detvalapp}
\det{\bA(z,\theta)} = 
\frac{\left( z \sint - \sinzt \right) \left( z \sint - \sinztc \right)}{4 
\sqsinz} \, .
\eeq
If $z$ is not an integer, and satisfies either
\beq
z\sint - \sinztc = 0 , \label{eq:implfun1app0}
\eeq
or
\beq
z \sint - \sinzt = 0 \, ,\label{eq:implfun2app0}
\eeq
then $\det{(\bA(z,\theta))} = 0$.
\cref{sec:appatone1} deals with the implicit functions
defined by~\cref{eq:implfun1app0}
and, similarly,~\cref{sec:appatone2} deals with the implicit
functions defined by~\cref{eq:implfun2app0}.
The principal result of this section is~\cref{thm:mainsingpowtoneapp}, 
which is a restatement of~\cref{thm:mainsingpowtone}.

\subsubsection{Analysis of implicit function $z$ in~\cref{eq:implfun1app0} 
\label{sec:appatone1}}
Suppose that $H: \C \times \C \to \C $ is the entire function defined by
\beq
H(z,\theta) = z \sint - \sinztc \, . \label{eq:defh1}
\eeq
In this section, we investigate the implicit functions $z(\theta)$ 
which satisfy $H(z(\theta),\theta) = 0$.

We begin by stating the connection between $\sinc{(z)}$
and the function $H(z,\theta)$ defined in~\cref{eq:defh1}.
\begin{lem}
\label{lem:htog}
Suppose that $G: \C \times \C \to \C$ is the entire function defined by
\beq
\label{eq:defG}
G(w,\alpha) = \sinc{(w)} + \sinc{(\alpha)} \, .
\eeq
Then $G(w, \alpha) = 0$ if and only if $H(z,\theta) = 0$ 
where $z = \frac{w}{\alpha}$ and $\theta = 2 - \frac{\alpha}{\pi}$.
\end{lem} 
\begin{proof}
Since $z = \frac{w}{\alpha}$ and
$\theta = 2-\frac{\alpha}{\pi}$. 
\begin{align}
H(z,\theta) = z \sint - \sinztc &= 0  \quad \iff \\ 
-z \sin{(\pi (2-\theta))} - \sinztc &=0 \quad \iff \\
\frac{w}{\alpha}\sin{(\alpha)} + \sin{w} &= 0 \quad \iff \\ 
G(w,\alpha) &= 0 \, .
\end{align}
\end{proof}

A simple calculation shows that
\beq
\label{eq:sincder}
\ders{}{z} \sinc{(z)} = \tan{(z)} - z \, .
\eeq
In the following lemma, we discuss the zeros of $\tan{(z)}-z$.
\begin{lem}
\label{lem:sincderzeros}
There exists a countable collection of real $\lambda_{j}>0$, $j=1,2,\ldots$
such that all the zeros of $\tan{(z)}-z$ are given by
$\{ -\lambda_{j} \}_{j=1}^{\infty} \cup \{ \lambda_{j} \}_{j=1}^{\infty} 
\cup \{0 \}$, 
where $\lambda_{j} \in (j\pi + \frac{\pi}{4}, j\pi + \frac{\pi}{2})$.
\end{lem}
\begin{proof}
We first show that all the zeros of $\tan{(z)}-z$ are real.
We observe that, if $z=x+iy$, then  
\beq
\tan{(z)} = \frac{\sin{(2x)}}{ \cos{(2x)} + \cosh{(2y)}} + i
\frac{\sinh{(2y)}}{\cos{(2x)} + \cosh{(2y)}} \, , 
\eeq
If $\tan{z} = z$, then 
\beq
\frac{\sin{(2x)}}{2x} = \frac{\sinh{(2y)}}{2y} \, .
\eeq
For all $y \neq 0$, $|\frac{\sinh{(2y)}}{2y}|>1$ 
and for all $x$, $|\frac{\sin{(2x)}}{2x}| \leq 1$.
Thus, if $\tan{(z)} - z = 0$, then $z \in \R$.

Next, we observe that $x=0$ clearly satisfies $\tan{(x)} = x$.
Furthermore, we note that if $\lambda \in \R$ satifies $\tan{(\lambda)}=\lambda$, 
then $-\lambda$ also satisfies $\tan{(-\lambda)}=-\lambda$ since  
since $\tan{(x)}$ is an odd function of $x$.
Thus, we restrict our attention to the roots of $\tan{(x)}=x$ 
for $x>0$.
There are no roots of $\tan{(x)} = x$ on the intervals 
$(j\pi + \pi/2, (j+1)\pi]$, $j=0,1,2\ldots$, since 
$\tan{(x)} < 0$ for all $x \in (j\pi + \pi/2, (j+1)\pi)$, $j=0,1,2\ldots$.
Moreover, for all $j\geq 1$, there are no roots of $\tan{(x)} =x$ on
the intervals $[j\pi, j\pi + \pi/4]$, since 
$0\leq \tan{(x)} \leq 1$ for $x \in [j\pi, j\pi + \pi/4]$ ,
and $x>1$  on $[j\pi, j\pi +\pi/4]$.
\begin{figure}[h!]
\begin{center}
\includegraphics[height=5.2cm]{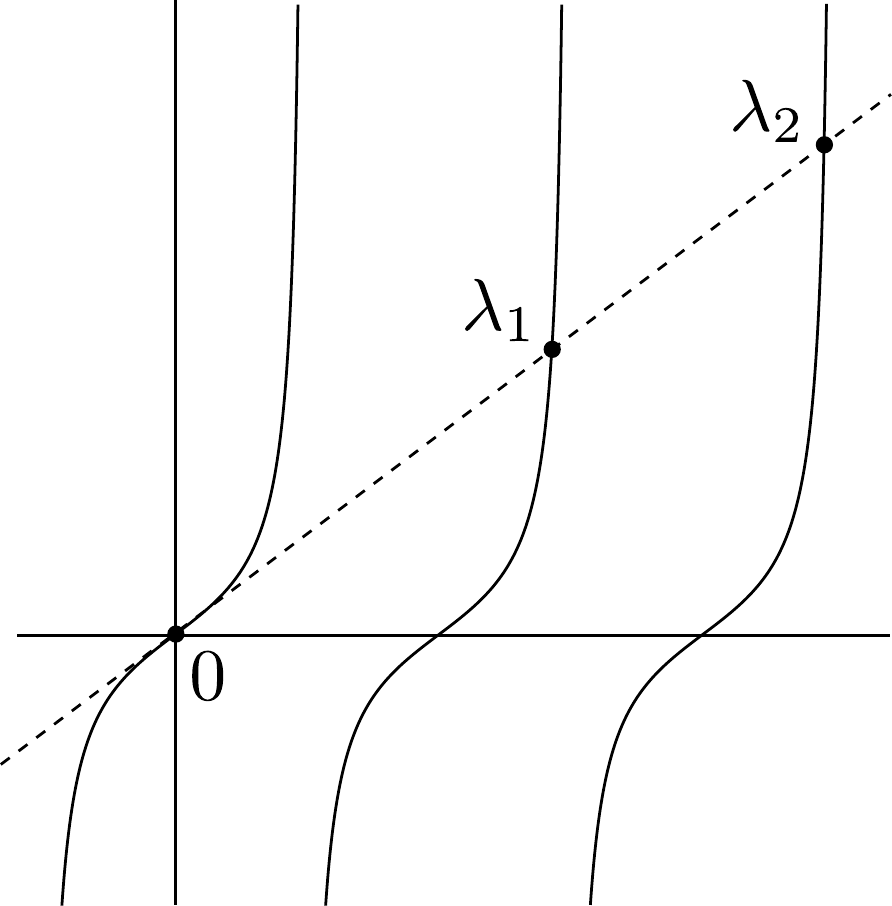}
\caption{The solutions $\lambda_{1}$, $\lambda_{2}$ of $\tan{(x)}=x$.}
\label{fig:tancrossing}
\end{center}
\end{figure}
Also, there are no roots of $\tan{(x)} = x$ for $x \in (0,\pi/2)$, since
$\ders{}{x} \tan{(x)} > 1$ for all $x\in (0,\pi/2)$.
Finally, for each $j=1,2,\ldots$,  
$\tan{(x)}: (j\pi+\pi/4, j\pi + \pi/2) \to (1,\infty)$ 
is a bijection. Thus, there exists exactly one value $\lambda_{j} 
\in (j\pi+ \pi/4, j\pi + \pi/2)$ such that
$\tan{(x)} = x$.
\end{proof}

The following lemma describes some elementary properties of $\lambda_{j}$,
$j=1,2,\ldots \infty$.
\begin{lem}
\label{lem:lamjprop}
Suppose that $\lambda_{j}$ are as defined in~\cref{lem:sincderzeros}. 
Then
\begin{itemize}
\item $\sinc{(\lambda_{j})} = \cos{(\lambda_{j})}$
\item $\cos{(\lambda_{2j})}>0$, $j=1,2,\ldots$, and 
$\cos{(\lambda_{2j-1})}<0$, $j=1,2,\ldots $  
\item $|\cos{(\lambda_{j})}| > |\cos{(\lambda_{j+1})}|$
\end{itemize}
\end{lem}

\begin{proof}
Recall that $\lambda_{j}$, satisfy $\tan{(\lambda_{j})} = \lambda_{j}$.
Thus, 
\beq
\sinc{(\lambda_{j})} = \frac{\sin{(\lambda_{j})}}{\lambda_{j}} = 
\frac{\tan{(\lambda_{j})} \cos{(\lambda_{j})}}{\lambda_{j}} = 
\cos{(\lambda_{j})} \, .
\eeq

Since $\lambda_{j} \in (j \pi + \pi/4, j\pi + \pi/2)$, $j=1,2,\ldots$, 
it follows that $\cos{(\lambda_{2n})} > 0$ and
$\cos{(\lambda_{2n-1})}<0$.

Finally, suppose $\lambda_{j} = j\pi + \theta_{j}$, where 
$\theta_{j} \in (\pi/4,\pi/2)$, $j=1,2,\ldots$. 
We first show that $\theta_{j+1}>\theta_{j}$.
We note that $\tan{(\lambda_{j})} = \tan{(\theta_{j} + j \pi)} = 
\tan{(\theta_{j})}$. 
Thus, 
\beq
\tan{(\theta_{j+1})} = \lambda_{j+1} > \lambda_{j} = \tan{(\theta_{j})} \, .
\eeq
Since, $\tan{(\theta)}$ is a strictly monotonically increasing function 
for $\theta \in (0,\pi/2)$, 
we conclude that $\theta_{j+1}>\theta_{j}$.
Then,
$|\cos{(\lambda_{j})}| = |\cos{(j\pi + \theta_{j})}| = 
\cos{(\theta_{j})}$.
Since $\cos{(\theta)}$ 
is a strictly monotonically decreasing function for
$\theta \in (0,\pi/2)$, we conclude that
$|\cos{(\lambda_{j})}| = \cos{(\theta_{j})} > \cos{(\theta_{j+1})} =
|\cos{(\lambda_{j+1})}|$, $j=1,2,\ldots$.
\end{proof}

In the following lemma, we describe contours in the complex plane
for which $\sinc{(z)}$ is a real number.

\begin{lem}
\label{lem:xj}
Suppose that $j$ is a positive integer and that $\lambda_{j}$ 
is defined in~\cref{lem:sincderzeros}.
Then there exists a function $x_{j}: \R \to (j\pi, \lambda_{j}]$
which satisfies
\beq
\label{eq:sinclevreal}
x_{j} (y)  = \tan{(x_{j}(y))} \cdot y \cdot \coth{(y)} \, , \quad
x_{j}(0) = \lambda_{j} \, ,
\eeq
for all $y\in \R$.
Furthermore, if $z=x_{j}(y) + iy$, then $\sinc{(z)} \in \R$.
\end{lem}

\begin{proof}
It suffices to show existence of the function $x_{j}(y)$ 
which satisfies~\cref{eq:sinclevreal} for $y\geq 0$, 
since if $(x_{j}(y),y)$ satisfies~\cref{eq:sinclevreal}, 
then $(x_{j}(-y),-y)$ also satisfies~\cref{eq:sinclevreal}, i.e.
the function $x_{j}(y)$ defined for $y\in[0,\infty)$ 
can be extended to $y\in (-\infty, \infty)$ 
using an even extension.

We observe that $y \coth{y}:[0,\infty) \to [1,\infty)$ is a
strictly monotically increasing function and a bijection. 
Furthermore, an argument similar to the proof of~\cref{lem:sincderzeros}
shows that
for each $m \in [1,\infty)$, there exists a unique solution $x_{j}$ 
of the equation $x/m = \tan{(x)}$ contained in the interval
$(j\pi,\lambda_{j}]$.
Moreover, the mapping from $m \to x_{j}$ is monotonically decreasing
as a function of $m$ and maps
$m\in [1,\infty) \to x_{j} \in (j\pi, \lambda_{j}]$ (see~\cref{fig:tancrossing2}).
Combining both of these statements, it follows that 
there exists a unique $x_{j}(y)$ for each $y$ 
which satisfies~\cref{eq:sinclevreal}
and moreover, $x_{j}(y): [0,\infty) \to (j\pi, \lambda_{j}]$
is a bijection.

\begin{figure}[h!]
\begin{center}
\includegraphics[height=5.2cm]{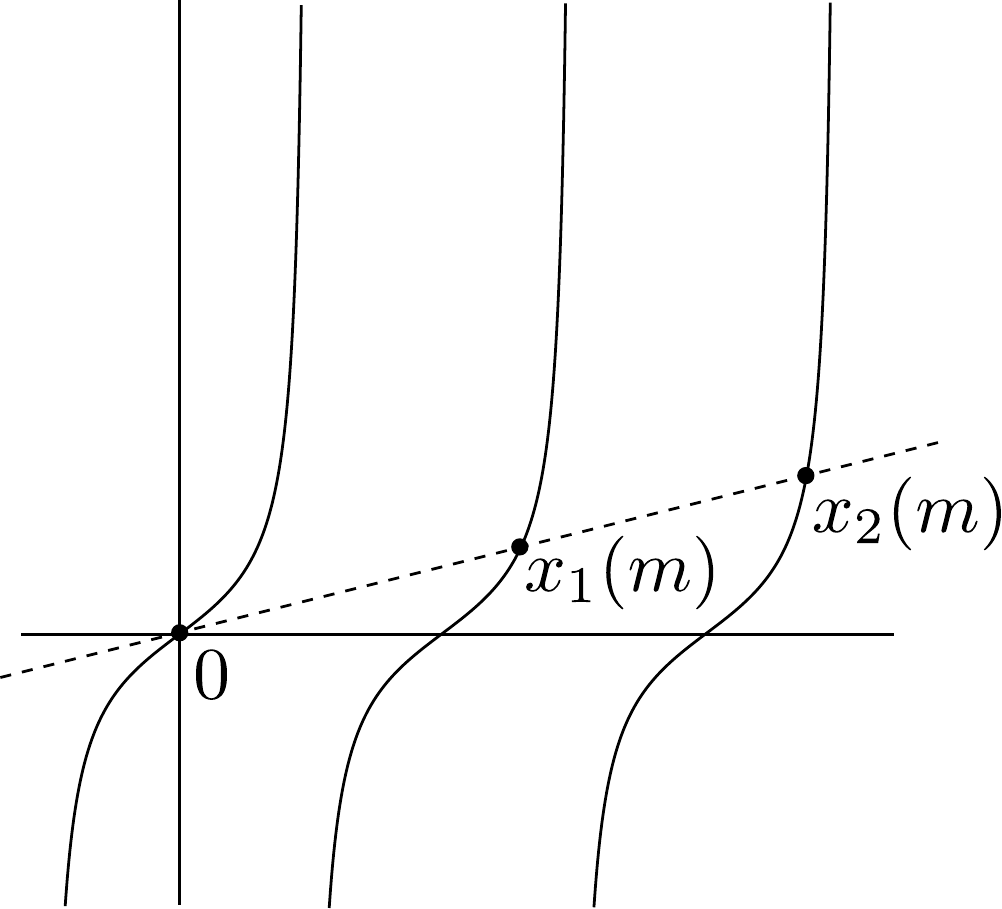}
\caption{An illustrative figure to demonstrate the solutions
$x_{1}(m)$ and $x_{2}(m)$ of $\tan{(x)} = x/m$ as a function of $m$.}
\label{fig:tancrossing2}
\end{center}
\end{figure}

Finally, a simple calculation shows that, if $z=x+iy$, then
\beq
\sinc{(z)}  = \frac{1}{x^2 + y^2}(x \sin{(x)}\cosh{(y)} + y \cos{(x)} \sinh{(y)}
+ i(x\cos{(x)}\sinh{(y)} - y \sin{(x)} \cosh{(y)} )) \, ,
\eeq
from which it follows that $\sinc{(x_{j}(y) +  iy)}$ is real
for all $y \in \R$.
\end{proof}

In the following lemma, we describe the behavior of the $\sinc$
function along the curve $(x_{j}(y),y)$, $j=1,2\ldots$. 

\begin{lem}
\label{lem:xjvals}
Suppose that $j$ is a positive integer. 
Suppose $x_{j}: \R \to (j\pi,\lambda_{j}]$ is as defined in~\cref{lem:xj}.
Suppose further that $z_{j}: \R \to \C$ is defined by 
$z_{j}(y) = x_{j}(y)+iy$.
Then the following holds
\begin{itemize}
\item Case 1, $j$ is even: $\sinc{(z_{j}(y))}$ is a strictly monotonically
increasing function of $y$ for all $y>0$, and
$\sinc{(z_{j}(y))}: (0,\infty) \to (\cos{(\lambda_{j})},\infty)$
is a bijection. 
Likewise, $\sinc{(z_{j}(y))}$ is a strictly monotonically decreasing 
function of $y$ for all $y < 0$, and 
$\sinc{(z_{j}(y))}: (-\infty, 0) \to (\cos{(\lambda_{j})}, \infty)$ 
is a bijection. 
\item Case 2, $j$ is odd: $\sinc{(z_{j}(y))}$ is a strictly monotonically
decreasing function of $y$ for all $y>0$, and
$\sinc{(z_{j}(y))}: (0,\infty) \to (-\infty, \cos{(\lambda_{j})})$
is a bijection. 
Likewise, $\sinc{(z_{j}(y))}$ is a strictly monotonically decreasing 
function of $y$ for all $y < 0$, and 
$\sinc{(z_{j}(y))}: (-\infty, 0) \to (-\infty, \cos{(\lambda_{j})})$ 
is a bijection.
\end{itemize}
\end{lem}

\begin{proof}
We prove the result for the case when $j$ is even. 
The proof for the case when $j$ is odd follows in a similar manner. 

A simple calculation shows that 
\beq
\label{eq:sincval}
\sinc{(z_{j}(y))} =  \frac{\cos{(x_{j}(y))} \sinh{(y)}}{y}
\eeq
Recall that $x_{j}(y) = x_{j}(-y)$, and hence, 
$\sinc{(z_{j}(y))} = \sinc{(z_{j}(-y))}$.
Thus, it suffices to prove the result when $y>0$.

Using~\cref{lem:sincderzeros}, we note that
$d/dz (\sinc{(z_{j}(y))}) \neq 0$ for all $y>0$. 
Hence for every $z_{j}(y)$, there exists a $\delta>0$ such that
$\sinc{(z)}$ is one-one for all $z\in |z-z_{j}(y)| < \delta$.
It then follows that $\sinc{(z_{j}(y))}$ is either a strictly 
monotonically increasing function or a strictly monotonically
decreasing function for all $y>0$.

When $j$ is even, using~\cref{eq:sincval} and that
$\lim_{y\to \infty} x_{j}(y) = j\pi$, we conclude
that $\lim_{y \to \infty} \sinc{(z_{j}(y))} = \infty$. 
Finally, from~\cref{lem:xj,lem:lamjprop}, we note that
$\sinc{(z_{j}(0))} = \sinc{(\lambda_{j})} = \cos{(\lambda_{j})}$.
Thus,  $\sinc{(z_{j}(y))}$ is a strictly 
monotonically increasing function
and $\sinc{(z_{j}(y))}: (0,\infty) \to (\cos{(\lambda_{j})},\infty)$
is a bijection.
\end{proof}

We note that $x(y) = 0$ satisfies~\cref{eq:sinclevreal} for all $y\in \R$.
Moreover, if $z=x(y)+iy = iy$, we note that $\sinc{(z)}$ is real. 
Thus, it is natural to define $x_{0}(y) \equiv 0$ for all $y$.

In the following lemma, we discuss the inverse of $\sinc{(z)}$.

\begin{lem}
\label{lem:gammajinv}
Suppose that $j$ is a positive integer.
Suppose $x_{j}(y)$, $j=1,2,\ldots$ for $y\in \R$, are as defined in~\cref{lem:xj}.
Suppose further that we define $x_{0}(y) = 0$ for all $y\in \R$.
Let $\bHp$ denote the upper half plane and $\bHm$ denote the lower half plane.
Furthermore, for any set $A \subset \C$, we denote the closure of $A$ by
$\overline{A}$.
Suppose $\Gamma_{j,+}$ 
is the open set in the upper half plane bounded by the 
curves $x_{j}(y)$ and $x_{j+1}(y)$, i.e.
\beq
\Gamma_{j,+} = \{ (x,y): x_{j}(y)<x<x_{j+1}(y)\, , \quad \text{and} \quad y>0  \} \, ,
\eeq
for $j=0,1,2,\ldots$ (see~\cref{fig:app2}). 
Similarly suppose that $\Gamma_{j,-}$ is the open set in the
lower half plane bounded by the curves $x_{j}(y)$ and $x_{j+1}(y)$, i.e.
\beq
\Gamma_{j,-} = \{ (x,y): x_{j}(y)<x<x_{j+1}(y)\, , \quad \text{and} \quad y<0  \} \, ,
\eeq
for $j=1,2,\ldots$ (see~\cref{fig:app2}).
Then the following holds.
\begin{itemize}
\item Case 1, $j$ is even: $\sinc{(z)}: \overline{\Gamma}_{j,+} \to \obHm$
is a bijection which maps $\Gamma_{j,+} \to \R$. 
Moreover, the inverse function, which we denote by $\sinc^{-1}_{j,+}(z)$,
is a bijection from $\obHm \to \overline{\Gamma}_{j,+}$ and
is analytic for $z \in \bHm$.
Similarly, $\sinc{(z)}: \overline{\Gamma}_{j,-} \to \obHp$ 
is a bijection which maps $\Gamma_{j,-} \to \R$. 
The inverse function, which we denote by $\sinc^{-1}_{j,-}(z)$, 
is a bijection from $\obHp \to \overline{\Gamma}_{j,-}$ and
is analytic for $z \in \bHp$.
\item Case 2, $j$ is odd: $\sinc{(z)}: \overline{\Gamma}_{j,+} \to \obHp$
is a bijection which maps $\Gamma_{j,+} \to \R$. 
Moreover, the inverse function, which we denote 
by $\sinc^{-1}_{j,+}(z)$,
is a bijection from $\obHp \to \overline{\Gamma}_{j,+}$ and
is analytic for $z \in \bHp$.
Similarly, $\sinc{(z)}: \overline{\Gamma}_{j,-} \to \obHm$ 
is a bijection which maps $\Gamma_{j,-} \to \R$. 
The inverse function, which we denote by $\sinc^{-1}_{j,-}(z)$, 
is a bijection from $\obHm \to \overline{\Gamma}_{j,-}$ and
is analytic for $z \in \bHm$.
\end{itemize}

\begin{figure}[h!]
\begin{center}
\includegraphics[scale=0.5]{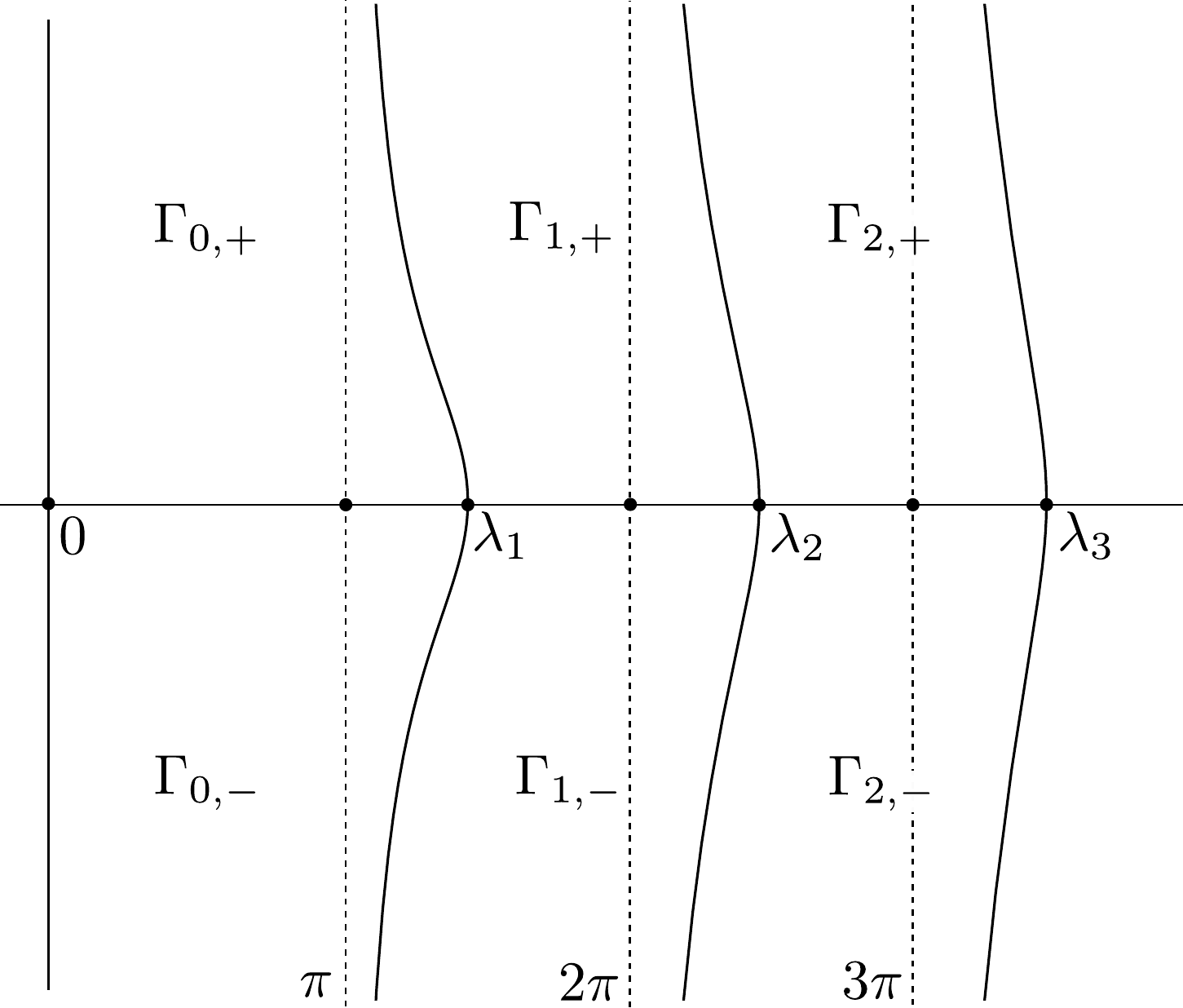}
\caption{The regions $\Gamma_{j,+}$ and $\Gamma_{j,-}$ $j=0,1,2\ldots$}
\label{fig:app2}
\end{center}
\end{figure}

\end{lem}

\begin{proof}
We prove the result for the case 
$\sinc{(z)}: \overline{\Gamma}_{j,+} \to \obHm$, when $j$ is even. 
The results for the other cases follows in a similar manner. 
First, 
it follows from~\cref{lem:sincderzeros} that $\ders{}{z} \sinc{(z)} \neq 0$ for all
$z \in \Gamma_{j,+}$
Thus, $\sinc{(z)}$ is conformal for $z \in \Gamma_{j,+}$.
Using~\cref{lem:xjvals}, we note that $\sinc{(z)}: \partial \Gamma_{j,+} \to \R$
is a bijection. 
Thus, $\sinc{(z)}$ either maps $\Gamma_{j,+}$ to either the lower half plane 
or the upper half plane. 
A simple calculation shows that when $j$ is even,
$\sinc{(z)}$ maps $\Gamma_{j,+}$ to the lower half plane.
Since $\sinc{(z)}$ is conformal for $z \in \Gamma_{j,+}$,
the inverse $\sinc^{-1}_{j,+}(z)$ exists, 
is a bijection from $\obHm \to \overline{\Gamma}_{j,+}$
and is analytic for $z \in \bHm$.
\end{proof}

In the following lemma, we discuss the solutions $\alpha \in [0,2\pi]$ 
of $\sinc{(\alpha)} = -\cos{(\lambda_{j})}$.

\begin{lem}
\label{lem:defalphj}
Suppose $j$ is a positive integer and $\lambda_{j}$ 
is defined in~\cref{lem:sincderzeros}.
\begin{itemize}
\item Case 1, $j$ is even: the equation 
$\sinc{(\alpha)} = -\cos{(\lambda_{j})}$ 
has only two solutions $\ajo$, and $\ajt$
on the interval $\alpha \in [0,2\pi]$ 
where $\pi < \ajo < \lambda_{1} < \ajt < 2\pi$.
\item Case 2, $j$ is odd: the equation 
$\sinc{(\alpha)} = -\cos{(\lambda_{j})}$ 
has only one solution $\ajo$ 
on the interval
$\alpha \in [0,2\pi]$, where $0<\ajo<\pi$.
\end{itemize}
\end{lem}

\begin{proof}
Suppose that $j$ is even. 
We note that 
$-\cos{(\lambda_j)}<0$ and furthermore
$-\cos{(\lambda_{j})} > -\cos{(\lambda_{1})}$ (see~\cref{lem:lamjprop}). 
Firstly, we note that $\sinc{(\alpha)}\geq 0$ for all
$\alpha \in [0,\pi]$. 
Thus, there are no solutions to $\sinc{(\alpha)} = -\cos{(\lambda_{j})}$
for $\alpha \in [0,\pi]$.
Refering to~\cref{fig:sinccurve}, we observe that 
$\sinc{(\alpha)}:(\pi, \lambda_{1}) \to  (-\cos{(\lambda_{1})},0)$
is a bijection. 
Thus, there exists a unique $\ajo \in (\pi,\lambda_{1})$ such
that $\sinc{(\ajo)} = -\cos{(\lambda_{j})}$.
Similarly, $\sinc{(\alpha)}: (\lambda_{1},2\pi) \to
(-\cos{(\lambda_{1})},0)$ is also a bijection. 
Thus, there exists a unique $\ajt \in (\lambda_{1},2\pi)$ such that
$\sinc{(\ajt)} = -\cos{(\lambda_{j})}$.

Suppose now that $j$ is odd.
We note that $-\cos{(\lambda_{j})}>0$ (see~\cref{lem:lamjprop}). 
$\sinc{(\alpha)}\leq 0$ for $\alpha \in [\pi,2\pi]$.
Thus, there are no solutions to $\sinc{(\alpha)} = -\cos{(\lambda_{j})}$
for $\alpha \in [\pi,2\pi]$.
Finally, referring to~\cref{fig:sinccurve}, we observe that 
$\sinc{(\alpha)}: (0,\pi) \to (0,1)$ is a bijection.
Thus, there exists a unique $\ajo \in (0,\pi)$ such that
$\sinc{(\ajo)} = -\cos{(\lambda_{j})}$.
\end{proof}

\begin{figure}[h!]
\begin{center}
\includegraphics[width=7cm]{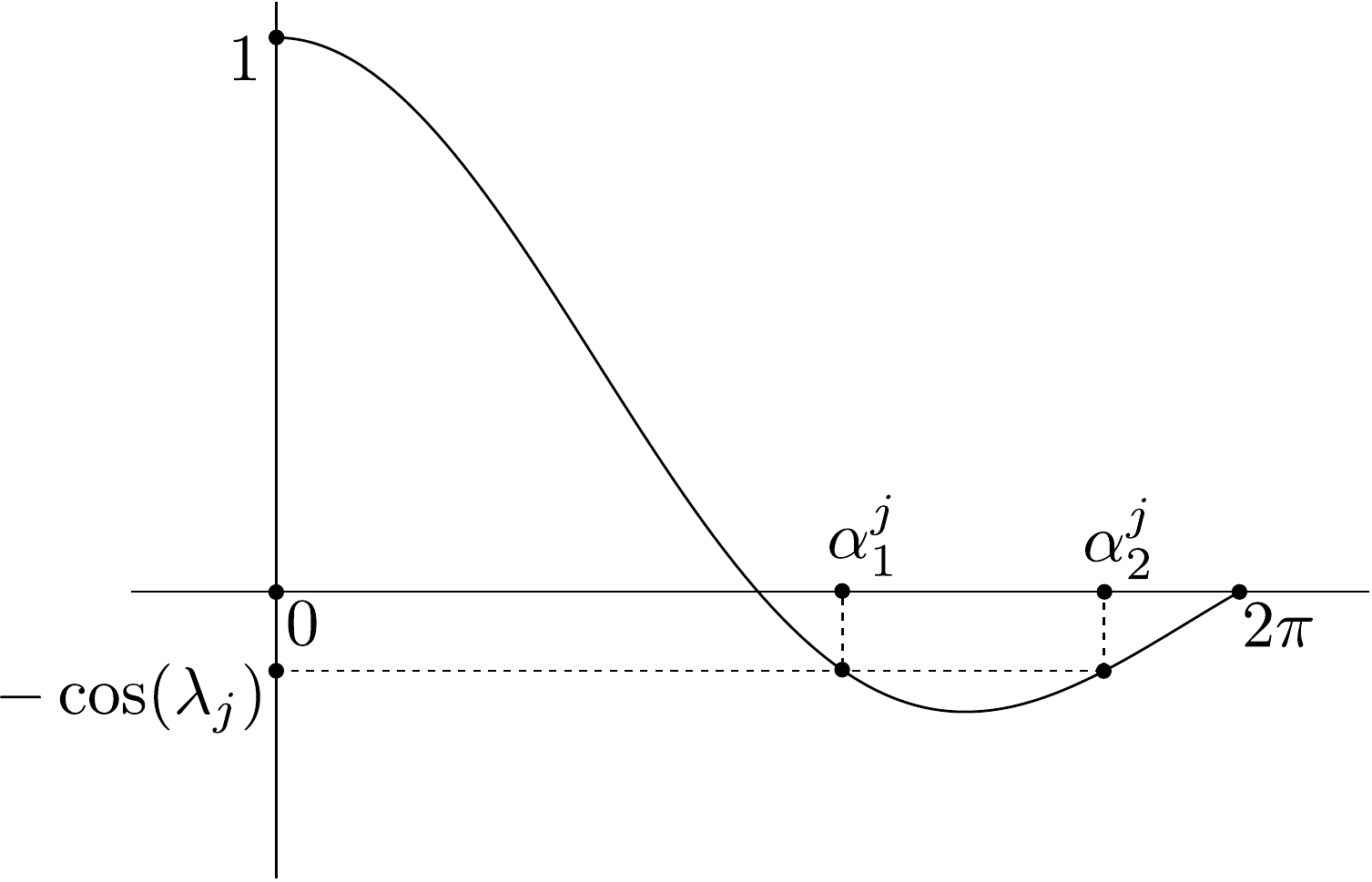}
\includegraphics[width=7cm]{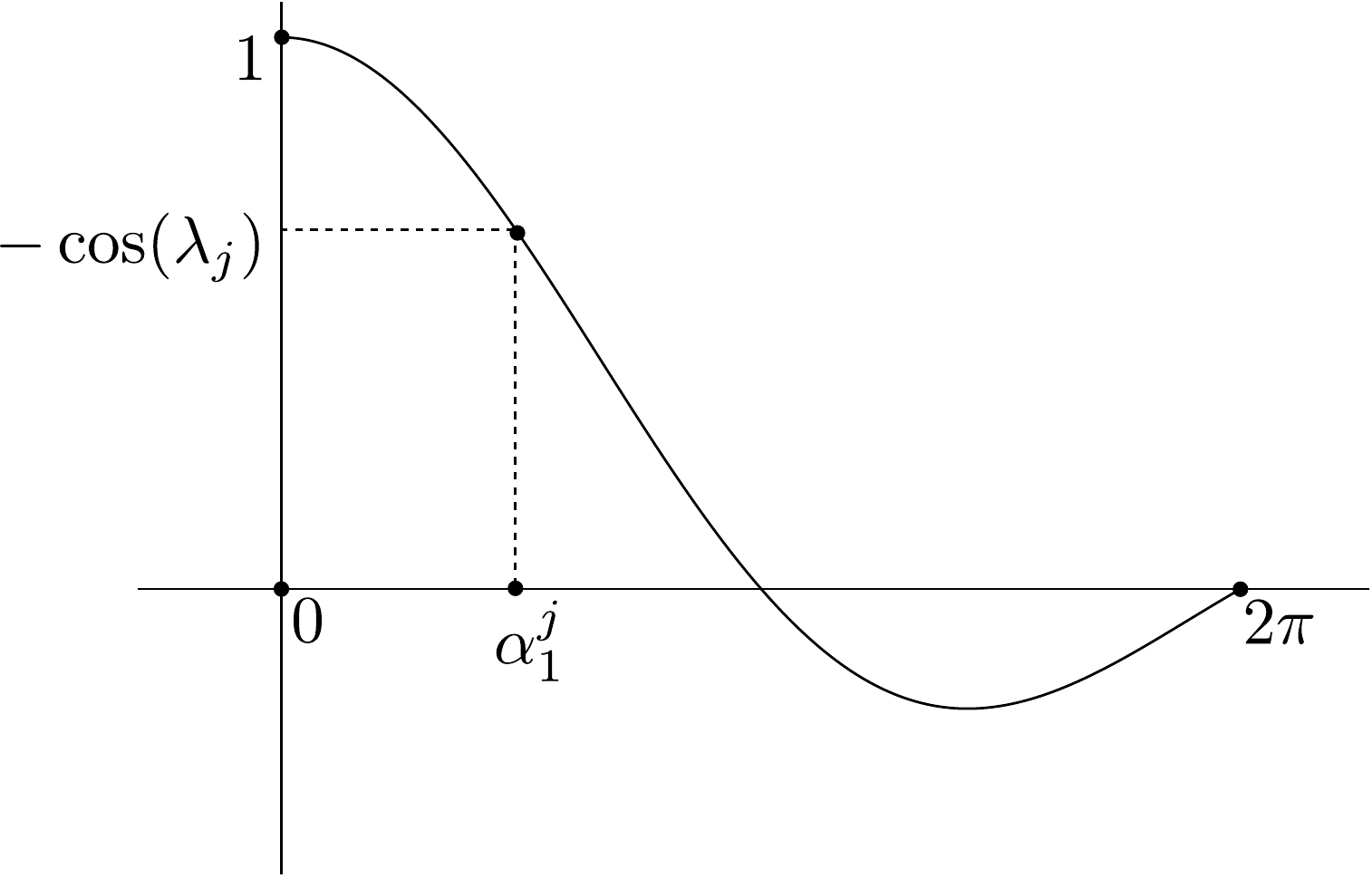}
\caption{Solutions of  $\sinc{(\alpha)} = -\cos{(\lambda_{j})}$. 
Case 1, $j$ is even ($-\cos{(\lambda_{j})}<0$) (left), 
and case 2, $j$ is odd ($-\cos{(\lambda_{j})}>0$) (right).}
\label{fig:sinccurve}
\end{center}
\end{figure}

\subsubsubsection{Even case, $z(1) = 2m$ \label{sec:z12m}}
In this section, we analyze the implicit functions $z(\theta)$ 
which satisfy $H(z,\theta)=0$, with $z(1) = 2m$ where 
$m$ is a positive integer. 
The principal result of this section is~\cref{lem:defz1}.

\begin{lem}
\label{lem:defw1}
Suppose that $m$ is a positive integer, and that $G(w,\alpha)$
is as defined in~\cref{eq:defG}.
Suppose the regions $\Gamma_{j,+}, \Gamma_{j,-}$, $j=0,1,\ldots$ 
are as defined in~\cref{lem:gammajinv}.
Suppose that $\al{2m-1}{1}$, $\al{2m}{1}$, and $\al{2m}{2}$
are as defined in~\cref{lem:defalphj}.
As before, for any set $A$ let $\overline{A}$ denote the closure of $A$.
Furthermore, suppose that $D$ is the strip in the lower half plane with
$0<\text{Re}(\alpha)<2\pi$, i.e. 
\beq
D = 
\{ \alpha \in \C: \, 0<\text{Re}(\alpha)<2\pi \, ,\quad 
\text{Im}(\alpha) < 0 \} \,.
\eeq
Suppose that $D_{1}$ is the region $\overline{D} \cap \overline{\Gamma}_{0,-}$ and
$D_{2}$ is the region $\overline{D} \setminus 
\overline{\Gamma}_{0,-}$ (see~\cref{fig:wdom1}).

\begin{figure}[h!]
\begin{center}
\includegraphics[width=7cm]{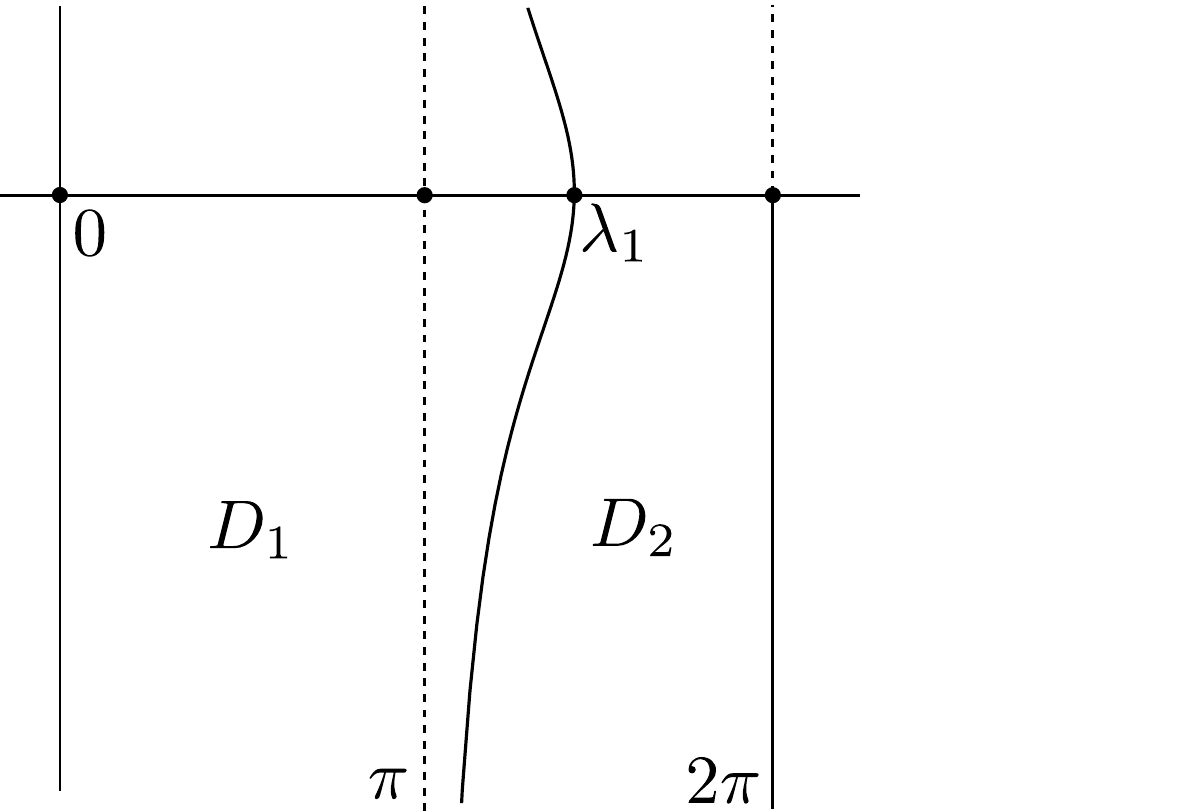}
\caption{The regions $D_{1} =
\overline{D} \cap
\overline{\Gamma}_{0,-}$ and
$D_{2} = \overline{D} \setminus \overline{\Gamma}_{0,-}$. 
}
\label{fig:wdom1}
\end{center}
\end{figure}

Suppose finally that $w(\alpha): \overline{D} \to \C$ is defined by
\beq
w(\alpha) = \begin{cases}
\sinc^{-1}_{2m-1,-} (-\sinc{(\alpha)}) & \quad \alpha \in 
D_{1} \\
\sinc^{-1}_{2m,-} (-\sinc{(\alpha)}) & \quad \alpha \in D_{2} \, .
\end{cases}
\eeq
Then 
for all $\alpha \in D$,
$w(\alpha)$ satisfies $G(w(\alpha),\alpha) = 0$ and is an analytic
function for $\alpha \in D$.
Moreover, $w(\pi) = 2m \pi$.
\end{lem}

\begin{proof}
Suppose, as before, that $\bHp$ denotes the upper half plane
and $\bHm$ denotes the lower half plane. 
We first note that for all
$\alpha \in \overline{D}$,
the function $w$ is well defined and satisfies $G(w(\alpha),\alpha)=0$.
For $\alpha \in D_{1}$, $w(\alpha) = 
\sinc^{-1}_{2m-1,-}(-\sinc{(\alpha)})$.
The domain of definition for 
for $\sinc^{-1}_{2m-1,-}(z)$
is $z\in \obHp$, and using~\cref{lem:gammajinv}, 
$-\sinc{(\alpha)} \in \obHp$
for $\alpha \in D_{1}$. 
Moreover, for $\alpha \in \overline{D}_{1}$, 
\begin{align}
G(w(\alpha),\alpha) &= \sinc{(w(\alpha))} + \sinc{(\alpha)} \\
&= \sinc{(\sinc^{-1}_{2m-1,-} (-\sinc{(\alpha)})} 
 + \sinc{(\alpha)} \\
&= -\sinc{(\alpha)} + \sinc{(\alpha)} = 0 \, .
\end{align}
Similarly,
for $\alpha \in D_{2}$, 
$w(\alpha) = \sinc^{-1}_{2m,-}(-\sinc{(\alpha)})$.
The domain of definition
for $\sinc^{-1}_{2m,-}(z)$
is $z\in \obHm$, and using~\cref{lem:gammajinv},
$-\sinc{(\alpha)} \in \obHm$ 
for all $\alpha \in D_{2}$. 
Moreover, for $\alpha \in \overline{D}_{2}$, 
\begin{align}
G(w(\alpha),\alpha) &= \sinc{(w(\alpha))} + \sinc{(\alpha)} \\
&= \sinc{(\sinc^{-1}_{2m,-} (-\sinc{(\alpha)})} 
 + \sinc{(\alpha)} \\
&= -\sinc{(\alpha)} + \sinc{(\alpha)} = 0 \, .
\end{align}

Clearly, $w(\alpha)$ is analytic for $\alpha \in D_{1} \cap D$ 
since both $\sinc^{-1}_{2m-1,-}(z)$ and $\sinc{(z)}$ are analytic functions
on their respective domains of definition.
Similarly, $w(\alpha)$ is analytic for $\alpha \in 
D_{2} \cap D$.
In order to show that $w(\alpha)$ is analytic for 
$\alpha \in D$,
it suffices to show that $w$ is continuous across 
$\overline{D}_{1} \cap \overline{D}_{2}$.
It follows from the defintions of the regions $D_{1},D_{2}$, that 
$\overline{D}_{1} \cap \overline{D}_{2}$ 
is precisely the curve $(x_{1}(y),y)$ for
$y \in (-\infty,0]$.
For each $y \in (-\infty,0)$, let $\alpha(y) = x_{1}(y)+iy$.
Then 
\beq
\label{eq:sincdefx1}
\{-\sinc(\alpha(y)) \, : \quad -\infty<y<0 \} = (-\cos{(\lambda_{1})},\infty) \,.
\eeq
Let $w(y) = x_{2m}(y) + iy$, for $y\in(-\infty,0)$,
then $\sinc{(w(y))} \in (\cos{(\lambda_{2m})},\infty)$.
Moreover,
$\sinc{(w(y))}$ is a monotonically decreasing function of $y$ 
for $y<0$ (see~\cref{lem:xjvals}).
Furthermore, using~\cref{lem:lamjprop}, we note that
$-\cos{(\lambda_{1})}> \cos{(\lambda_{2m})}$. 
Thus, there exists a unique $y_{1} \in (-\infty,0)$ 
such that $\sinc{(w(y_{1}))} = -\cos{(\lambda_{1})}$.

Referring to~\cref{fig:inv2mm12m}, we observe that
\begin{align}
\label{eq:sincdefx2ml}
\{ \sinc^{-1}_{2m-1,-}{(y)}\, : 
\quad -\cos{(\lambda_{1})} < y < \infty \} 
= \{ x_{2m}(y) + iy \, ,\quad  -\infty<y<y_{1} \} \, .
\end{align}
Similarly,
\beq
\label{eq:sincdefx2mr}
\{ \sinc^{-1}_{2m,-}{(y)} \, :
\quad -\cos{(\lambda_{1})} < y < \infty \} 
= \{ x_{2m}(y) + iy \, ,\quad  -\infty<y<y_{1} \} \, .
\eeq
Combining~\cref{eq:sincdefx1,eq:sincdefx2ml,eq:sincdefx2mr}, we
conclude that $w(\alpha)$ is continuous across 
$\overline{D}_{1} \cap \overline{D}_{2}$. 
It then follows from Morera's theorem that $w(\alpha)$ is analytic
for $\alpha \in D$. 

Finally $\pi \in D_{1}$, and
it follows from the definition of $\sinc^{-1}_{2m-1,-}(z)$, 
that 
\beq
w(\pi) = 
\sinc^{-1}_{2m-1,-} (-\sinc{(\pi)}) = 
\sinc^{-1}_{2m-1,-} (0) = 2m\pi \, ,
\eeq
from which the result follows.
\end{proof}

\begin{rem}
In~\cref{fig:inv2mm12m}, we provide a more detailed description of 
the values of $w(\alpha) \in \C$,
which satisfies $G(w(\alpha),\alpha) = 0$ and $w(\pi) = 2m\pi$,
for $\alpha \in (0,2\pi)$.
\end{rem}

\begin{figure}[h!]
\begin{center}
\begin{subfigure}[The values of $\sinc{(\alpha)}$ 
for $\alpha \in (0,2\pi)$]{
\includegraphics[width=8.5cm]{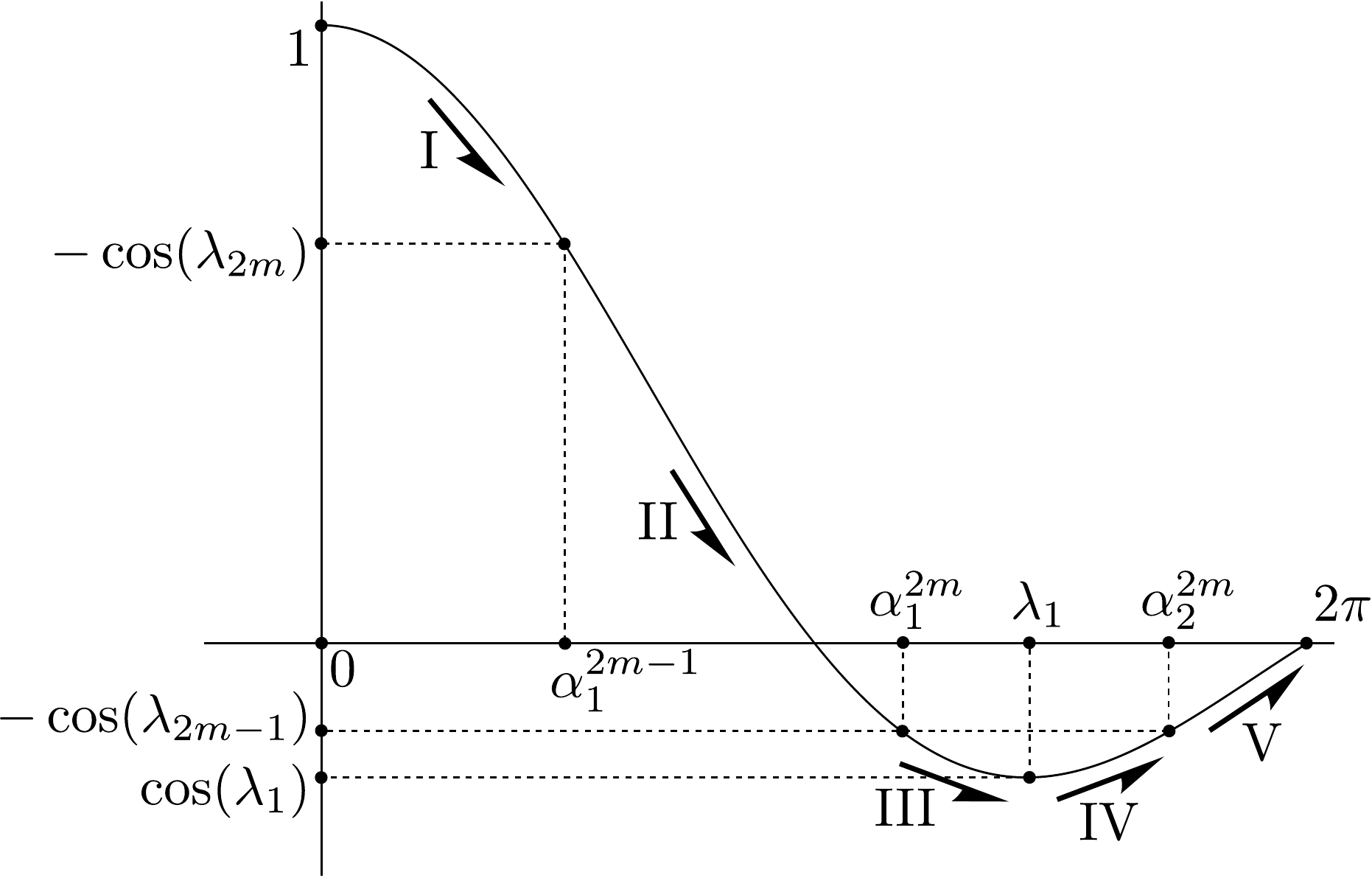} \label{fig:inv2mm12ma}}
\end{subfigure}\\
\vspace*{2ex}
\begin{subfigure}[The corresponding values of $w(\alpha)$ which
satisfy $G(w(\alpha),\alpha)=0$ with
$w(\pi) = 2m\pi$]{
\includegraphics[width=9.5cm]{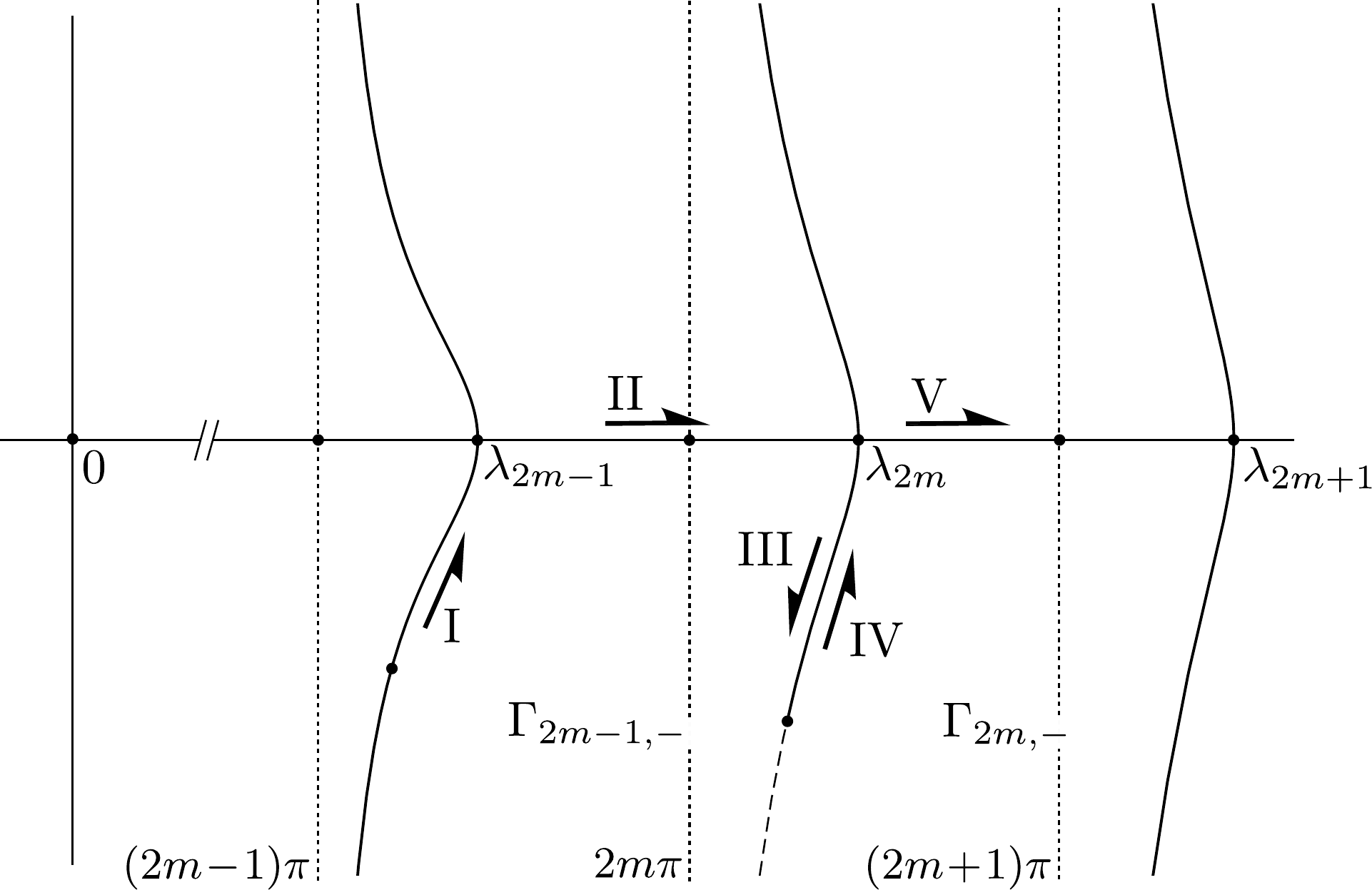} \label{fig:inv2mm12mb}}
\end{subfigure}
\caption{The values $\sinc{(\alpha)}$ for $\alpha \in (0,2\pi)$  (\cref{fig:inv2mm12ma}) and
the corresponding values of $w(\alpha)$ which satisfy $G(w(\alpha),\alpha)=0$ with
$w(\pi) = 2m\pi$ (\cref{fig:inv2mm12mb}).
In~\cref{fig:inv2mm12mb}, segment I represents $w(\alpha)$ for $\alpha \in (0,\al{2m-1}{1})$,
segment II represents $w(\alpha)$ for $\alpha \in (\al{2m-1}{1},\al{2m}{1})$, 
segment III represents $w(\alpha)$ for $\alpha \in (\al{2m}{1},\lambda_{1})$,
segment IV represents $w(\alpha)$ for $\alpha \in (\lambda_{1}, \al{2m}{2})$,
and finally segment V represents $w(\alpha)$ for $\alpha \in (\al{2m}{2},2\pi)$.}.
\label{fig:inv2mm12m}
\end{center}
\end{figure}

In the following lemma, we further 
extend the domain of definition of $w(\alpha)$ 
defined in~\cref{lem:defw1}
to a simply connected open set containing the 
strip in the lower half plane with $0<\text{Re}(\alpha)<2\pi$ 
that includes the interval 
$(0,2\pi) \setminus \{ \al{2m-1}{1},\al{2m}{1},\al{2m}{2} \}$.

\begin{lem}
\label{lem:defw2}
Suppose that $m$ is a positive integer, and that $G(w,\alpha)$
is as defined in~\cref{eq:defG}.
Suppose that $\al{2m-1}{1}$, $\al{2m}{1}$, and $\al{2m}{2}$
are as defined in~\cref{lem:defalphj}.
As before, let $\overline{A}$ denote the closure of the set $A$.
Furthermore, suppose that the region $D$ and the analytic function
$w(\alpha): \overline{D} \to \C$ is as defined in~\cref{lem:defw1}.
Suppose now that $\tilde{D}$ is the strip in the lower half plane with 
$0 < \text{Re}(\alpha) < 2\pi$ that includes the interval
$(0,2\pi) \setminus \{ \al{2m-1}{1}, \al{2m}{1}, \al{2m}{2} \}$, i.e.,
\beq
\label{eq:tildeD}
\tilde{D} = \{ \alpha \in \C : \, 0<\text{Re}(\alpha)<2\pi \, , \quad
\text{Im}(\alpha) \geq 0 \}  \setminus \{ \al{2m-1}{1}, \al{2m}{1}, \al{2m}{2} \} \, .
\eeq
Then there exists a simply connected open set $\tilde{D} \subset \tilde{V} \subset \C$
(see~\cref{fig:vtildedom}
and an analytic function $\tw(\alpha): \tilde{V} \to \C$ which
satisfies $G(\tw(\alpha),\alpha) = 0$ for all $\alpha \in \tilde{V}$ and
$\tw(\pi) = 2m\pi$. 
Moreover $\tw(\alpha) = w(\alpha)$ for all $\alpha \in \overline{D} \cap \tilde{V}$.
\end{lem}

\begin{figure}[h!]
\begin{center}
\includegraphics[width=6cm]{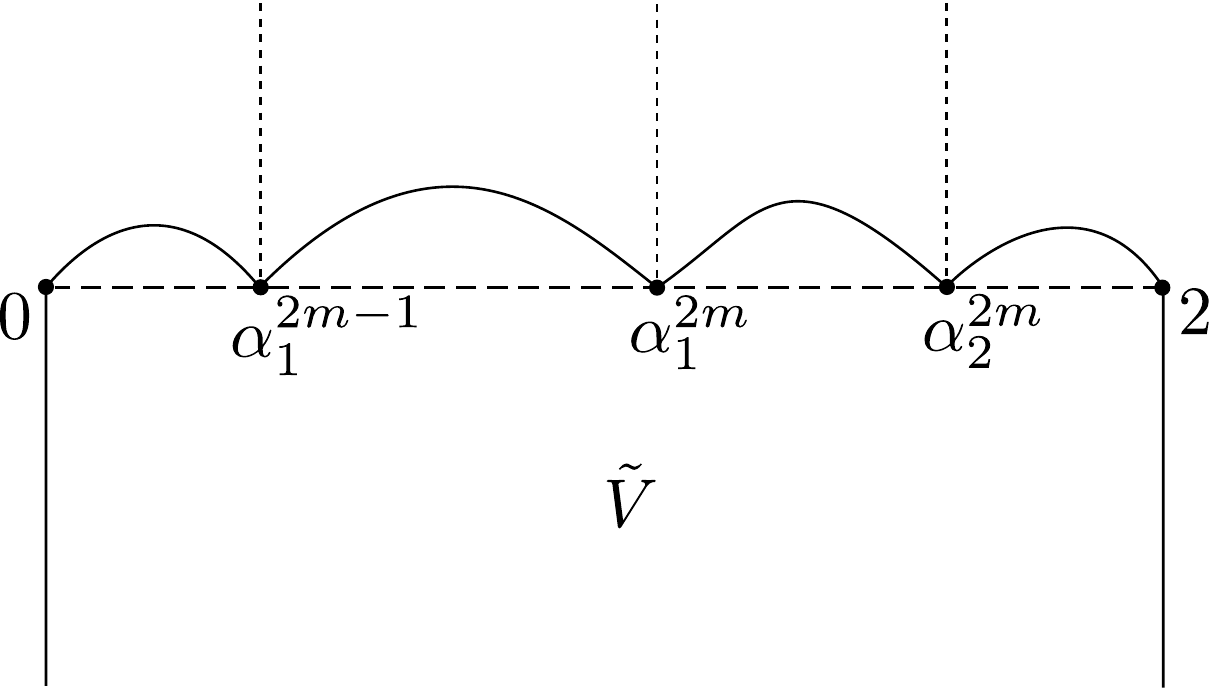}
\caption{An illustrative region of analyticity $\tilde{V}$ of
the function $w(\alpha)$, which satisfies $G(w(\alpha),\alpha)=0$.}
\label{fig:vtildedom}
\end{center}
\end{figure}

\begin{proof}
For all $\alpha \in \overline{D} \cap \tilde{V}$, we define $\tw(\alpha) = w(\alpha)$.
We also note that the interval $(0,2\pi) 
\setminus \{ \al{2m-1}{1}, \al{2m}{1}, \al{2m}{2} \} \subset \overline{D} \cap \tilde{V}$.
Furthermore, $\tw(\alpha)$ also satisfies $G(\tw(\alpha),\alpha) = 0$ for all 
$\alpha \in \overline{D} \cap \tilde{V}$,
since $w(\alpha)$ satisfies $G(w(\alpha),\alpha) = 0$.
A simple calculation shows that $\pd_{w} G(\tw(\alpha),\alpha) =
\tan{(\tw(\alpha))} - w(\alpha)$.
Moreover, it follows from the definition of $\tw(\alpha)$ 
that $\tw(\alpha) \neq \lambda_{j}$, $j=1,2,\ldots$ 
for all $\alpha \in (0,2\pi) \setminus \{ \al{2m-1}{1},\al{2m}{1},\al{2m}{2} \}$.
Thus, we conclude from~\cref{lem:sincderzeros} that 
$\pd_{w} G(\tw(\alpha_{0}),\alpha_{0}) \neq 0$ 
for each $\alpha_{0} \in (0,2\pi) \setminus \{ \al{2m-1}{1}, \al{2m}{1},
\al{2m}{2} \}$.
Finally, by the implicit function theorem there exists a $\delta>0$ and
an implicit function $\tw(\alpha): |\alpha - \alpha_{0}| \to \C$ 
which satisfies $G(\tw(\alpha),\alpha) = 0$, from which the result follows.
\end{proof}

We now present the principal result of this section.
\begin{lem}
\label{lem:defz1}
Suppose that $m$ is a positive integer and that $H(z,\theta)$ is
as defined in~\cref{eq:defh1}. 
Suppose that $\al{2m-1}{1}$, $\al{2m}{1}$, and $\al{2m}{2}$
are as defined in~\cref{lem:defalphj}.
Suppose that $\theta_{1},\theta_{2}$, and $\theta_{3}$ are given by
\beq
\theta_{1} = 2-\frac{\al{2m}{2}}{\pi} \, , \quad
\theta_{2} = 2-\frac{\al{2m}{1}}{\pi} \, , \quad
\theta_{3} = 2-\frac{\al{2m-1}{1}}{\pi} \, . 
\eeq
Suppose further that $D$ is the strip in the upper half plane with 
$0 < \text{Re}(\theta) < 2$ that includes the interval
$(0,2) \setminus \{ \theta_{1},\theta_{2},\theta_{3} \}$, i.e.
\beq
D = \{ \theta \in \C: \, 0 < \text{Re}(\theta) < 2 \, , \quad \text{Im}(\theta) \geq 0 \} 
\setminus \{ \theta_{1},\theta_{2},\theta_{3} \} \, .
\eeq
Then there exists a simply connected open set $D \subset V \subset \C$ (see~\cref{fig:vdom}) 
and an
analytic function $z(\theta): V \to \C$ which satisfies $H(z(\theta),\theta) = 0$
for all $\theta \in V$ and $z(1) = 2m$.
\end{lem}

\begin{figure}[h!]
\begin{center}
\includegraphics[width=6cm]{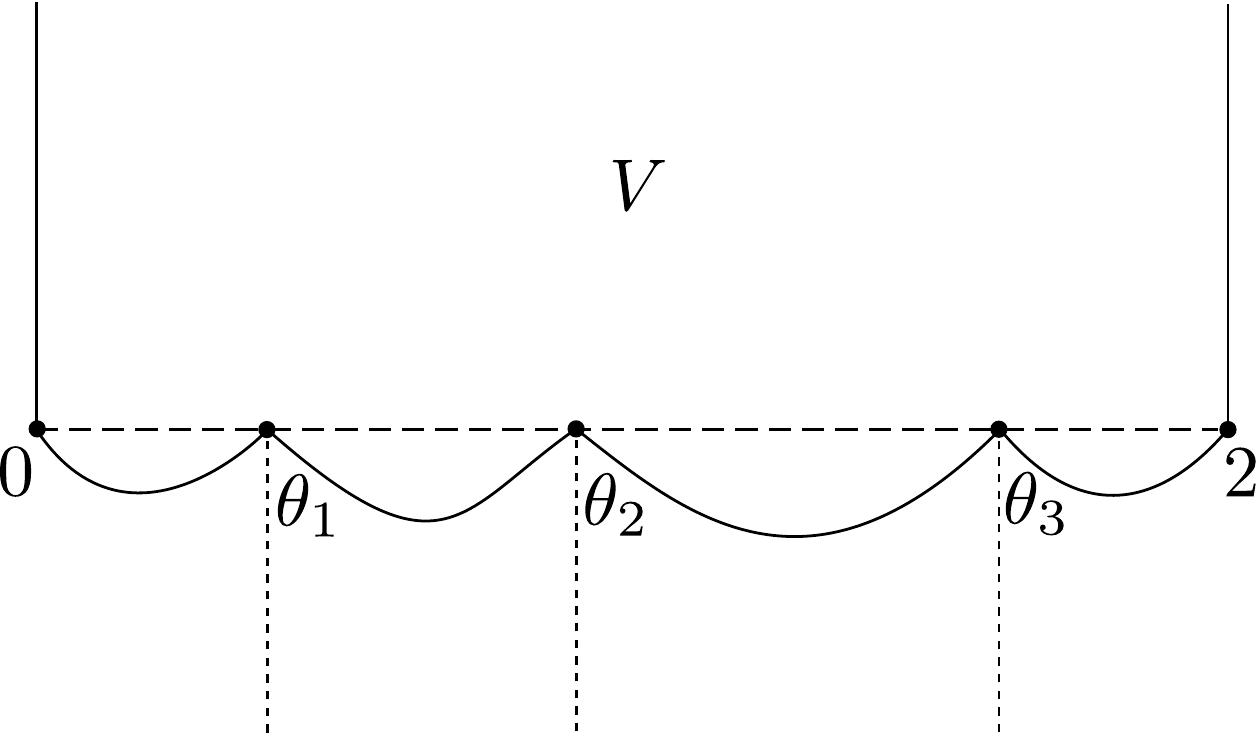}
\caption{An illustrative region of analyticity $V$ of the  
function $z(\theta)$, which satisfies $H(z(\theta),\theta)=0$.}
\label{fig:vdom}
\end{center}
\end{figure}

\begin{proof}
Suppose that $\tilde{V}$ and $\tw(\alpha): \tilde{V} \to \C$ are
as defined in~\cref{lem:defw2}.
Recall that $\tilde{V}$ is an open set containing the strip $\tilde{D}$ defined
in~\cref{eq:tildeD}.
Let 
\beq
\theta = 2 - \frac{\alpha}{\pi} \, ,\quad 
z(\theta) = \frac{w((2-\theta)\pi)}{\pi(2-\theta)} \, , \quad \text{and} \quad 
V = 2 - \frac{\tilde{V}}{\pi} \, .
\eeq
For all $\alpha \in \tilde{V}$, we note that $\theta \in V$.
Furthermore, $\tw(\pi) = 2m\pi$ implies that $z(1) = 2m$.
Finally, using~\cref{lem:htog}, we conclude that $z(\theta)$ 
satisfies $H(z(\theta),\theta) = 0$.
\end{proof}

\begin{rem}
Using the Taylor expansion of $\sinc{(\alpha)}$ in the neighborhood of
$\al{2m-1}{1}$, $\al{2m}{1}$ and $\al{2m}{2}$, it is straightforward
to show that $w(\alpha)$ in~\cref{lem:defw1} has square root 
singularities at $\alpha = \al{2m-1}{1},\al{2m}{1},\al{2m}{2}$. 
It then follows from the definition of $z(\theta)$ in~\cref{lem:defz1} has
square root singularities at $\theta= \theta_{1},\theta_{2}$, and
$\theta_{3}$. 
Thus, $\theta_{1},\theta_{2}$, and $\theta_{3}$ are branch points
for the function $z(\theta)$.
\end{rem}

\subsubsubsection{Odd case, $z(1) = 2m-1$, $m\neq 1$}
In this section, we analyze the implicit functions which satisfy
$H(z,\theta) = 0$, with $z(1) = 2m-1$, where $m \geq 2$ 
is an integer. 
The principal result of this section is~\cref{lem:defz2}.
The proofs of the results presented in this section are
analogous to the corresponding proofs in~\cref{sec:z12m}.
We present the statements of the theorem without proofs for brevity.

In the following lemma, we construct an analytic function $w(\alpha)$
which satisfies $G(w(\alpha),\alpha) = 0$ with $w(1) = (2m-1)\pi$.

\begin{lem}
\label{lem:defw3}
Suppose that $m\geq 2$ is an integer, and that $G(w,\alpha)$
is as defined in~\cref{eq:defG}.
Suppose the regions $\Gamma_{j,+}, \Gamma_{j,-}$, $j=0,1,\ldots$ 
are as defined in~\cref{lem:gammajinv}.
Suppose that $\al{2m-1}{1}$, $\al{2m-2}{1}$, and $\al{2m-2}{2}$
are as defined in~\cref{lem:defalphj}.
As before, let $\overline{A}$ denote the closure of the set $A$.
Furthermore, suppose that $D$ is the strip in the lower half plane with
$0<\text{Re}(\alpha)<2\pi$, i.e. 
\beq
D = \{ \alpha \in \C: \, 0<\text{Re}(\alpha)<2\pi \, ,\quad 
\text{Im}(\alpha) < 0 \} \,.
\eeq
Suppose that $D_{1}$ is the region $\overline{D} \cap \overline{\Gamma}_{0,-}$ and
$D_{2}$ is the region $\overline{D} \setminus D_{1}$.
Suppose finally that $w(\alpha): \overline{D} \to \C$ is defined by
\beq
w(\alpha) = \begin{cases}
\sinc^{-1}_{2m-2,+} (-\sinc{(\alpha)}) & \quad \alpha \in 
D_{1} \\
\sinc^{-1}_{2m-3,+} (-\sinc{(\alpha)}) & \quad \alpha \in D_{2} \, .
\end{cases}
\eeq
Then 
for all $\alpha \in D$,
$w(\alpha)$ satisfies $G(w(\alpha),\alpha) = 0$ and is an analytic
function for $\alpha \in D$.
Moreover, $w(\pi) = (2m-1) \pi$.
\end{lem}

\begin{figure}[h!]
\begin{center}
\begin{subfigure}[The values $\sinc{(\alpha)}$ for $\alpha \in (0,2\pi)$]{
\includegraphics[width=8.5cm]{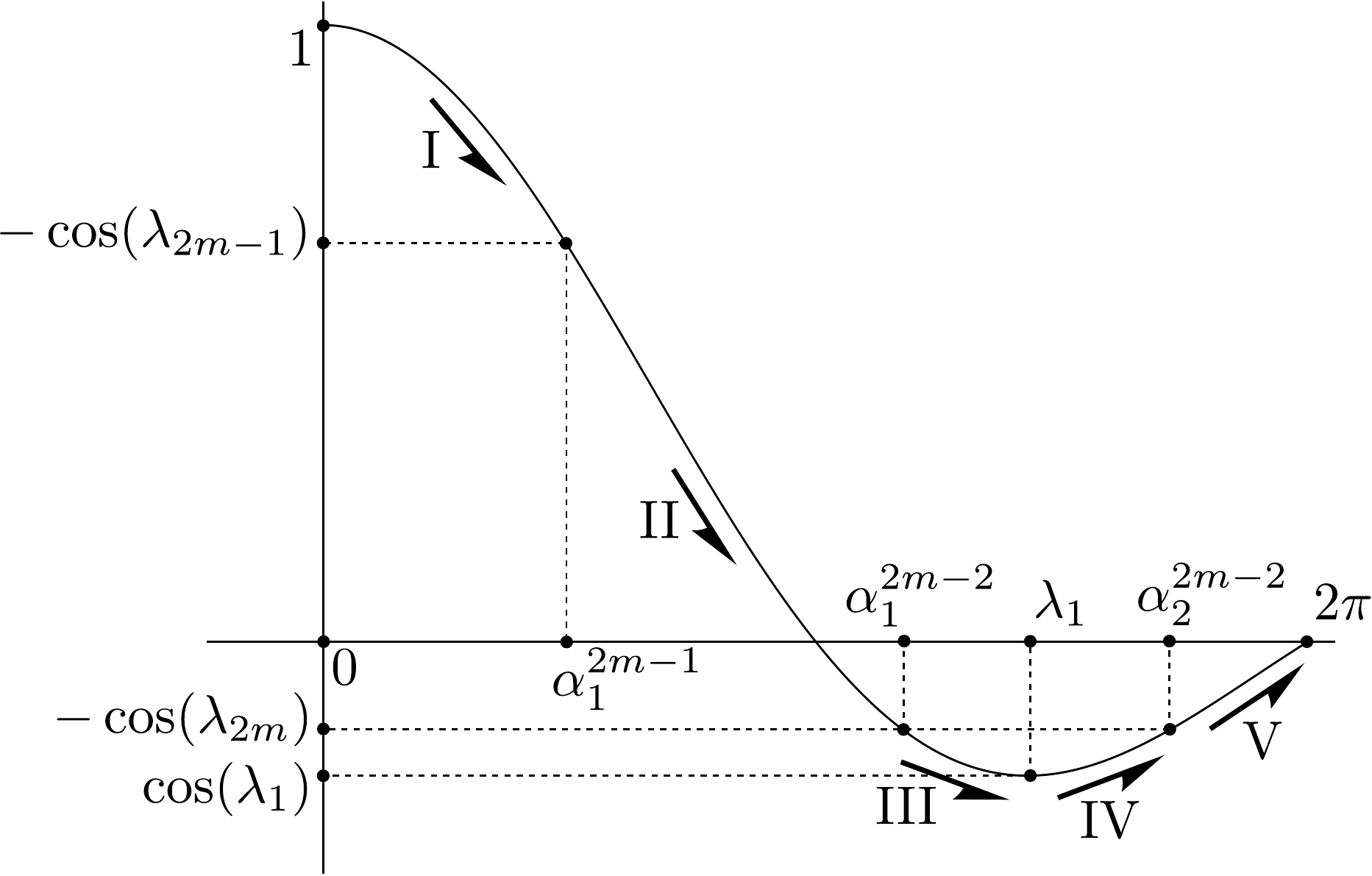} \label{fig:w2m-1a}}
\end{subfigure} \\
\vspace*{2ex}
\begin{subfigure}
[The corresponding values of $w(\alpha)$ which satisfy $G(w(\alpha),\alpha)=0$ with
$w(\pi) = (2m-1)\pi$]{
\includegraphics[width=9.5cm]{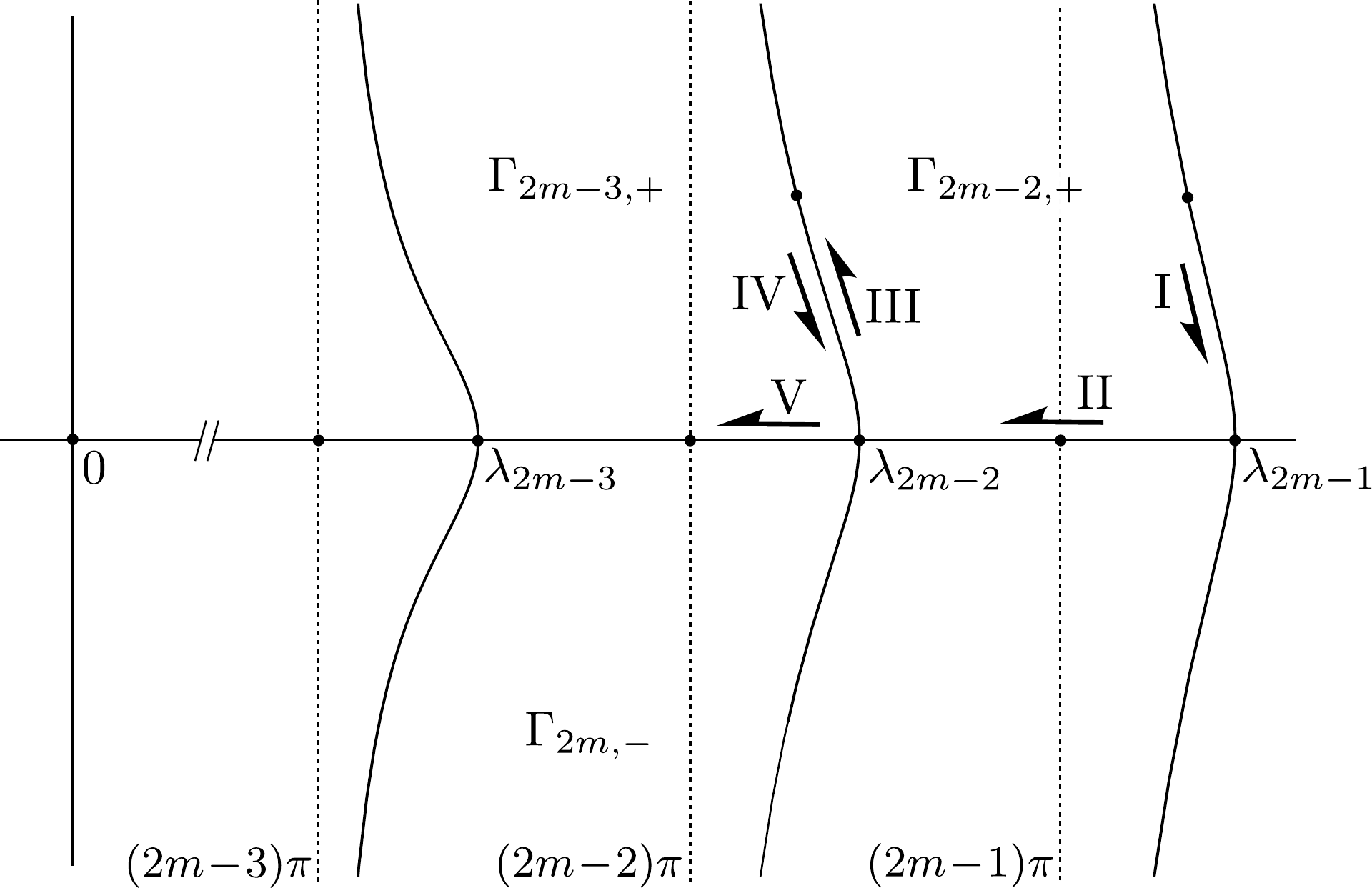} \label{fig:w2m-1b}}
\end{subfigure}
\caption{
The values $\sinc{(\alpha)}$ for $\alpha \in (0,2\pi)$  (\cref{fig:w2m-1a}) 
and
the corresponding values of $w(\alpha)$ which satisfy $G(w(\alpha),\alpha)=0$ with
$w(\pi) = (2m-1)\pi$ (\cref{fig:w2m-1b}).
In~\cref{fig:w2m-1b}, segment I represents $w(\alpha)$ for $\alpha \in (0,\al{2m-1}{1})$,
segment II represents $w(\alpha)$ for $\alpha \in (\al{2m-1}{1},\al{2m-2}{1})$, 
segment III represents $w(\alpha)$ for $\alpha \in (\al{2m-2}{1},\lambda_{1})$,
segment IV represents $w(\alpha)$ for $\alpha \in (\lambda_{1}, \al{2m-2}{2})$,
and finally segment V represents $w(\alpha)$ for $\alpha \in (\al{2m-2}{2},2\pi)$.}
\label{fig:w2m-1}
\end{center}
\end{figure}

\begin{rem}
Referring to~\cref{fig:w2m-1}, we provide a detailed description for
the behavior of $w(\alpha)$ defined in~\cref{lem:defw3} for 
$\alpha \in (0,2\pi)$.
\end{rem}

We present the principal result of this section in the following lemma.
\begin{lem}
\label{lem:defz2}
Suppose that $m\geq 2$ is an integer and that $H(z,\theta)$ is
as defined in~\cref{eq:defh1}. 
Suppose that $\al{2m-1}{1}$, $\al{2m-2}{1}$, and $\al{2m-2}{2}$
are as defined in~\cref{lem:defalphj}.
Suppose that $\theta_{1},\theta_{2}$, and $\theta_{3}$ are given by
\beq
\theta_{3} = 2-\frac{\al{2m-2}{2}}{\pi} \, ,\quad
\theta_{2} = 2-\frac{\al{2m-2}{1}}{\pi} \, , \quad
\theta_{1} = 2-\frac{\al{2m-1}{1}}{\pi} \, . 
\eeq
Suppose further that $D$ is the strip in the upper half plane with 
$0 < \text{Re}(\theta) < 2$ that includes the interval
$(0,2) \setminus \{ \theta_{1},\theta_{2},\theta_{3} \}$, i.e.
\beq
D = \{ \theta \in \C: \, 0 < \text{Re}(\theta) < 2 \, , \quad \text{Im}(\theta) \geq 0 \} 
\setminus \{ \theta_{1},\theta_{2},\theta_{3} \} \, .
\eeq
Then there exists a simply connected open set $D \subset V \subset \C$ and an
analytic function $z(\theta): V \to \C$ which satisfies $H(z(\theta),\theta) = 0$
for all $\theta \in V$ and $z(1) = 2m-1$.
\end{lem}

\subsubsubsection{Odd case, $z(1) = 1$}
In this section, we analyze the implicit functions which satisfy
$H(z,\theta) = 0$, with $z(1) = 1$.
The principal result of this section is~\cref{lem:defz3}.
The proofs of the results presented in this section are
analogous to the corresponding proofs in~\cref{sec:z12m}.
We present the statements of the theorem without proofs for brevity.

In the following lemma, we construct an analytic function $w(\alpha)$
which satisfies $G(w(\alpha),\alpha) = 0$ with $w(1) = \pi$.

\begin{lem}
\label{lem:defw4}
Suppose that $G(w,\alpha)$
is as defined in~\cref{eq:defG}.
Suppose the regions $\Gamma_{0,+}, \Gamma_{0,-}$
are as defined in~\cref{lem:gammajinv}.
Suppose that $\al{1}{1}$ 
is as defined in~\cref{lem:defalphj}.
As before, let $\overline{A}$ denote the closure of the set $A$.
Furthermore, suppose that $D$ is the strip in the lower half plane with
$0<\text{Re}(\alpha)<2\pi$, i.e. 
\beq
D = \{ \alpha \in \C: \, 0<\text{Re}(\alpha)<2\pi \, ,\quad 
\text{Im}(\alpha) < 0 \} \,.
\eeq
Suppose that $D_{1}$ is the region $\overline{D} \cap \overline{\Gamma}_{0,-}$ and
$D_{2}$ is the region $\overline{D} \setminus D_{1}$.
Suppose finally that $w(\alpha): \overline{D} \to \C$ is defined by
\beq
w(\alpha) = \begin{cases}
\sinc^{-1}_{0,+} (-\sinc{(\alpha)}) & \quad \alpha \in 
D_{1} \\
\sinc^{-1}_{0,-} (-\sinc{(\alpha)}) & \quad \alpha \in D_{2} \, .
\end{cases}
\eeq
Then 
for all $\alpha \in D$,
$w(\alpha)$ satisfies $G(w(\alpha),\alpha) = 0$ and is an analytic
function for $\alpha \in D$.
Moreover, $w(\pi) = \pi$.
\end{lem}

\begin{figure}[h!]
\begin{center}
\begin{subfigure}[
The values $\sinc{(\alpha)}$ for $\alpha \in (0,2\pi)$  
]
{
\includegraphics[width=7.4cm]{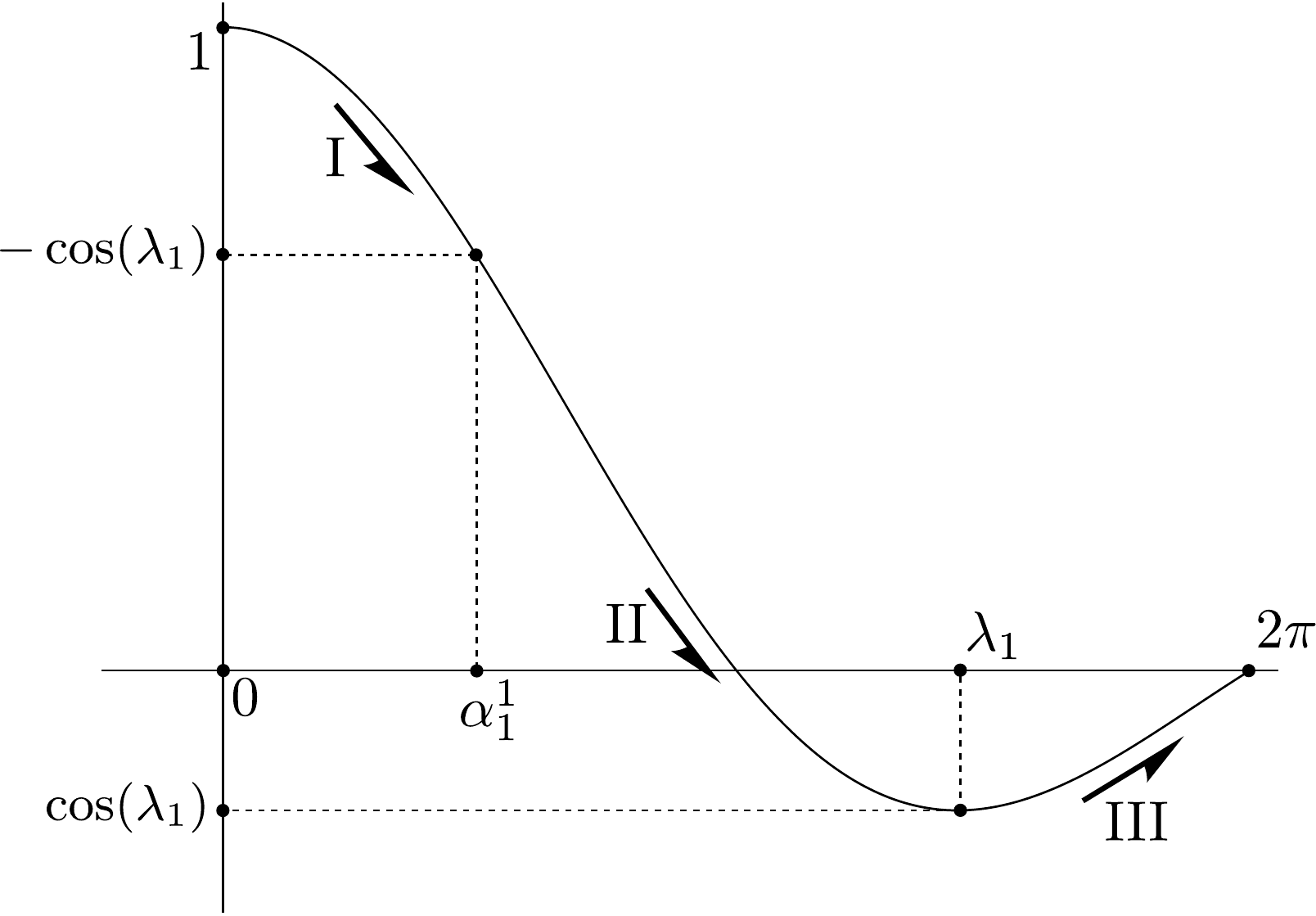} \label{fig:w1a}}
\end{subfigure} \\
\vspace*{2ex}
\begin{subfigure}[
The corresponding values of $w(\alpha)$ which satisfy $G(w(\alpha),\alpha)=0$ with
$w(\pi) = \pi$ 
]
{
\includegraphics[width=6.5cm]{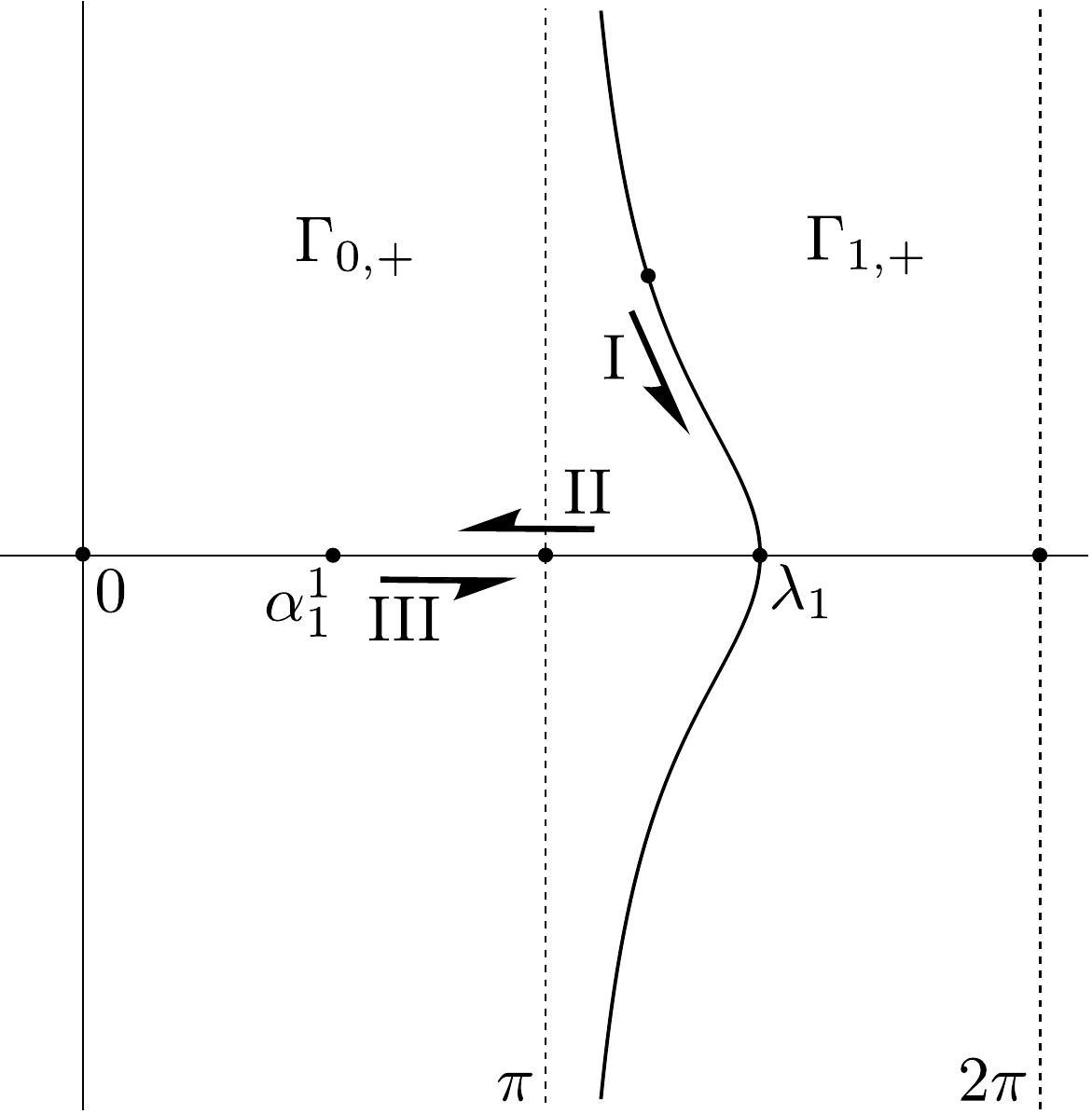} \label{fig:w1b}
}
\end{subfigure}
\caption{
The values $\sinc{(\alpha)}$ for $\alpha \in (0,2\pi)$  (\cref{fig:w1a}) and
the corresponding values of $w(\alpha)$ which satisfy $G(w(\alpha),\alpha)=0$ with
$w(\pi) = \pi$ (\cref{fig:w1a}).
In~\cref{fig:w1b}, segment I represents $w(\alpha)$ for $\alpha \in (0,\al{2m-1}{1})$,
segment II represents $w(\alpha)$ for $\alpha \in (\al{2m-1}{1},\al{2m-2}{1})$, 
and finally segment III represents $w(\alpha)$ for $\alpha \in (\al{2m-2}{1},\lambda_{1})$}
\label{fig:w1}
\end{center}
\end{figure}

\begin{rem}
Referring to~\cref{fig:w1}, we provide a detailed description for
the behavior of $w(\alpha)$ defined in~\cref{lem:defw4} for 
$\alpha \in (0,2\pi)$.
\end{rem}

We present the principal result of this section in the following lemma.
\begin{lem}
\label{lem:defz3}
Suppose that $H(z,\theta)$ is as defined in~\cref{eq:defh1}. 
Suppose that $\al{1}{1}$ is as defined in~\cref{lem:defalphj}.
Furthermore, suppose that $\theta_{1}$ are given by
\beq
\theta_{1} = 2-\frac{\al{1}{1}}{\pi} \, ,
\eeq
Suppose further that $D$ is the strip in the upper half plane with 
$0 < \text{Re}(\theta) < 2$ that includes the interval
$(0,2) \setminus \{ \theta_{1} \}$, i.e.
\beq
D = \{ \theta \in \C: \, 0 < \text{Re}(\theta) < 2 \, , \quad \text{Im}(\theta) \geq 0 \} 
\setminus \{ \theta_{1} \} \, .
\eeq
Then there exists a simply connected open set $D \subset V \subset \C$ 
(see~\cref{fig:vdom2}) and an
analytic function $z(\theta): V \to \C$ which satisfies $H(z(\theta),\theta) = 0$
for all $\theta \in V$ and $z(1) = 1$.
\end{lem}

\begin{figure}[h!]
\begin{center}
\includegraphics[width=3.5cm]{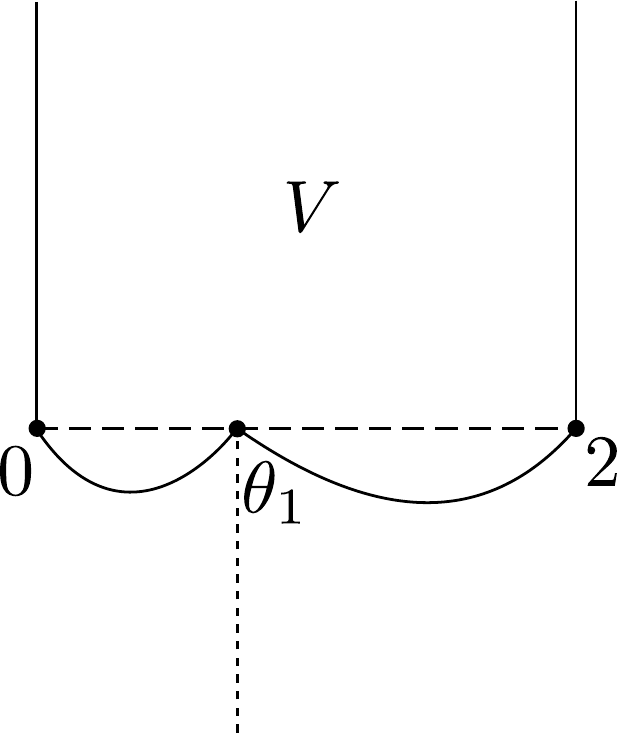}
\caption{An illustrative region of analyticity $V$ of the  
function $z(\theta)$, which satisfies $H(z(\theta),\theta)=0$ with $z(1) = 1$.}
\label{fig:vdom2}
\end{center}
\end{figure}

\subsubsection{Analysis of implicit function $z$ in~\cref{eq:implfun2app0} 
\label{sec:appatone2}}
In this section, we investigate the implicit functions which satisfy
\beq
H(z,\theta) = z\sint - \sinzt = 0 \, . \label{eq:defh2}
\eeq
The analysis for the implicit functions $z(\theta)$ which satisfy
$H(z,\theta) = 0$ is similar to the analysis of the analogous
implicit functions in~\cref{sec:appatone1}.
For conciseness, we present the statements of the theorems in this section
and omit the proofs.

We first state the connection between $\sinc{(z)}$ and the function
$H(z,\theta)$ defined in~\cref{eq:defh2}.

\begin{lem}
\label{lem:htog2}
Suppose that $G: \C \times \C \to \C$ is the entire function defined by
\beq
\label{eq:defG2}
G(w,\alpha) = \sinc{(w)} - \sinc{(\alpha)} \, .
\eeq
Then $G(w, \alpha) = 0$ if and only if $H(z,\theta) = 0$ 
where $z = \frac{w}{\alpha}$ and $\theta = \frac{\alpha}{\pi}$.
\end{lem} 

In the following lemma, we discuss the solutions $\alpha \in [0,2\pi]$ 
of $\sinc{(\alpha)} = \cos{(\lambda_{j})}$.

\begin{lem}
\label{lem:defalphj2}
Suppose $j$ is a positive integer and $\lambda_{j}$ 
is defined in~\cref{lem:sincderzeros}.
\begin{itemize}
\item Case 1, $j$ is even: the equation 
$\sinc{(\alpha)} = \cos{(\lambda_{j})}$ 
has only one solution $\bjo$
on the interval $\alpha \in [0,2\pi]$ 
where $0 < \bjo < \pi$ 
(see~\cref{fig:sinccurve2}).
\item Case 2, $j$ is odd: the equation 
$\sinc{(\alpha)} = -\cos{(\lambda_{j})}$ 
has only two solutions $\bjo$, and $\bjt$ 
on the interval
$\alpha \in [0,2\pi]$, where $\pi<\bjo<\lambda_{1} < \bjt < 2\pi$
(see~\cref{fig:sinccurve2}).
\end{itemize}
\end{lem}

\begin{figure}[h!]
\begin{center}
\includegraphics[width=7cm]{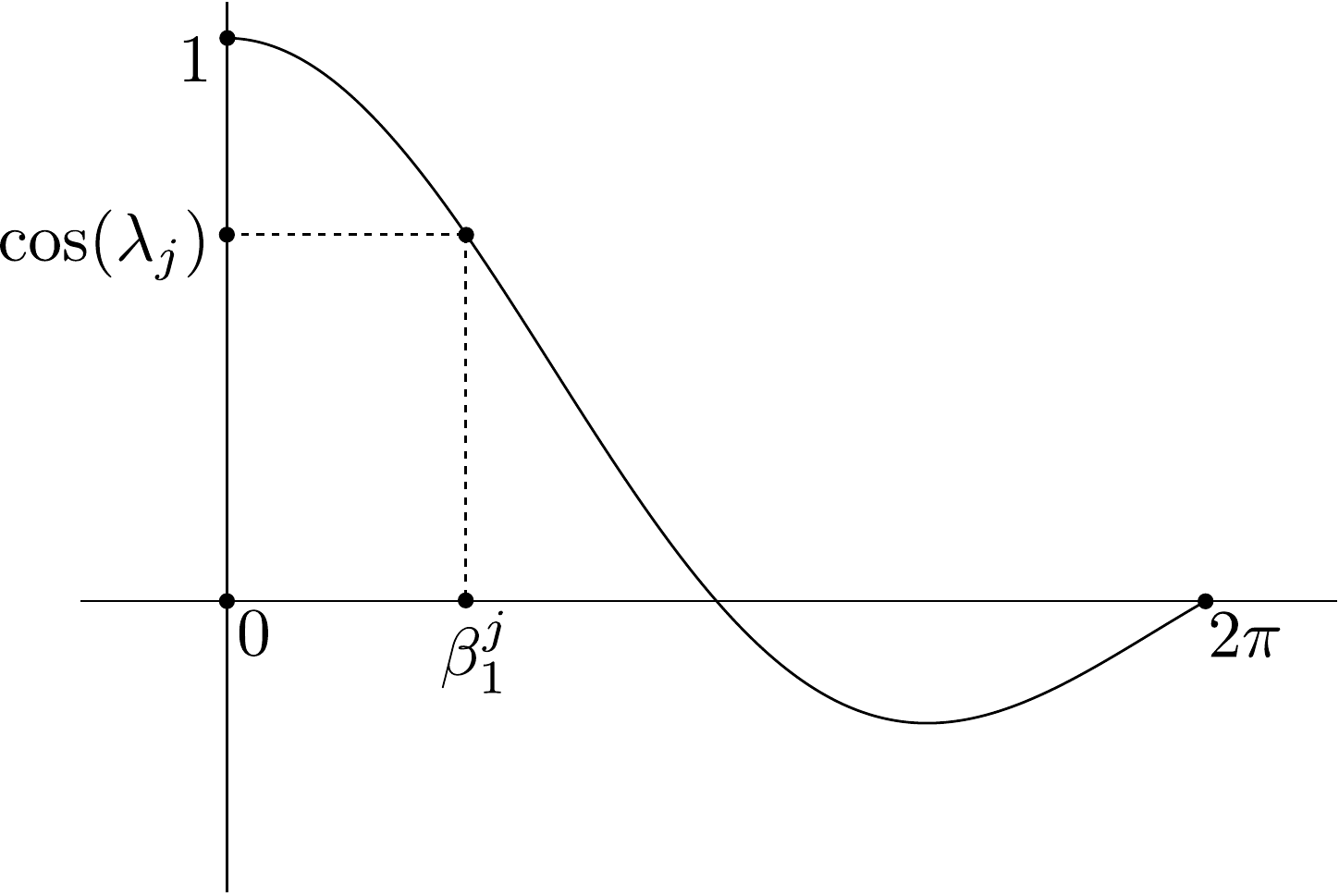}
\includegraphics[width=7cm]{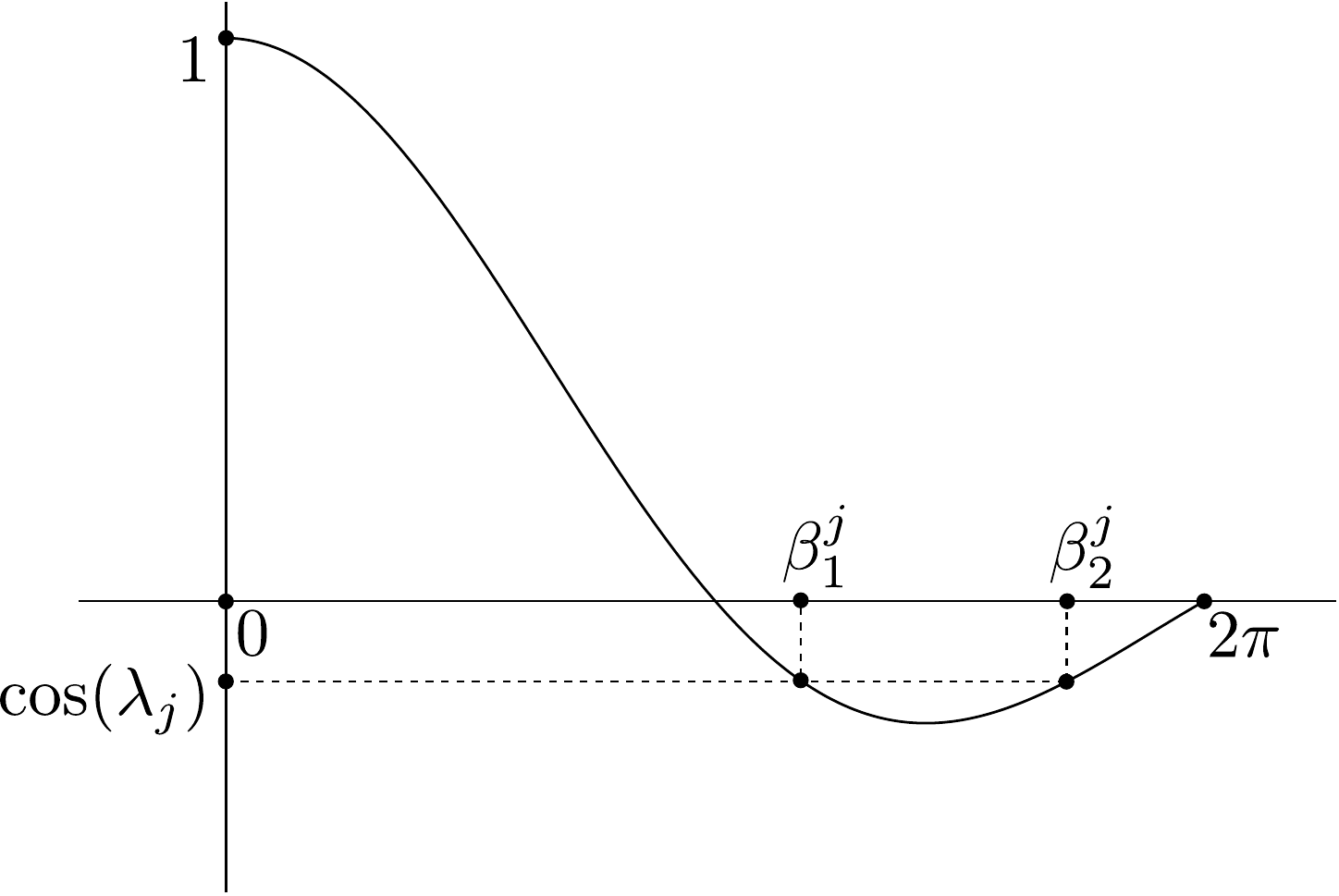}
\caption{Solutions of  $\sinc{(\alpha)} = \cos{(\lambda_{j})}$. 
Case 1, $j$ is even ($\cos{(\lambda_{j})}>0$) (left), and case 2, 
$j$ is odd ($\cos{(\lambda_{j})}<0$) (right).}
\label{fig:sinccurve2}
\end{center}
\end{figure}

As before, the analysis of the implicit functions $w(\alpha)$ 
which satisfy $G(w,\alpha) = 0$ (see~\cref{eq:defG2}), and the analogous functions
$z(\theta)$ which satisfy $H(z,\theta)$ (see~\cref{eq:defh2}), is
split into three cases.
In~\cref{sec:z2mm1fun2}, we analyze the functions implicit functions
$z(\theta)$ for the case $z(1) = 2m-1$, where $m$ is a positive integer, 
in~\cref{sec:z2mfun2}, we analyze the implicit functions $z(\theta)$
for the case $z(1) = 2m$, where $m\geq 2$ is an integer, and finally
in~\cref{sec:z1fun2}, we analyze the implicit function $z(\theta)$ 
for the case $z(1) = 2$. 

\subsubsubsection{Odd case, $z(1) = 2m-1$ \label{sec:z2mm1fun2}}
In this section, we analyze the implicit functions which satisfy
$H(z,\theta) = 0$, with $z(1) = 2m-1$, where $m$ 
is a positive integer. 
The principal result of this section is~\cref{lem:defz4}.

In the following lemma, we construct an analytic function $w(\alpha)$
which satisfies $G(w(\alpha),\alpha) = 0$ with $w(1) = (2m-1)\pi$.

\begin{lem}
\label{lem:defw5}
Suppose that $m$ is a positive integer, and that $G(w,\alpha)$
is as defined in~\cref{eq:defG2}.
Suppose the regions $\Gamma_{j,+}, \Gamma_{j,-}$, $j=0,1,\ldots$ 
are as defined in~\cref{lem:gammajinv}.
Suppose that $\bet{2m-2}{1}$, $\bet{2m-1}{1}$, and $\bet{2m-1}{2}$
are as defined in~\cref{lem:defalphj2}.
As before, let $\overline{A}$ denote the closure of the set $A$.
Furthermore, suppose that $D$ is the strip in the upper half plane with
$0<\text{Re}(\alpha)<2\pi$, i.e. 
\beq
D = \{ \alpha \in \C: \, 0<\text{Re}(\alpha)<2\pi \, ,\quad 
\text{Im}(\alpha) > 0 \} \,.
\eeq
Suppose that $D_{1}$ is the region $\overline{D} \cap \overline{\Gamma}_{0,+}$ and
$D_{2}$ is the region $\overline{D} \setminus \overline{\Gamma}_{0,+}$ (see~\cref{fig:wdom2}).

\begin{figure}[h!]
\begin{center}
\includegraphics[width=5cm]{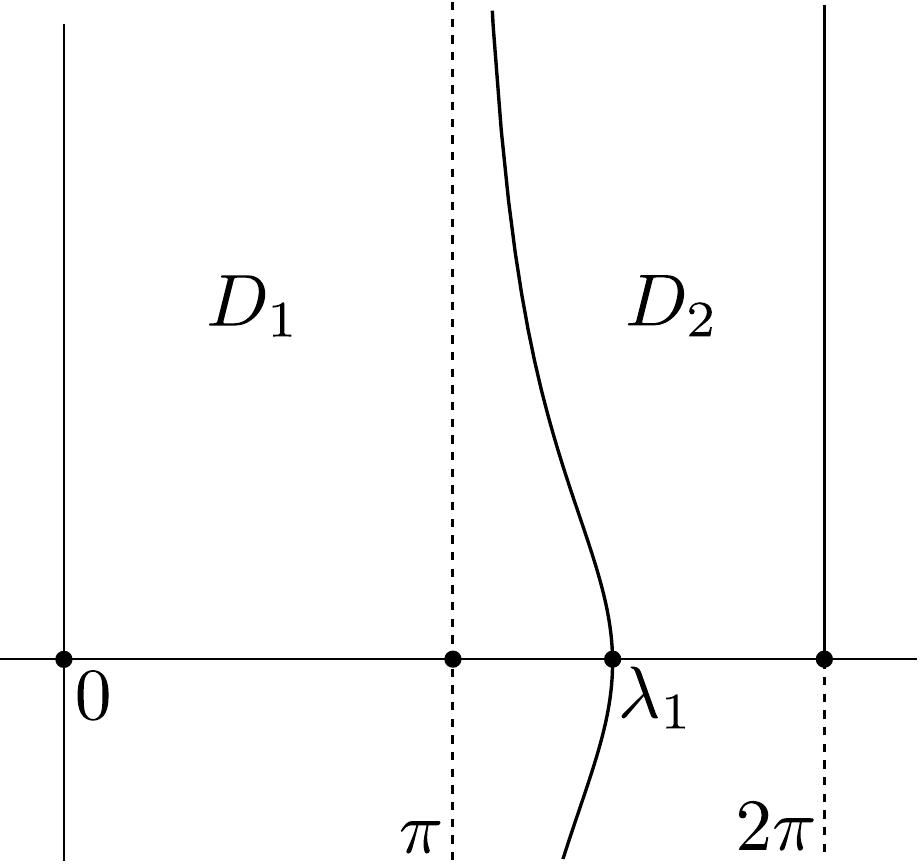}
\caption{The regions $D_{1} = 
\overline{D} \cap \overline{\Gamma}_{0,+}$
and $D_{2} = \overline{D} \setminus
\overline{\Gamma}_{0,+}$.}
\label{fig:wdom2}
\end{center}
\end{figure}

Suppose finally that $w(\alpha): \overline{D} \to \C$ is defined by
\beq
w(\alpha) = \begin{cases}
\sinc^{-1}_{2m-2,-} (\sinc{(\alpha)}) & \quad \alpha \in 
D_{1} \\
\sinc^{-1}_{2m-1,-} (\sinc{(\alpha)}) & \quad \alpha \in D_{2} \, .
\end{cases}
\eeq
Then 
for all $\alpha \in D$,
$w(\alpha)$ satisfies $G(w(\alpha),\alpha) = 0$ and is an analytic
function for $\alpha \in D$.
Moreover, $w(\pi) = (2m-1) \pi$.
\end{lem}

\begin{figure}[h!]
\begin{center}
\begin{subfigure}[The values $\sinc{(\alpha)}$ for $\alpha \in (0,2\pi)$]
{\includegraphics[width=8.5cm]{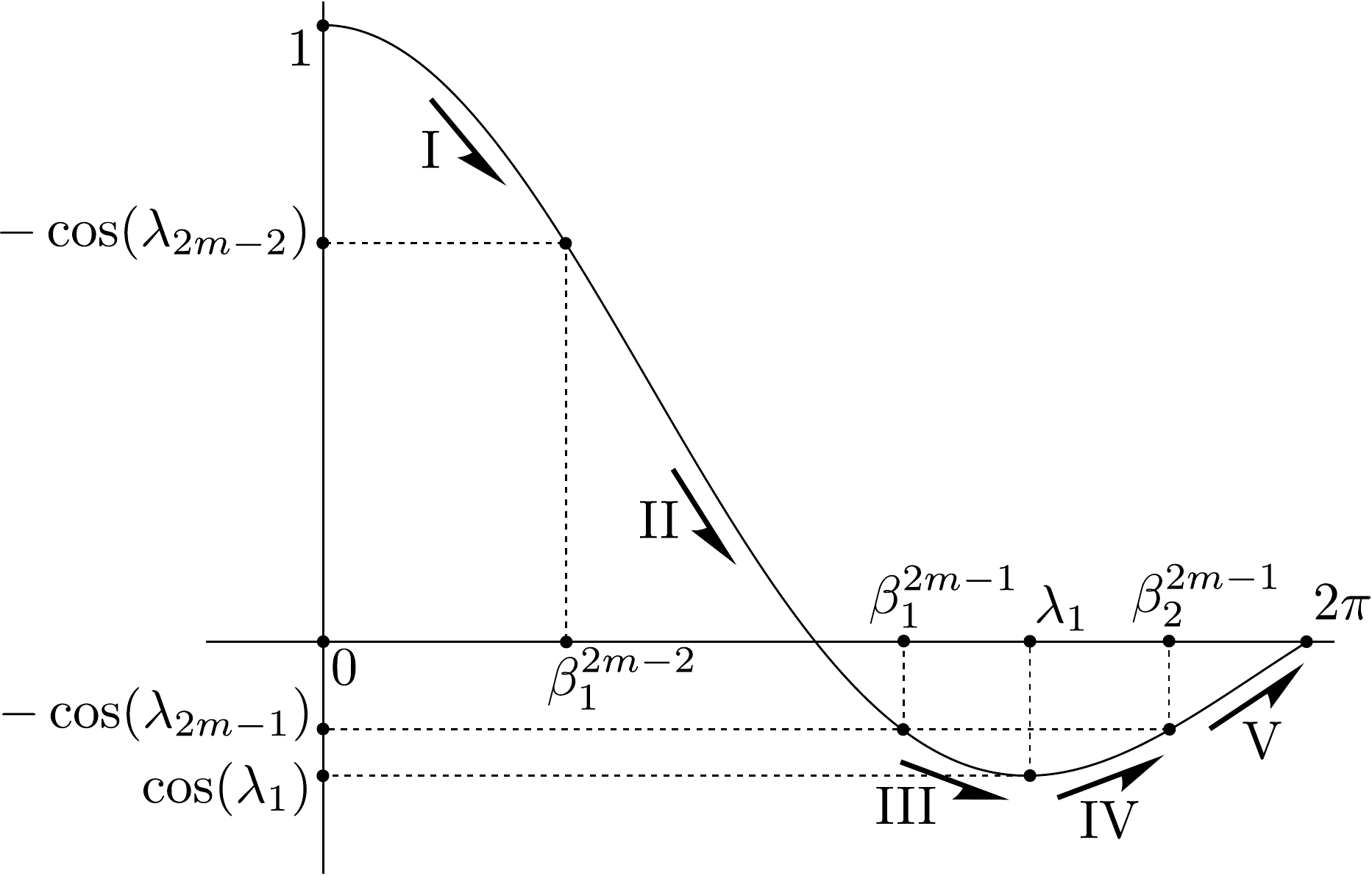} \label{fig:w2m-1fun2a}}
\end{subfigure} \\
\vspace*{2ex}
\begin{subfigure}[
The corresponding values of $w(\alpha)$ which satisfy $G(w(\alpha),\alpha)=0$ with
$w(\pi) = (2m-1)\pi$]{\includegraphics[width=9.5cm]{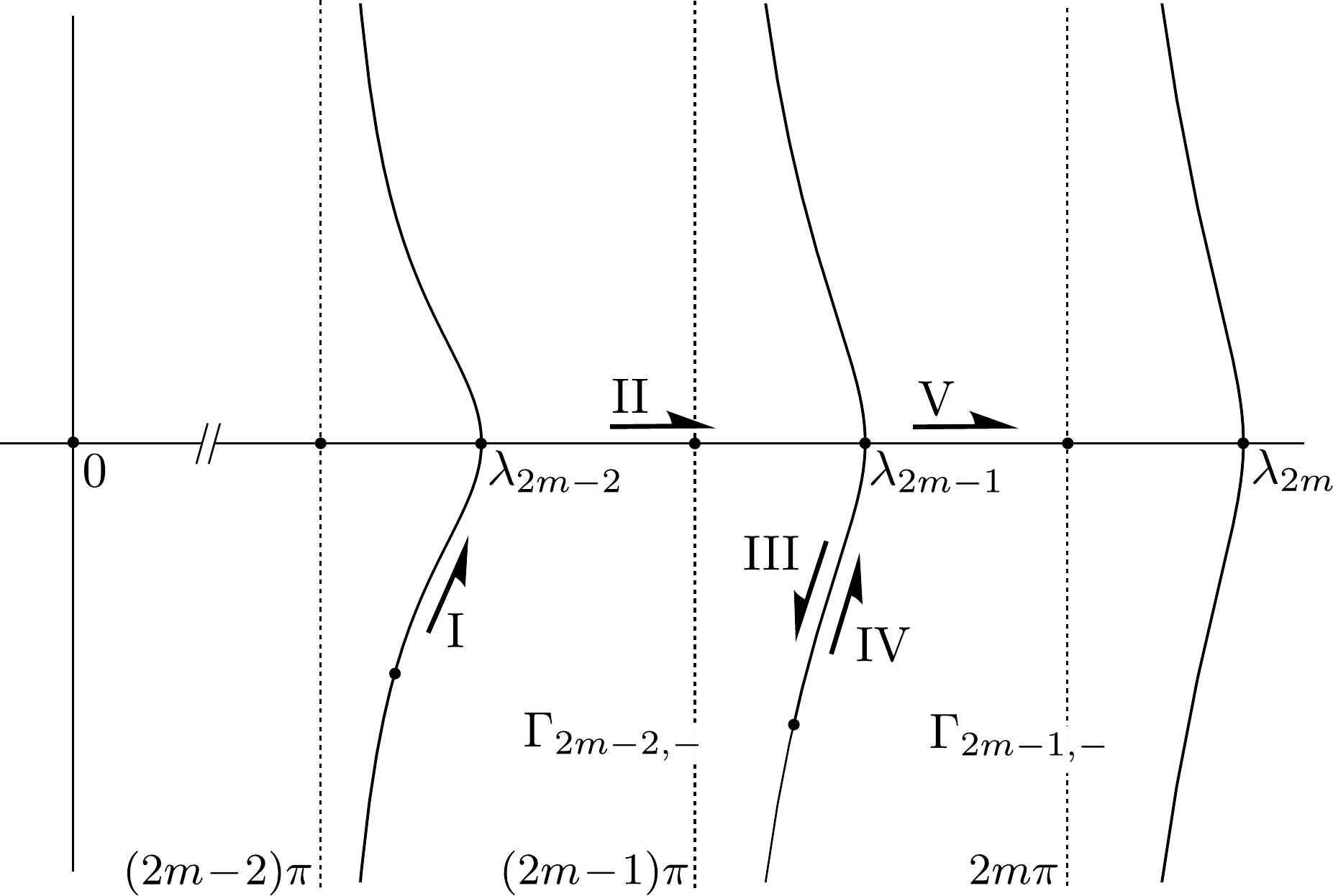} \label{fig:w2m-1fun2b}}
\end{subfigure}
\caption{
The values $\sinc{(\alpha)}$ for $\alpha \in (0,2\pi)$  (\cref{fig:w2m-1fun2a}) 
and
the corresponding values of $w(\alpha)$ which satisfy $G(w(\alpha),\alpha)=0$ with
$w(\pi) = (2m-1)\pi$ (\cref{fig:w2m-1fun2b}).
In~\cref{fig:w2m-1fun2b}, segment I represents $w(\alpha)$ for $\alpha \in (0,\bet{2m-2}{1})$,
segment II represents $w(\alpha)$ for $\alpha \in (\bet{2m-2}{1},\bet{2m-1}{1})$, 
segment III represents $w(\alpha)$ for $\alpha \in (\bet{2m-1}{1},\lambda_{1})$,
segment IV represents $w(\alpha)$ for $\alpha \in (\lambda_{1}, \bet{2m-1}{2})$,
and finally segment V represents $w(\alpha)$ for $\alpha \in (\bet{2m-1}{2},2\pi)$.}
\label{fig:w2m-1fun2}
\end{center}
\end{figure}

\begin{rem}
Referring to~\cref{fig:w2m-1fun2}, we provide a detailed description for
the behavior of $w(\alpha)$ defined in~\cref{lem:defw5} for 
$\alpha \in (0,2\pi)$.
\end{rem}

We present the principal result of this section in the following lemma.
\begin{lem}
\label{lem:defz4}
Suppose that $m$ is a positive integer and that $H(z,\theta)$ is
as defined in~\cref{eq:defh2}. 
Suppose that $\bet{2m-2}{1}$, $\bet{2m-1}{1}$, and $\bet{2m-1}{2}$
are as defined in~\cref{lem:defalphj2}.
Suppose that $\theta_{1},\theta_{2}$, and $\theta_{3}$ are given by
\beq
\theta_{1} = \frac{\bet{2m-2}{1}}{\pi} \, , \quad
\theta_{2} = \frac{\bet{2m-1}{1}}{\pi} \, , \quad
\theta_{3} = \frac{\bet{2m-1}{2}}{\pi} \, .
\eeq
Suppose further that $D$ is the strip in the upper half plane with 
$0 < \text{Re}(\theta) < 2$ that includes the interval
$(0,2) \setminus \{ \theta_{1},\theta_{2},\theta_{3} \}$, i.e.
\beq
D = \{ \theta \in \C: \, 0 < \text{Re}(\theta) < 2 \, , 
\quad \text{Im}(\theta) \geq 0 \} 
\setminus \{ \theta_{1},\theta_{2},\theta_{3} \} \, .
\eeq
Then there exists a simply connected open set $D \subset V \subset \C$ and an
analytic function $z(\theta): V \to \C$ which satisfies $H(z(\theta),\theta) = 0$
for all $\theta \in V$ and $z(1) = 2m-1$.
\end{lem}

\subsubsubsection{Even case, $z(1) = 2m$, $m\neq 1$ \label{sec:z2mfun2}}
In this section, we analyze the implicit functions which satisfy
$H(z,\theta) = 0$, with $z(1) = 2m$, where $m\geq 2$ 
is an integer. 
The principal result of this section is~\cref{lem:defz5}.

In the following lemma, we construct an analytic function $w(\alpha)$
which satisfies $G(w(\alpha),\alpha) = 0$ with $w(1) = 2m\pi$.

\begin{lem}
\label{lem:defw6}
Suppose that $m \geq 2$ is an integer, and that $G(w,\alpha)$
is as defined in~\cref{eq:defG2}.
Suppose the regions $\Gamma_{j,+}, \Gamma_{j,-}$, $j=0,1,\ldots$ 
are as defined in~\cref{lem:gammajinv}.
Suppose that $\bet{2m}{1}$, $\bet{2m-1}{1}$, and $\bet{2m-1}{2}$
are as defined in~\cref{lem:defalphj2}.
As before, let $\overline{A}$ denote the closure of the set $A$.
Furthermore, suppose that $D$ is the strip in the upper half plane with
$0<\text{Re}(\alpha)<2\pi$, i.e. 
\beq
D = \{ \alpha \in \C: \, 0<\text{Re}(\alpha)<2\pi \, ,\quad 
\text{Im}(\alpha) > 0 \} \,.
\eeq
Suppose that $D_{1}$ is the region $\overline{D} \cap \overline{\Gamma}_{0,+}$ and
$D_{2}$ is the region $\overline{D} \setminus D_{1}$ (see~\cref{fig:wdom2}).
Suppose finally that $w(\alpha): \overline{D} \to \C$ is defined by
\beq
w(\alpha) = \begin{cases}
\sinc^{-1}_{2m-1,+} (\sinc{(\alpha)}) & \quad \alpha \in 
D_{1} \\
\sinc^{-1}_{2m-2,+} (\sinc{(\alpha)}) & \quad \alpha \in D_{2} \, .
\end{cases}
\eeq
Then 
for all $\alpha \in D$,
$w(\alpha)$ satisfies $G(w(\alpha),\alpha) = 0$ and is an analytic
function for $\alpha \in D$.
Moreover, $w(\pi) = 2m \pi$.
\end{lem}

\begin{figure}[h!]
\begin{center}
\begin{subfigure}
[The values $\sinc{(\alpha)}$ for $\alpha \in (0,2\pi)$]
{\includegraphics[width=8.5cm]{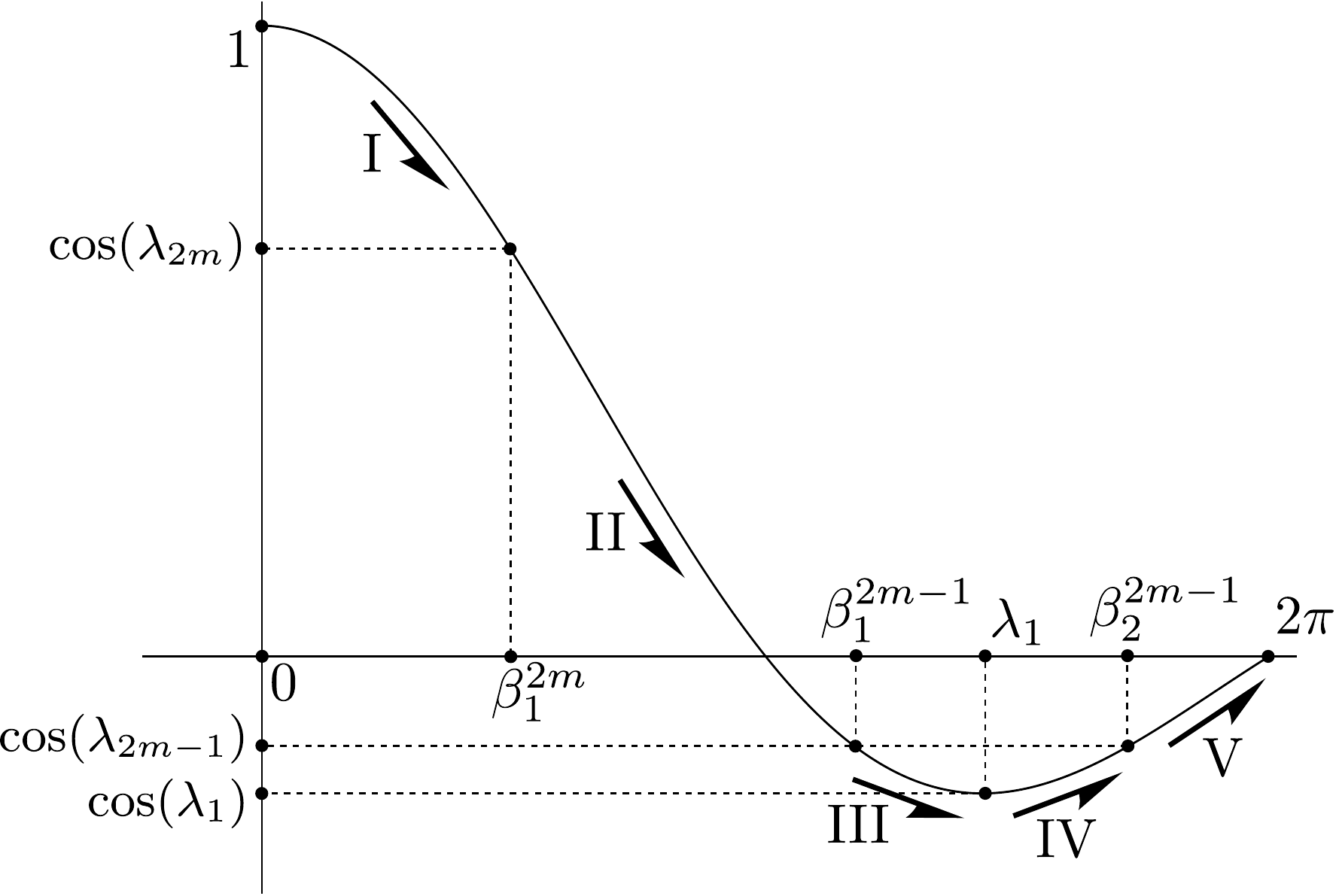} \label{fig:w2mfun2a}}
\end{subfigure} \\
\vspace*{2ex}
\begin{subfigure}[
The corresponding values of $w(\alpha)$ which satisfy $G(w(\alpha),\alpha)=0$ with
$w(\pi) = 2m\pi$]
{\includegraphics[width=9.5cm]{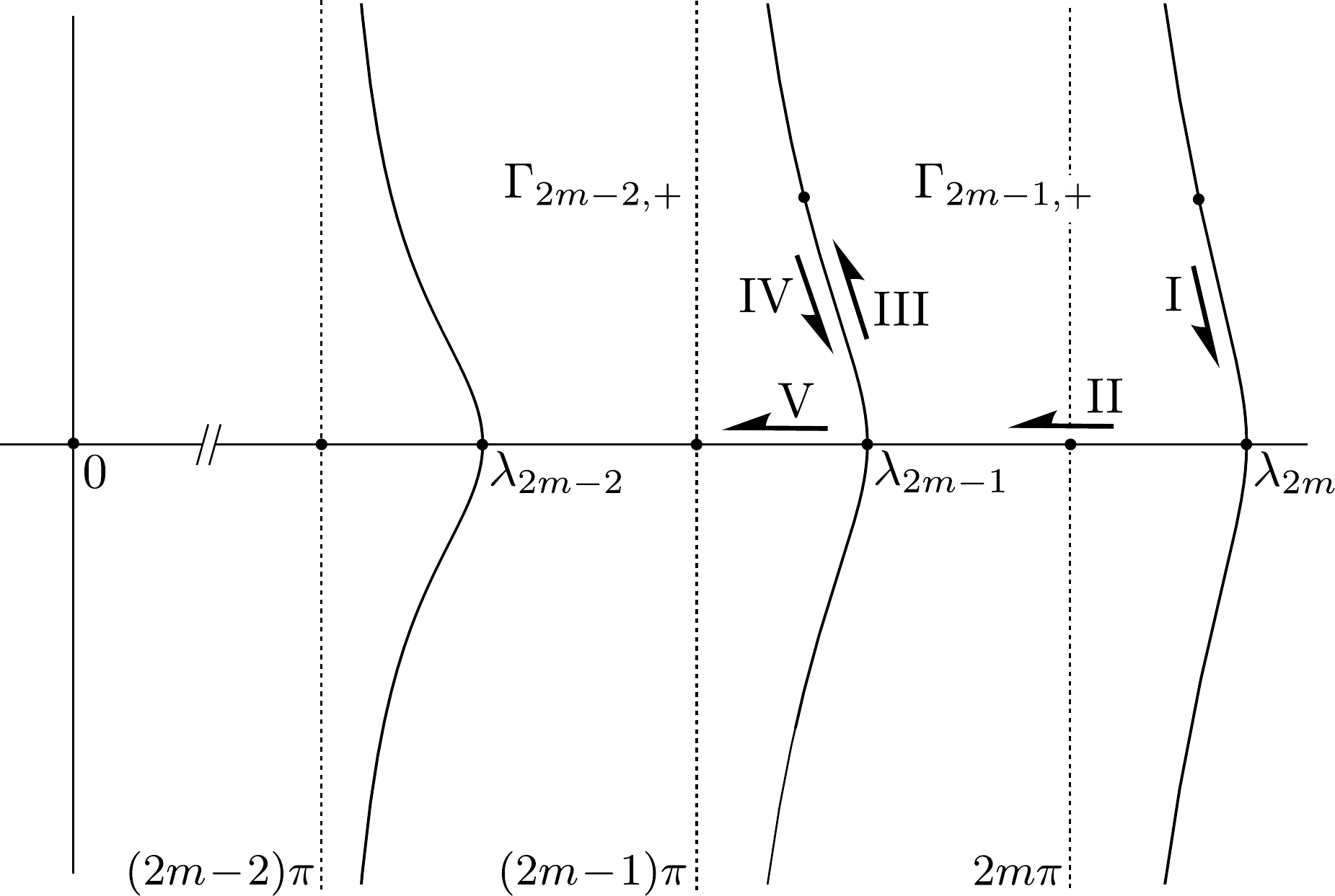} \label{fig:w2mfun2b} }
\end{subfigure}
\caption{
The values $\sinc{(\alpha)}$ for $\alpha \in (0,2\pi)$  (\cref{fig:w2mfun2a}) and
the corresponding values of $w(\alpha)$ which satisfy $G(w(\alpha),\alpha)=0$ with
$w(\pi) = 2m\pi$ (\cref{fig:w2mfun2b}).
In~\cref{fig:w2mfun2b}, segment I represents $w(\alpha)$ for $\alpha \in (0,\bet{2m}{1})$,
segment II represents $w(\alpha)$ for $\alpha \in (\bet{2m}{1},\bet{2m-1}{1})$, 
segment III represents $w(\alpha)$ for $\alpha \in (\bet{2m-1}{1},\lambda_{1})$,
segment IV represents $w(\alpha)$ for $\alpha \in (\lambda_{1}, \bet{2m-1}{2})$,
and finally segment V represents $w(\alpha)$ for $\alpha \in (\bet{2m-1}{2},2\pi)$.}
\label{fig:w2mfun2}
\end{center}
\end{figure}

\begin{rem}
Referring to~\cref{fig:w2mfun2}, we provide a detailed description for
the behavior of $w(\alpha)$ defined in~\cref{lem:defw6} for 
$\alpha \in (0,2\pi)$.
\end{rem}

We present the principal result of this section in the following lemma.
\begin{lem}
\label{lem:defz5}
Suppose that $m \geq 2$ is an integer and that $H(z,\theta)$ is
as defined in~\cref{eq:defh2}. 
Suppose that $\bet{2m}{1}$, $\bet{2m-1}{1}$, and $\bet{2m-1}{2}$
are as defined in~\cref{lem:defalphj2}.
Suppose that $\theta_{1},\theta_{2}$, and $\theta_{3}$ are given by
\beq
\theta_{1} = \frac{\bet{2m}{1}}{\pi} \, , \quad
\theta_{2} = \frac{\bet{2m-1}{1}}{\pi} \, , \quad
\theta_{3} = \frac{\bet{2m-1}{2}}{\pi} \, .
\eeq
Suppose further that $D$ is the strip in the upper half plane with 
$0 < \text{Re}(\theta) < 2$ that includes the interval
$(0,2) \setminus \{ \theta_{1},\theta_{2},\theta_{3} \}$, i.e.
\beq
D = \{ \theta \in \C: \, 0 < \text{Re}(\theta) < 2 \, , 
\quad \text{Im}(\theta) \geq 0 \} 
\setminus \{ \theta_{1},\theta_{2},\theta_{3} \} \, .
\eeq
Then there exists a simply connected open set $D \subset V \subset \C$ and an
analytic function $z(\theta): V \to \C$ which satisfies $H(z(\theta),\theta) = 0$
for all $\theta \in V$ and $z(1) = 2m$.
\end{lem}

\subsubsubsection{Even case, $z(1) = 2$ \label{sec:z1fun2}}
In this section, we analyze the implicit functions which satisfy
$H(z,\theta) = 0$, with $z(1) = 2$.
The principal result of this section is~\cref{lem:defz6}.

In the following lemma, we construct an analytic function $w(\alpha)$
which satisfies $G(w(\alpha),\alpha) = 0$ with $w(1) = 2\pi$.

\begin{lem}
\label{lem:defw7}
Suppose that $G(w,\alpha)$
is as defined in~\cref{eq:defG2}.
Suppose the regions $\Gamma_{0,+}, \Gamma_{0,-}$
are as defined in~\cref{lem:gammajinv}.
Suppose that $\bet{2}{1}$ 
is as defined in~\cref{lem:defalphj2}.
As before, let $\overline{A}$ denote the closure of the set $A$.
Furthermore, suppose that $D$ is the strip in the upper half plane with
$0<\text{Re}(\alpha)<2\pi$, i.e. 
\beq
D = \{ \alpha \in \C: \, 0<\text{Re}(\alpha)<2\pi \, ,\quad 
\text{Im}(\alpha) > 0 \} \,.
\eeq
Suppose that $D_{1}$ is the region $\overline{D} \cap \overline{\Gamma}_{0,+}$ and
$D_{2}$ is the region $\overline{D} \setminus D_{1}$.
Suppose finally that $w(\alpha): \overline{D} \to \C$ is defined by
\beq
w(\alpha) = \begin{cases}
\sinc^{-1}_{1,+} (\sinc{(\alpha)}) & \quad \alpha \in 
D_{1} \\
\sinc^{-1}_{0,+} (\sinc{(\alpha)}) & \quad \alpha \in D_{2} \, .
\end{cases}
\eeq
Then 
for all $\alpha \in D$,
$w(\alpha)$ satisfies $G(w(\alpha),\alpha) = 0$ and is an analytic
function for $\alpha \in D$.
Moreover, $w(\pi) = 2\pi$.
\end{lem}

\begin{figure}[h!]
\begin{center}
\begin{subfigure}[
The values $\sinc{(\alpha)}$ for $\alpha \in (0,2\pi)$]
{\includegraphics[width=7.4cm]{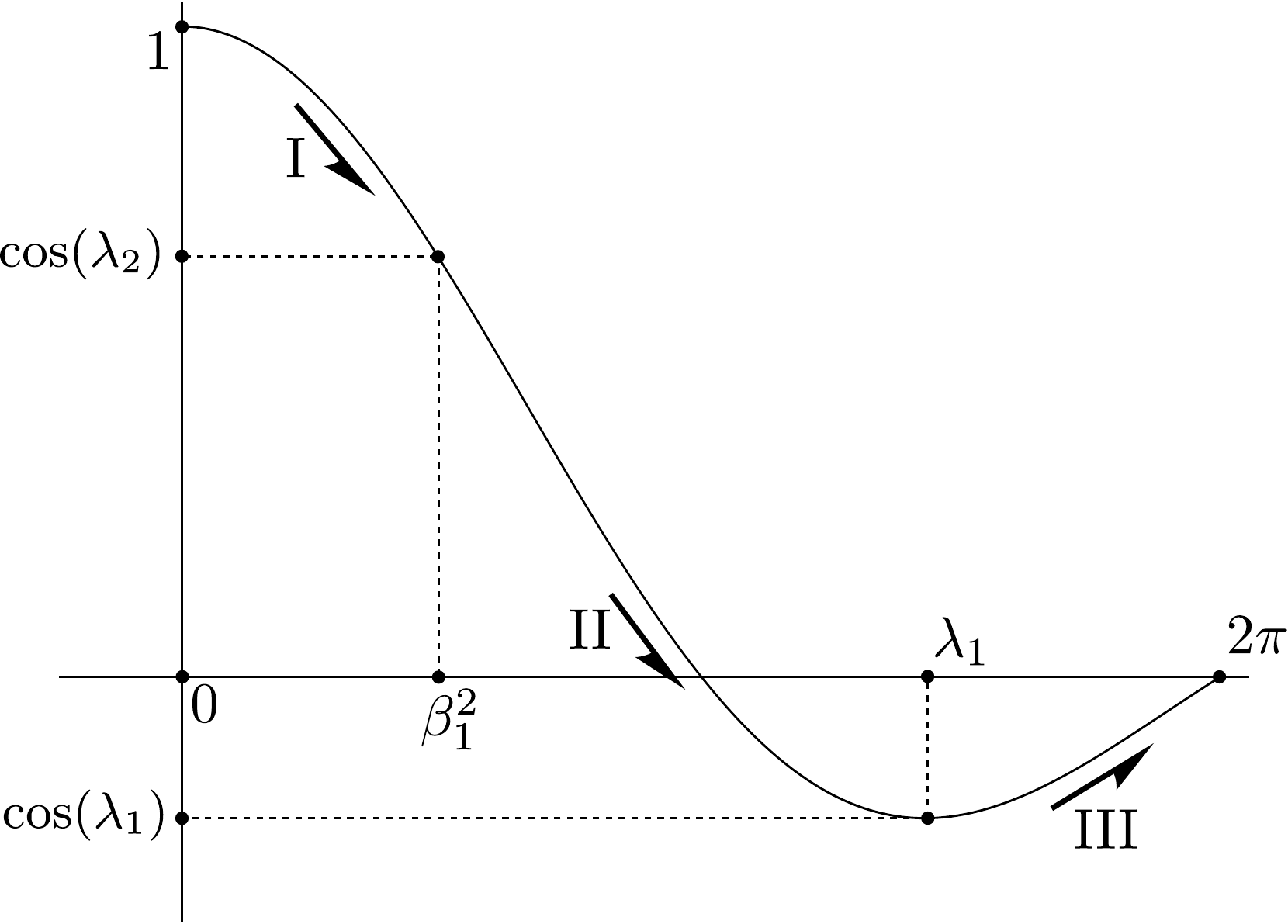} \label{fig:w2fun2a}}
\end{subfigure} \\
\vspace*{2ex}
\begin{subfigure}
[The corresponding values of $w(\alpha)$ which satisfy $G(w(\alpha),\alpha)=0$ with
$w(\pi) = 2\pi$ 
]
{\includegraphics[width=9cm]{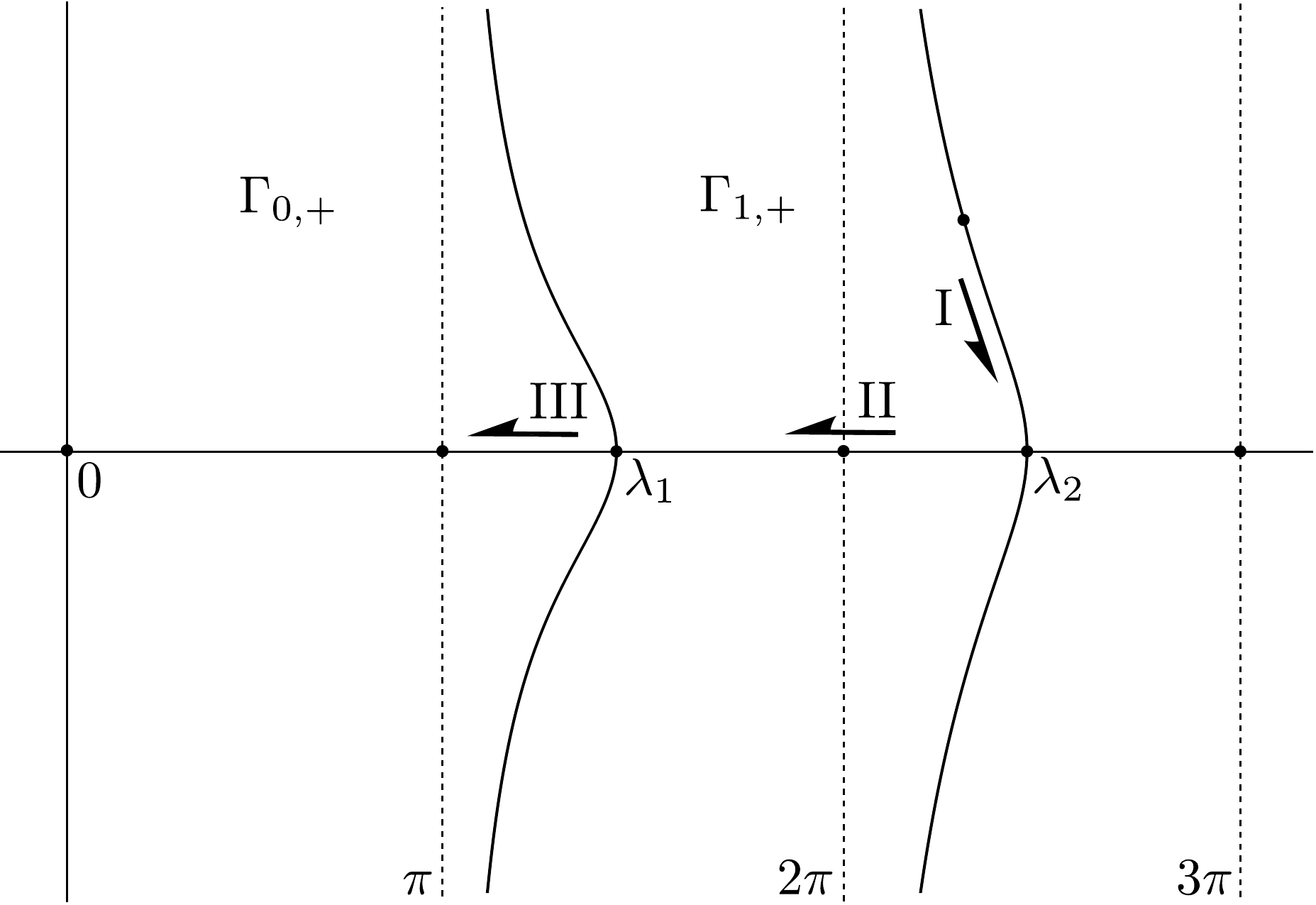} \label{fig:w2fun2b}}
\end{subfigure}
\caption{
The values $\sinc{(\alpha)}$ for $\alpha \in (0,2\pi)$  (\cref{fig:w2fun2a}) and
the corresponding values of $w(\alpha)$ which satisfy $G(w(\alpha),\alpha)=0$ with
$w(\pi) = 2\pi$ (\cref{fig:w2fun2b}).
In~\cref{fig:w2fun2b}, segment I represents $w(\alpha)$ for $\alpha \in (0,\bet{2}{1})$,
segment II represents $w(\alpha)$ for $\alpha \in (\bet{2}{1},\lambda_{1})$, 
and finally segment III represents $w(\alpha)$ for $\alpha \in (\lambda_{1},2\pi)$.}
\label{fig:w2fun2}
\end{center}
\end{figure}

\begin{rem}
Referring to~\cref{fig:w2fun2}, we provide a detailed description for
the behavior of $w(\alpha)$ defined in~\cref{lem:defw7} for 
$\alpha \in (0,2\pi)$.
\end{rem}

We present the principal result of this section in the following lemma.
\begin{lem}
\label{lem:defz6}
Suppose that $H(z,\theta)$ is as defined in~\cref{eq:defh2}. 
Suppose that $\bet{2}{1}$ is as defined in~\cref{lem:defalphj2}.
Furthermore, suppose that $\theta_{1}$ is given by
\beq
\theta_{1} = \frac{\bet{2}{1}}{\pi} \, ,
\eeq
Suppose further that $D$ is the strip in the upper half plane with 
$0 < \text{Re}(\theta) < 2$ that includes the interval
$(0,2) \setminus \{ \theta_{1} \}$, i.e.
\beq
D = \{ \theta \in \C: \, 0 < \text{Re}(\theta) < 2 \, , \quad \text{Im}(\theta) \geq 0 \} 
\setminus \{ \theta_{1} \} \, .
\eeq
Then there exists a simply connected open set $D \subset V \subset \C$ and an
analytic function $z(\theta): V \to \C$ which satisfies $H(z(\theta),\theta) = 0$
for all $\theta \in V$ and $z(1) = 2$.
\end{lem}

Finally, we now present the principal result of~\cref{sec:appatone}.

\begin{thm}
\label{thm:mainsingpowtoneapp}
Suppose that $N\geq 2$ is an integer.
Then there exists $3N-2$ real numbers 
$\theta_{1},\theta_{2},\ldots \theta_{3N-2} \in (0,2)$
such that the following holds.
Suppose that $D$ is the strip in the upper half plane 
with $0<\text{Re}(\theta)<2$ that includes the
interval $(0,2)\setminus \{ \theta_{j} \}_{j=1}^{3N-2}$, i.e.
\beq
\label{eq:defDmaintoneapp}
D = \{ \theta \in \C \, : \, \text{Re}(\theta) \in (0,2) \, , 
\quad 0 \leq \text{Im}(\theta) < \infty \} \setminus \{ \theta_{j} \}_{j=1}^{3N-2} \, . 
\eeq
Then, there exists a simply connected open set $D \subset V \subset \C$
and analytic functions $z_{n,1}(\theta): V\to \C$, $n=1,2\ldots N$, 
which satisfy
\beq
z\sint - \sinztc = 0 \, , \quad z(1) =  n \, , \label{eq:implfun1mainapp}
\eeq
for $\theta \in V$, and analytic functions
$z_{n,2}(\theta):V \to \C$, $n=2,3,\ldots N$, which satisfy
\beq
z\sint - \sinzt = 0 \, , \quad z(1) =  n \, , \label{eq:implfun2mainapp}
\eeq
for $\theta \in V$ (see~\cref{fig:vdomapp} for an illustrative domain $V$).
Moreover, the functions 
$z_{n,1}(\theta)$, $n=1,2\ldots N$, do not take integer values for all
$\theta \in V \setminus \{ 1\}$,  
and satisfy
$\det{\bA(z_{n,1}(\theta),\theta)}=0$, $n=1,2\ldots N$, 
for all $\theta \in V$
(see~\cref{eq:defamat,eq:detval}).
Similarly, the functions
$z_{n,2}(\theta)$, $n=2,3,\ldots N$, do not take integer values
for all $\theta \in V \setminus \{ 1 \}$, and satisfy
$\det{\bA(z_{n,2}(\theta),\theta)}=0$, $n=2,3\ldots N$, 
for all $\theta \in V$
(see~\cref{eq:defamat,eq:detval}).
\end{thm}

\begin{figure}
\begin{center}
\includegraphics[width=8cm]{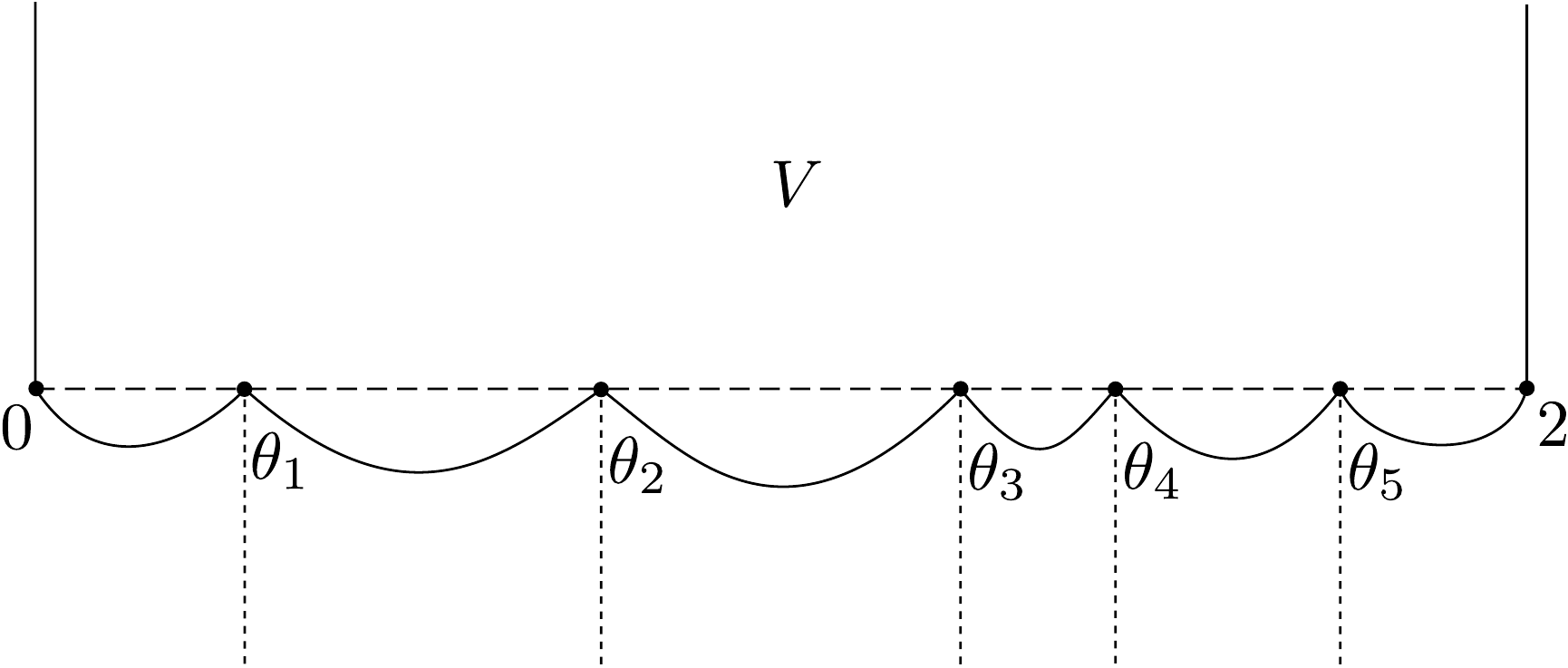}
\caption{An illustrative domain $V$ for the case $N=2$, where 
$\theta_{1}, \theta_{2}, \ldots \theta_{5}$ are the combined 
branch points of the functions 
$z_{1,1}(\theta)$, $z_{2,1}(\theta)$, and $z_{2,2}(\theta)$.
The large dashes are used to denote the interval $(0,2)$ for reference. 
The thin dashes are the locations of the branch cuts at $\theta_{1},
\theta_{2}\ldots \theta_{5}$.}
\label{fig:vdomapp}
\end{center}
\end{figure}

\begin{proof}
Suppose that $m$ is a positive integer.
Let $z_{n,1}(\theta)$, $n=2m$, $n \leq N$, be
the analytic functions defined in~\cref{lem:defz1} and $V_{n,1}$, 
$n=2m$, $n\leq N$ be the regions of analyticity of $z_{n,1}(\theta)$.
Similarly, 
let $z_{n,1}(\theta)$, $n=2m+1$, $n\leq N$, be the analytic functions
defined in~\cref{lem:defz2} and $V_{n,1}$, $n=2m+1$, $n\leq N$ 
be the regions of analyticity of $z_{n,1}(\theta)$.
Let $z_{1,1}(\theta)$ be the analytic function defined in~\cref{lem:defz3}
and $V_{1,1}$ be the region of analyticity of $z_{1,1}(\theta)$. 
Proceeding in a similar manner, let $z_{n,2}(\theta)$, and $V_{n,2}$,
$n=2,3\ldots N$, be the analytic functions and their domain 
of analyticity as defined in~\cref{lem:defz4,lem:defz5,lem:defz6}.
Suppose that $\ajo$ and $\ajt$ are as defined in~\cref{lem:defalphj}, and
$\bjo$ and $\bjt$ are as defined in~\cref{lem:defalphj2}.
Let $\theta_{1}, \theta_{2}, \ldots \theta_{3N-2}$ be the colletion of
numbers
\begin{align}
\left\{ \theta_{j} \right\}_{j=1}^{3N-2} =  
&\left\{ 2- \frac{\al{2m-1}{1}}{\pi} \right\}_{m=1}^{2m-1 \leq N}   \cup 
\left\{ 2- \frac{\al{2m}{1}}{\pi} \right\}_{m=1}^{2m \leq N} \cup
\left\{ 2- \frac{\al{2m}{2}}{\pi} \right\}_{m=1}^{2m \leq N} \cup \nonumber \\
& \left\{ \frac{\bet{2m+1}{1}}{\pi} \right\}_{m=1}^{2m+1 \leq N}  \cup 
\left\{ \frac{\bet{2m+1}{1}}{\pi} \right\}_{m=1}^{2m+1 \leq N} \cup
\left\{ \frac{\bet{2m}{2}}{\pi} \right\}_{m=1}^{2m \leq N}  \, .
\end{align}
Clearly, $V = \cap_{n=1}^{N} V_{n,1} \cap_{n=2}^{N} V_{n,2}$
is a simply connected open neighborhood such that $D \subset V$
and $V$ is the common region of analyticity of the 
functions $z_{n,1}(\theta)$, $n=1,2,\ldots N$, and $z_{n,2}(\theta)$,
$n=2,3,\ldots N$.
Furthermore, it follows 
from~\cref{lem:defz1,lem:defz2,lem:defz3,lem:defz4,lem:defz5,lem:defz6}
that the domain $V$, the analytic functions $z_{n,1}(\theta)$,
$n=1,2,\ldots N$, and the analytic functions $z_{n,2}(\theta)$, $n=2,3,\ldots N$,
satisfy all the conditions of the theorem.
\end{proof}

\subsection{Tangential even, normal odd case \label{sec:appanote}} 
Suppose that $\bA(z,\theta)$ is the $2\times 2$ matrix 
defined in~\cref{eq:defamatnote}. 
We recall that 
\beq
\label{eq:detvalappnote}
\det{\bA(z,\theta)} = 
\frac{\left( z \sint + \sinzt \right) \left( z \sint + \sinztc \right)}{4 
\sqsinz} \, .
\eeq
If $z$ is not an integer, and satisfies either
\beq
z\sint + \sinztc = 0 , \label{eq:implfun1app0note}
\eeq
or
\beq
z \sint + \sinzt = 0 \, ,\label{eq:implfun2app0note}
\eeq
then $\det{(\bA(z,\theta)} = 0$.
The analysis of the implicit functions 
$z(\theta)$ 
which satisfy~\cref{eq:implfun1app0note,eq:implfun2app0note}
for the tangential even, normal odd case
is similar to the analysis of the implicit functions
$z(\theta)$ which satisfy~\cref{eq:implfun1app0,eq:implfun2app0}
for the tangential odd, normal even case.
The principal result of this section is~\cref{thm:mainsingpownoteapp}, 
which is a restatement of~\cref{thm:mainsingpownote}.
For brevity, we omit the proof.

\begin{thm}
\label{thm:mainsingpownoteapp}
Suppose that $N\geq 2$ is an integer.
Then there exists $3N-2$ real numbers 
$\theta_{1},\theta_{2},\ldots \theta_{3N-2} \in (0,2)$
such that the following holds.
Suppose that $D$ is the strip in the upper half plane 
with $0<\text{Re}(\theta)<2$ that includes the
interval $(0,2)\setminus \{ \theta_{j} \}_{j=1}^{3N-2}$, i.e.
\beq
\label{eq:defDmainnoteapp}
D = \{ \theta \in \C \, : \, \text{Re}(\theta) \in (0,2) \, , 
\quad 0 \leq \text{Im}(\theta) < \infty \} \setminus \{ \theta_{j} \}_{j=1}^{3N-2} \, . 
\eeq
Then, there exists a simply connected open set $D \subset V \subset \C$
and analytic functions $z_{n,1}(\theta): V\to \C$, $n=2,3\ldots N$, 
which satisfy
\beq
z\sint + \sinztc = 0 \, , \quad z(1) =  n \, , \label{eq:implfun3mainapp}
\eeq
for $\theta \in V$, and analytic functions
$z_{n,2}(\theta):V \to \C$, $n=1,2,\ldots N$, which satisfy
\beq
z\sint + \sinzt = 0 \, , \quad z(1) =  n \, , \label{eq:implfun4mainapp}
\eeq
for $\theta \in V$ (see~\cref{fig:vdomapp} for an illustrative domain $V$).
Moreover, the functions 
$z_{n,1}(\theta)$, $n=2,3\ldots N$, do not take integer values for all
$\theta \in V \setminus \{ 1\}$,  
and satisfy
$\det{\bA(z_{n,1}(\theta),\theta)}=0$, $n=2,3\ldots N$, 
for all $\theta \in V$
(see~\cref{eq:defamatnote,eq:detvalnote}).
Similarly, the functions
$z_{n,2}(\theta)$, $n=1,2,\ldots N$, do not take integer values
for all $\theta \in V \setminus \{ 1 \}$, and satisfy
$\det{\bA(z_{n,2}(\theta),\theta)}=0$, $n=1,2\ldots N$, 
for all $\theta \in V$
(see~\cref{eq:defamatnote,eq:detvalnote}).
\end{thm}

%%%%%%%%%%%%%%%%%%%%%%%%%%%%%%%%%%%%%%%%%%%%%%%%%%
\section{Appendix B \label{sec:appb}}
In this section we compute the limit $\theta \to 1$, of the linear transformation
$\bB(\theta)$, which maps coefficients of singular basis functions 
for the solution of the integral equation to the Taylor expansion coefficients
of the velocity field.
In~\cref{sec:appbtone}, we investigate the tangential odd, normal even case
(see~\cref{eq:tone})
and, in~\cref{sec:appbnote}, we investigate the tangential even,
normal odd case (see~\cref{eq:teno}).

\subsection{Tangential odd, normal even case \label{sec:appbtone}}
Suppose that $\bA(z,\theta)$ is the $2\times 2$ matrix given by~\cref{eq:defamat}
defined in~\cref{sec:tone}.
Suppose $N$ is a positive integer.
Suppose further that, as in~\cref{thm:mainsingpowtone},
$z_{n,1}(\theta)$, $n=1,2,\ldots N$, are analytic functions
satisfying $\det{\bA(z_{n,1}(\theta),\theta)}=0$
for $\theta \in V_{\delta} \subset \C$, where 
$V_{\delta}$ is a neighborhood of the contour $C_{\delta}$.
Similarly, suppose that $z_{n,2}(\theta)$, $n=2,3\ldots N$, are
analytic functions satisfying $\det{\bA(z_{n,2}(\theta),\theta)}=0$ 
for $\theta \in V_{\delta}$. 
Let $(p_{n,1}, q_{n,1}) \in \cN \{ \bA(z_{n,1}(\theta),\theta) \} $,
$n=1,2,\ldots N$, and 
$(p_{n,2},q_{n,2}) \in \cN \{ \bA(z_{n,2}(\theta),\theta) \}$,
$n=2,3,\ldots N$.
We further assume that the vectors $(p_{n,1},q_{n,1})$, $n=1,2,\ldots $ 
and $(p_{n,2},q_{n,2})$, $n=2,3,\ldots$, are $\ell^{2}$ normalized.
Suppose that $z_{1,2}(\theta) \equiv 1$, $p_{1,2} = 0$,
and $q_{1,2} = 1$.
Finally, suppose that
$\bB(\theta)$ is a $(2N+2) \times (2N+2)$ matrix defined in~\cref{lem:multpowexp}
for $\theta \in V_{\delta}$.
The $2\times 2$ blocks of $\bB(\theta)$ 
are given by
\beq
\label{eq:bmatdefapptoneln}
\bB_{\ell,n} (\theta)  = 
\left[
\begin{array}{c;{2pt/2pt}c}
\bF(\ell,z_{n,1}(\theta),\theta)
\bbmat
p_{n,1}(\theta) \\
q_{n,1}(\theta)
\ebmat &
\bF(\ell,z_{n,2}(\theta),\theta)
\bbmat
p_{n,2}(\theta) \\
q_{n,2}(\theta)
\ebmat
\end{array}
\right] \, ,
\eeq
for $\ell,n=1,2,\ldots N$, where $\bF$ is defined in~\cref{eq:kerssmoothmat}, 
except for the case $\ell=n=1$.
In the case $\ell=n=1$, the matrix $\bB_{1,1}(\theta)$ is given by
\beq
\label{eq:bmatdefapptone11}
\bB_{1,1}(\theta) = 
\left[
\begin{array}{c;{2pt/2pt}c}
\bF(1,z_{1,1}(\theta),\theta)
\bbmat
p_{1,1}(\theta) \\
q_{1,1}(\theta)
\ebmat &
\bF_{1}(\theta)
\end{array}
\right] \, ,
\eeq
where $\bF$ is defined in~\cref{eq:kerssmoothmat}, and 
$\bF_{1}$ is defined in~\cref{eq:a11def}.
Finally, if either $\ell=0$ or $n=0$, then the matrices $\bB_{\ell,n}(\theta)$
are given by
\begin{align}
\label{eq:bmatdefapptonen0}
\bB_{\ell,0}(\theta) &= 
\bbmat
0 & 0 \\
0 & 0
\ebmat \, , \\
\label{eq:bmatdefapptone00}
\bB_{0,0}(\theta) &= \bF_{0}(\theta) \, , \\
\label{eq:bmatdefapptone0n}
\bB_{0,n}(\theta) &=
\bF(n,0,\theta) \, , 
\end{align}
for $\ell,n=1,2,\ldots N$, 
where $\bF$ is defined in~\cref{eq:kerssmoothmat}, 
and $\bF_{0}$ is defined in~\cref{eq:a00def}. 

In the following lemma, we describe the behavior of 
$z_{n,1}(\theta),p_{n,1}(\theta)$, and $q_{n,1}(\theta)$, $n=1,2,\ldots$, 
in the vicinity of $\theta=1$.

\begin{lem}
\label{lem:pqzclosearm}
Suppose that $z_{n,1}(\theta)$, $n=1,2,\ldots$, satisfying 
$\det{\bA(z_{n,1}(\theta),\theta)} = 0$, be as defined in~\cref{thm:mainsingpowtone}. 
Suppose further that 
$(p_{n,1}(\theta),q_{n,1}(\theta))$, $n=1,2,\ldots$, be an $\ell^2$ normalized
null vector
of $\bA(z_{n,1},\theta)$.
Then in the neighborhood of $\theta=1$, 
\begin{align}
z_{2n,1}(\theta)  &= 2n - \frac{\pi^2 n(4n^2 - 1)}{3} (\theta-1)^3 + 
O(|\theta-1|^4) \\ 
p_{2n,1}(\theta) &= -\frac{\pi(2n-1)}{2}(\theta-1) + O(|\theta-1|^3) \\
q_{2n,1}(\theta) &= 1 + O(|\theta-1|^2) \\
z_{2n+1,1}(\theta)  &= 2n + 1 + 2(2n+1) (\theta-1) + O(|\theta-1|^2) \\ 
p_{2n+1,1}(\theta) &= 1 + O(|\theta-1|) \\
q_{2n+1,1}(\theta) &= O(|\theta-1|) \, . 
\end{align}
\end{lem}

\begin{proof}
Let the superscript $'$ denote derivative with
respect to $\theta$ and $\pd_{t}$ denote the partial derivative with respect
to $t$.
Suppose that $H(z,\theta)$ is given by
\beq
H(z,\theta) = z\sint - \sinztc
\eeq
From~\cref{thm:mainsingpowtone}, we recall that 
$z_{n,1}(\theta)$, $n=1,2,\ldots$, satisfy $H(z_{n,1},\theta)=0$. 
Using the implicit function theorem
\beq
z'(\theta)
= -\frac{\pd_{\theta} H}{\pd_{z} H}
= z(1-\sec{(\pi z)})
\eeq
Thus, 
\beq 
z'_{2n,1}(1) = 0 
\, , \quad \text{and} \quad 
z'_{2n+1,1}(1) = 2(2n+1)\, . 
\eeq
Implicitly differentiating $H$ twice we get
\beq
\pd_{z} H \cdot z''(\theta) + \left(\ders{}{\theta} \pd_{z} H + 
\pd_{z}\pd_{\theta} H \right) \cdot z'(\theta) + \pd_{\theta} 
\pd_{\theta} H = 0 \, .
\eeq
Since $z_{2n,1}'(1) = 0$, we get
\beq
z_{2n,1}''(1) = -\frac{\pd_{\theta \theta} H }{\pd_{z} H} =
\pi z^{2} \tan{(\pi z)} = 0 \, .
\eeq
Proceeding in a similar fashion, we observe that
\beq
z_{2n,1}'''(1) = -\frac{\pd_{\theta \theta \theta} H }{\pd_{z} H} =
\pi^2 z_{2n,1}(\sec{(\pi z_{2n,1})} - z_{2n,1}^2) = 2\pi^2 n (1-4n^2)   \, .
\eeq
Let
\begin{align} 
\tilde{p}_{n,1}(\theta) &= -\sinz + 2 a_{2,2}(z_{n,1}(\theta),\theta) \, , \\
\tilde{q}_{n,1}(\theta) &= 2 \sinz a_{2,1}(z_{n,1}(\theta),\theta) \, ,
\end{align}
where $a_{2,1}$ and $a_{2,2}$ are defined in~\cref{eq:defsing21}
and~\cref{eq:defsing22} respectively.
Clearly, $(\tilde{p}_{n,1},\tilde{q}_{n,1}) \in \cN(\bA)$.
We then set
\beq
p_{n,1}(\theta) = \frac{\tilde{p}_{n,1}}{\sqrt{(\tilde{p}_{n,1}^2 + 
\tilde{q}_{n,1}^2)}} 
\, , \quad \text{and} \quad
q_{n,1}(\theta) = \frac{\tilde{q}_{n,1}}{\sqrt{(\tilde{p}_{n,1}^2 + 
\tilde{q}_{n,1}^2)}} 
\eeq
The required Taylor expansions for $p_{n,1},q_{n,1}$ are then readily obtained
by using the Taylor expansions of $z_{n,1}(\theta)$.
\end{proof}

In the next lemma, we now describe the behavior of 
$z_{n,2}(\theta),p_{n,2}(\theta)$ and $q_{n,2}(\theta)$, $n=2,3,\ldots$, in the vicinity of $\theta=1$.

\begin{lem}
\label{lem:pqzouterarm}
Suppose that $z_{n,2}(\theta)$, $n=2,3,\ldots$, satisfying 
$\det{\bA(z_{n,2}(\theta),\theta)} = 0$, be as defined in~\cref{thm:mainsingpowtone}. 
Suppose further that 
$(p_{n,2}(\theta),q_{n,2}(\theta))$, $n=2,3,\ldots$, be an $\ell^2$ normalized
null vector
of $\bA(z_{n,2},\theta)$.
Then in the neighborhood of $\theta=1$, 
\begin{align}
z_{2n,2}(\theta)  &= 2n - 4n (\theta-1) + O(|\theta-1|^2) \\ 
p_{2n,2}(\theta) &= 1 + O(|\theta-1|) \\
q_{2n,2}(\theta) &= 0 + O(|\theta-1|) \\
z_{2n+1,1}(\theta)  &= 2n+1 + \frac{2\pi^2 n(n+1)(2n+1)}{3} (\theta-1)^3 + 
O(|\theta-1|^4) \\ 
p_{2n+1,2}(\theta) &= -\pi n(\theta-1) + O(|\theta-1|^3) \\
q_{2n+1,2}(\theta) &= 1 + O(|\theta-1|^2) 
\end{align}
\end{lem}

\begin{proof}
The proof proceeds in a similar manner as the  proof of~\cref{lem:pqzclosearm}. 
\end{proof}

Combining~\cref{lem:pqzclosearm,lem:pqzouterarm}, we 
present the principal result of this section, which computes the limit
$\theta \to 1$ 
of the matrix $\bB(\theta)$ in the following theorem.
\begin{thm}
\label{thm:blimapptone}
Suppose that $N$ is a positive integer and suppose further that
$\bB$ is given 
by~\cref{eq:bmatdefapptoneln,eq:bmatdefapptone11,eq:bmatdefapptonen0,eq:bmatdefapptone00,eq:bmatdefapptone0n}. 
Then
\beq
\lim_{\theta \to 1} \bB_{\ell,j} (\theta) = 
\begin{cases}
\bbmat
-1/2 & 0 \\
0 & -1/2 
\ebmat & \quad  \ell=j=0  \vspace*{3pt}\\
\bbmat
0 & -1/2 \\
-1/2 & 0 
\ebmat & \quad  \ell=j=2m \neq 0  \vspace*{3pt}\\
\bbmat
-1/2 & 0 \\
0 & -1/2 
\ebmat & \quad  \ell=j=2m+1 \vspace*{3pt}\\
\bbmat
0 & 0 \\
0 & 0
\ebmat
&\quad \text{otherwise}
\end{cases} \, ,
\label{eq:matatoneapp}
\eeq
for all $\ell,j=0,1,\ldots N$.
\end{thm}

\begin{proof}
Let $\tilde{\bF}(n,\theta) = 2\pi \bF(n,z,\theta) 
\cdot (n-z)$
for $j,\ell =1,2$ where $\bF$ is given by~\cref{eq:kerssmoothmat}.
On inspecting the entries of $\tilde{\bF}(n,\theta)$, we observe
that
\beq
\tilde{\bF}(n,1) = 
\bbmat
0 & 0 \\
0 & 0
\ebmat \, ,
\eeq
for all $n \in \N$. 
Since $z_{j,\ell}(1) = j$, we conclude that
\beq
\lim_{\theta \to 1} 
\bF(n,z_{j,\ell},\theta) 
\bbmat
p_{j,\ell} \\
q_{j,\ell} 
\ebmat
= 
\lim_{\theta \to 1} 
\frac{\tilde{\bF}(n,\theta)}{(-z_{j,\ell}+n)} 
\bbmat
p_{j,\ell} \\
q_{j,\ell} 
\ebmat
=
\bbmat
0 \\
0
\ebmat \, ,
\eeq
for all $j \neq n$, and $\ell=1,2$.
$\bB_{n,0}(\theta)=0$ is the zero matrix by definition and for
$\bB_{0,n}(\theta)$, we have 
\beq
\lim_{\theta \to 1} \bB_{0,n}(\theta) 
= \lim_{\theta \to 1} \bF(n,0,\theta)
= \lim_{\theta \to 1} \frac{\tilde{\bF}(n,\theta)}{2\pi n}
= \bbmat
0 & 0 \\
0 & 0
\ebmat \, , 
\eeq
for $n=1,2,\ldots N$.

We now turn our attention to the diagonal terms.
For $n=0$, it follows from a simple calculation that 
\beq
\lim_{\theta \to 1} \bB_{0,0}(\theta) = 
\lim_{\theta \to 1} \bF_{0}(\theta) =
\bbmat
-1/2 & 0 \\
0 & -1/2
\ebmat \, .
\eeq

For $n \geq 2$, using~\cref{lem:pqzclosearm} and~\cref{lem:pqzouterarm}, it follows from a
rather tedious calculation that 
\begin{align}
\label{eq:blim2n1}
\tilde{\bF}(2n,\theta)  
\bbmat
p_{2n,1}(\theta) \\
q_{2n,1}(\theta) 
\ebmat 
&= \bbmat
 O(|\theta-1|^{4}) \\
-\frac{\pi^3 n (4n^2 - 1)}{3} (\theta - 1)^3 + O(|\theta-1|^4)
\ebmat \\
\label{eq:blim2np11}
\tilde{\bF}(2n+1,\theta)
\bbmat
p_{2n+1,1}(\theta) \\
q_{2n+1,1}(\theta)
\ebmat
&=
\bbmat
(2n+1)\pi(\theta-1) + O(|\theta-1|^2) \\
O(|\theta-1|^2)
\ebmat \, ,
\end{align}
and
\begin{align}
\label{eq:blim2n2}
\tilde{\bF}(2n,\theta)  
\bbmat
p_{2n,2}(\theta) \\
q_{2n,2}(\theta) 
\ebmat 
&= \bbmat
- 4 n \pi (\theta - 1) + O(|\theta-1|^2) \\
O(|\theta-1|^2)
\ebmat \\
\label{eq:blim2np12}
\tilde{\bF}(2n+1,\theta)
\bbmat
p_{2n+1,2}(\theta) \\
q_{2n+1,2}(\theta)
\ebmat
&=
\bbmat
 O(|\theta-1|^{4}) \\
\frac{2 \pi^3 n (n+1)(2n+1)}{3} (\theta - 1)^3 + O(|\theta-1|^4)
\ebmat \, .
\end{align}
Finally, from the definition of $\bF_{1}$ in~\cref{eq:a11def}, we note that 
\beq
\label{eq:f11lim}
\lim_{\theta \to 1} \bF_{1}(\theta) = 
\bbmat
0 \\
-1/2
\ebmat \, .
\eeq
The result then follows from combining~\cref{eq:blim2n1,eq:blim2np11,eq:blim2n2,eq:blim2np12,eq:f11lim}.
\end{proof}

\subsection{Tangential even, normal odd case \label{sec:appbnote}}
Suppose that $\bA(z,\theta)$ is the $2\times 2$ matrix given by~\cref{eq:defamatnote}
defined in~\cref{sec:teno}.
Suppose $N$ is a positive integer.
Suppose further that, as in~\cref{thm:mainsingpownote},
$z_{n,1}(\theta)$, $n=2,3,\ldots N$, are analytic functions
satisfying $\det{\bA(z_{n,1}(\theta),\theta)}=0$
for $\theta \in V_{\delta} \subset \C$, where 
$V_{\delta}$ is a neighborhood of the contour $C_{\delta}$.
Similarly, suppose that $z_{n,2}(\theta)$, $n=1,2\ldots N$, are
analytic functions satisfying $\det{\bA(z_{n,2}(\theta),\theta)}=0$ 
for $\theta \in V_{\delta}$. 
Let $(p_{n,1}, q_{n,1}) \in \cN \{ \bA(z_{n,1}(\theta),\theta) \} $,
$n=2,3,\ldots N$, and 
$(p_{n,2},q_{n,2}) \in \cN \{ \bA(z_{n,2}(\theta),\theta) \}$,
$n=1,2,\ldots N$.
We further assume that the vectors $(p_{n,1},q_{n,1})$, $n=2,3,\ldots $ 
and $(p_{n,2},q_{n,2})$, $n=1,2,\ldots$, are $\ell^{2}$ normalized.
Suppose that $z_{1,2}(\theta) \equiv 1$, $p_{1,2} = 0$,
and $q_{1,2} = 1$.
Finally, suppose that
$\bB(\theta)$ is a $(2N+2) \times (2N+2)$ matrix defined in~\cref{lem:multpowexpnote}
for $\theta \in V_{\delta}$.
The $2\times 2$ blocks of $\bB(\theta)$ 
are given by
\beq
\label{eq:bmatdefappnoteln}
\bB_{\ell,n} (\theta)  = 
-\left[
\begin{array}{c;{2pt/2pt}c}
\bF(\ell,z_{n,1}(\theta),\theta)
\bbmat
p_{n,1}(\theta) \\
q_{n,1}(\theta)
\ebmat &
\bF(\ell,z_{n,2}(\theta),\theta)
\bbmat
p_{n,2}(\theta) \\
q_{n,2}(\theta)
\ebmat
\end{array}
\right] \, ,
\eeq
for $\ell,n=1,2,\ldots N$, where $\bF$ is defined in~\cref{eq:kerssmoothmat}, 
except for the case $\ell=n=1$.
In the case $\ell=n=1$, the matrix $\bB_{1,1}(\theta)$ is given by
\beq
\label{eq:bmatdefappnote11}
\bB_{1,1}(\theta) = 
\left[
\begin{array}{c;{2pt/2pt}c}
\bF_{1}(\theta) &
-\bF(1,z_{1,2}(\theta),\theta)
\bbmat
p_{1,2}(\theta) \\
q_{1,2}(\theta)
\ebmat 
\end{array}
\right] \, ,
\eeq
where $\bF$ is defined in~\cref{eq:kerssmoothmat}, and 
$\bF_{1}$ is defined in~\cref{eq:a11defnote}.
Finally, if either $\ell=0$ or $n=0$, then the matrices $\bB_{\ell,n}(\theta)$
are given by
\begin{align}
\label{eq:bmatdefappnoten0}
\bB_{\ell,0}(\theta) &= 
\bbmat
0 & 0 \\
0 & 0
\ebmat \, , \\
\label{eq:bmatdefappnote00}
\bB_{0,0}(\theta) &= \bF_{0}(\theta) \, , \\
\label{eq:bmatdefappnote0n}
\bB_{0,n}(\theta) &=
-\bF(n,0,\theta) \, , 
\end{align}
for $\ell,n=1,2,\ldots N$, 
where $\bF$ is defined in~\cref{eq:kerssmoothmat}, 
and $\bF_{0}$ is defined in~\cref{eq:a00defnote}. 
The proofs of the results in this section are similar to the corresponding proofs
in~\cref{sec:appbtone}.
For conciseness, we state the results without proof. 
The principal result of this section is~\cref{thm:blimappnote}.

In the following lemma, we describe the behavior of 
$z_{n,1}(\theta),p_{n,1}(\theta)$, and $q_{n,1}(\theta)$, $n=2,3,\ldots$, 
in the vicinity of $\theta=1$.

\begin{lem}
\label{lem:blimclosearmnoteapp}
Suppose that $z_{n,1}(\theta)$, $n=2,3,\ldots$, satisfying 
$\det{\bA(z_{n,1}(\theta),\theta)} = 0$, be as defined in~\cref{thm:mainsingpownote}. 
Suppose further that 
$(p_{n,1}(\theta),q_{n,1}(\theta))$, $n=2,3,\ldots$, be an $\ell^2$ normalized
null vector
of $\bA(z_{n,1},\theta)$.
\begin{align}
z_{2n,1}(\theta)  &= 2n  + 4n (\theta-1) + O(|\theta-1|^2) \\ 
p_{2n,1}(\theta) &= 1 + O(|\theta-1|) \\
q_{2n,1}(\theta) &= 0 + O(|\theta-1|) \\ 
z_{2n+1,1}(\theta)  &= 2n+1 - \frac{2\pi^2 n(n+1)(2n+1)}{3} (\theta-1)^3 + 
O(|\theta-1|^4) \\ 
p_{2n+1,1}(\theta) &= -\pi n(\theta-1) + O(|\theta-1|^3) \\
q_{2n+1,1}(\theta) &= 1 + O(|\theta-1|^2) 
\end{align}
\end{lem}

Similarly, in the following lemma, we describe the behavior of 
$z_{n,2}(\theta),p_{n,2}(\theta)$, and $q_{n,2}(\theta)$, $n=1,2,\ldots$, 
in the vicinity of $\theta=1$.
\begin{lem}
\label{lem:blimfurtherarmnoteapp}
Suppose that $z_{n,2}(\theta)$, $n=1,2,\ldots$, satisfying 
$\det{\bA(z_{n,2}(\theta),\theta)} = 0$, be as defined in~\cref{thm:mainsingpownote}. 
Suppose further that 
$(p_{n,2}(\theta),q_{n,2}(\theta))$, $n=1,2,\ldots$, be an $\ell^2$ normalized
null vector
of $\bA(z_{n,2},\theta)$.
Then in the neighborhood of $\theta=1$, 
\begin{align}
z_{2n,1}(\theta)  &= 2n + \frac{\pi^2 n(4n^2-1)}{3} (\theta-1)^3 + 
O(|\theta-1|^4) \\ 
p_{2n,2}(\theta) &= -\frac{\pi (2n-1)}{2}(\theta-1) + O(|\theta-1|^3) \\
q_{2n,2}(\theta) &= 1 + O(|\theta-1|^2) \\
z_{2n+1,2}(\theta)  &= 2n+1 - 2(2n+1) (\theta-1) + O(|\theta-1|^2) \\ 
p_{2n+1,2}(\theta) &= 1 + O(|\theta-1|) \\
q_{2n+1,2}(\theta) &= 0 + O(|\theta-1|) 
\end{align}
\end{lem}

Finally, we present the principal result of this section 
in the following lemma.
\begin{thm}
\label{thm:blimappnote}
Suppose that $N$ is a positive integer and suppose further that
$\bB$ is given 
by~\cref{eq:bmatdefappnoteln,eq:bmatdefappnote11,eq:bmatdefappnoten0,eq:bmatdefappnote00,eq:bmatdefappnote0n}. 
Then
\beq
\lim_{\theta \to 1} \bB_{\ell,j} (\theta) = 
\begin{cases}
\bbmat
-1/2 & 0 \\
0 & -1/2 
\ebmat & \quad  \ell=j=2m \vspace*{3pt}\\
\bbmat
0 & -1/2 \\
-1/2 & 0 
\ebmat & \quad  \ell=j=2m+1 \vspace*{3pt}\\
\bbmat
0 & 0 \\
0 & 0
\ebmat
&\quad \text{otherwise}
\end{cases} \, ,
\label{eq:matanoteapp}
\eeq
for all $\ell,j=0,1,2,\ldots N$.
\end{thm}

%%%%%%%%%%%%%%%%%%%%%%%%%%%%%%%%%%%%%%%%%%%%%%%%%%%%%%%%%%%%%%%%%%%%%%
%
% BIBLIOGRAPHY
%
%%%%%%%%%%%%%%%%%%%%%%%%%%%%%%%%%%%%%%%%%%%%%%%%%%%%%%%%%%%%%%%%%%%%%%

% load database
\bibliography{biblio}

\bibliographystyle{ieeetr}
%\bibliographystyle{elsarticle-num}

% include all articles in the database (not just those cited)

%%%%%%%%%%%%%%%%%%%%%%%%%%%%%%%%%%%%%%%%%%%%%%%%%%%%%%%%%%%%%%%%%%%%%%
%
% END BIBLIOGRAPHY
%
%%%%%%%%%%%%%%%%%%%%%%%%%%%%%%%%%%%%%%%%%%%%%%%%%%%%%%%%%%%%%%%%%%%%%%

\end{document}